\documentclass[letterpaper,11pt]{article}
\usepackage[margin=1in]{geometry} 

\usepackage{enumitem}
\usepackage[cmex10]{amsmath}
\usepackage{amssymb}
\usepackage{authblk}
\usepackage{amsmath}%
\usepackage{amsfonts}%
\usepackage{amssymb}%
\usepackage{amsthm}
\usepackage{breqn}
\usepackage{graphicx}

\usepackage[justification=centering]{caption}
\usepackage{float}

\usepackage{colonequals}
\usepackage{psfrag}
\usepackage{subfigure}
\usepackage{mathrsfs}
\usepackage[usenames,dvipsnames]{xcolor}
\usepackage{hyperref}

\hypersetup{%
  colorlinks=true,
  urlcolor=blue,%
  linkcolor=blue,%
  citecolor=blue,%
  linkbordercolor=blue,
  pdfborderstyle={/S/U/W 1}
}

\usepackage{booktabs}       

\usepackage{mdframed}
\usepackage{bm}
\usepackage{bbm}
\usepackage{cite}

\usepackage{multirow}

\newtheorem{thm}{Theorem}
\newtheorem{lem}{Lemma}

\newtheorem{prop}{Proposition}

\newtheorem{cor}{Corollary}

\theoremstyle{definition}

\newtheorem{defn}{Definition}

\newtheorem{rem}{Remark}
\newtheorem{example}{Example}

\newcommand{\etal}{\textit{et al.}~}

\newcommand{\calX}{\mathcal{X}}
\newcommand{\calA}{\mathcal{A}}

\newcommand{\calZ}{\mathcal{Z}}

\newcommand{\calU}{\mathcal{U}}

\newcommand{\calL}{\mathcal{L}}

\renewcommand{\tilde}{\widetilde}

\newcommand{\Reals}{\mathbb{R}}

\newcommand{\defined}{\triangleq}

\newcommand{\ExpVal}[2]{\mathbb{E}\left[ #2 \right]}

\newcommand{\sto}{\mbox{\normalfont s.t.}}

\newcommand{\EE}[1]{\ExpVal{}{#1}}

\newcommand{\indep}{\rotatebox[origin=c]{90}{$\models$}}

\newcommand{\textfn}[1]{{\small\textit{#1}}}

\newcommand{\rev}[1]{#1}

\providecommand{\keywords}[1]
{
  \small	
  \textbf{\textit{Index Terms---}} #1
}
\makeatletter
\newcommand{\printfnsymbol}[1]{%
  \textsuperscript{\@fnsymbol{#1}}%
}

\title{On the Robustness of Information-Theoretic Privacy Measures and Mechanisms}
\author{Mario~Diaz\thanks{Equal contribution.} , Hao~Wang\printfnsymbol{1}, Flavio~P.~Calmon, Lalitha~Sankar\thanks{This work was supported in part by the National Science Foundation under Grant CCF-1845852, CCF-1350914, CIF-1815361, CIF-1901243 and in part by a seed grant towards a Center for Data Privacy from Arizona State University.\\
Some parts of this paper were presented at the 2018 IEEE International Symposium on Information Theory (ISIT)~\cite{wang2018utility}.\\
M.~Diaz is with Centro de Investigaci\'{o}n en Matem\'{a}ticas, A.C. (email: diaztorres@cimat.mx). A substantial part of this work was developed while being a postdoctoral scholar at Arizona State University and Harvard University.\\
H.~Wang and F.~P.~Calmon are with Harvard University (e-mail: hao\_wang@g.harvard.edu; flavio@seas.harvard.edu). \\
L.~Sankar is with Arizona State University (e-mail: lsankar@asu.edu).
}
}

\date{}

\begin{document}
\maketitle
\vspace{-.2in}

\begin{abstract}
Consider a data publishing setting for a dataset composed by both private and non-private features. The publisher uses an empirical distribution, estimated from $n$ i.i.d.~samples, to design a privacy mechanism which is applied to new fresh samples afterward. In this paper, we study the discrepancy between the privacy-utility guarantees for the empirical distribution, used to design the privacy mechanism, and those for the true distribution, experienced by the privacy mechanism in practice. We first show that, for any privacy mechanism, these discrepancies vanish at speed $O(1/\sqrt{n})$ with high probability. These bounds follow from our main technical results regarding the Lipschitz continuity of the considered information leakage measures. Then we prove that the optimal privacy mechanisms for the empirical distribution approach the corresponding mechanisms for the true distribution as the sample size $n$ increases, thereby establishing the statistical consistency of the optimal privacy mechanisms. Finally, we introduce and study \textit{uniform privacy mechanisms} which, by construction, provide privacy to all the distributions within a neighborhood of the estimated distribution and, thereby, guarantee privacy for the true distribution with high probability.
\end{abstract}

\newpage
\section{Introduction}

Publishing of individual-level data is increasingly pervasive \cite{baker2016there,wicherts2012publish,donoho201750}. In many settings, directly disclosing unmodified data may lead to a privacy risk through unwanted inferences, particularly if the data contains (or is correlated with) private information.
For example, suppose that a hospital wishes to share patient-level data to further advance research on a certain disease. As shown by Sweeney \cite{sweeney2015only}, simply removing private features from the dataset (e.g., name) does not necessarily guarantee privacy as these features usually correlate with the remaining features in the dataset (e.g., ZIP code and gender). This modern data publishing issue, also present in politics \cite{granville2018facebook}, business \cite{narayanan2008robust}, and welfare \cite{handler1966privacy} to name a few, requires 
data disclosure methodologies with provable privacy guarantees, whilst still enabling a desired level of utility to be derived from the disclosed data.

A common approach to ensure privacy consists of perturbing the data using a \textit{privacy mechanism}: a (possibly) randomized mapping designed to control private information leakage, the nature of which depends on how privacy is quantified. Different notions of privacy leakage have been proposed to capture the ability of different adversaries to learn (estimate) private information, e.g., Shannon's mutual information \cite{rebollo2010t,sankar2013utility}, $k$-anonymity \cite{sweeney2002k}, differential privacy \cite{dwork2011differential}, maximal leakage \cite{issa2016maximal}, total variation \cite{rassouli2019optimal}, among others. In addition, privacy mechanisms should also guarantee the statistical utility of the disclosed data, usually quantified by some measure of similarity between the original and the perturbed data, e.g., distortion \cite{yamamoto1983source}, $f$-divergence \cite{kairouz2014extremal}, minimum mean-squared error \cite{asoodeh2017estimation}. In general, privacy and utility objectives compete with each other, making the design of privacy mechanisms a non-trivial task. When data privacy and statistical utility are measured using information-theoretic quantities (e.g., mutual information), most methods for the analysis and design of privacy mechanisms rely on the implicit assumption that the data distribution is, for the most part, known \cite{du2012privacy,sankar2013utility,calmon2015fundamental,asoodeh2016information,basciftci2016privacy,wang2017estimation,osia2019privacy,rassouli2019data,rassouli2019optimal,nageswaran2019data}. In practice, the designer has access only to a sample from the true distribution.

In this work we consider the following setup. We assume that data has both private and non-private features and, based on a sample of such pairs, the designer creates a randomized mapping called a privacy mechanism. As new samples arrive, the designed mechanism is applied to the non-private features in order to produce a sanitized version of them which is later disclosed. In this context, we assume that the adversary (data analyst) knows the true distribution of the data, the privacy mechanism being used, and the disclosed data set. The adversary's objective is then to illegitimately infer the private features associated to the disclosed data. We illustrate this procedure in Figure~\ref{fig:rob_framework}. Observe that this adversary is the strongest in terms of statistical knowledge and hence serves as a worst case benchmark. Apart from its statistical knowledge, we assume that the adversary has no side information regarding the disclosed data.

Since privacy mechanisms are designed using a sample, their privacy-utility guarantees might not generalize to the true distribution, thereby creating a privacy threat as the \textit{de facto} guarantees might be considerably different from those in the design phase. In this work, we study the effect of the discrepancy between the empirical and true distributions on the analysis and design of privacy mechanisms.

First, we consider the setting where the privacy mechanism is designed using an estimate of the data distribution to evaluate the privacy-utility guarantees. We derive bounds for the discrepancy between the privacy-utility guarantees for the empirical distribution (type) and the true (unknown) data distribution. In this context, these bounds can be used to asses the \textit{de facto} guarantees of privacy mechanisms when deployed in practice.

Next, we investigate the \emph{statistical consistency} of the optimal privacy mechanisms\footnote{Loosely speaking, a privacy mechanism is said to be optimal if it delivers the best possible utility constrained to a prescribed amount of privacy leakage. Precise definitions are given in Section~\ref{Subsection:ProblemSetup}.}. More specifically, assume that the privacy mechanism designer constructs an optimal privacy mechanism for the empirical distribution. As the number of samples increases, the optimal privacy mechanism designed for the empirical distribution naturally changes. We show that if such a sequence of privacy mechanisms converges, its limit is an optimal privacy mechanism for the true data distribution.

Finally, we introduce the notion of {\it uniform} privacy mechanism. When privacy is a priority, the privacy mechanism designer may be required to guarantee a specific level of privacy for the true distribution despite having access only to an estimate of it. Motivated by this setting, we consider privacy mechanisms that, by design, assure privacy for {\it every} distribution within a specific neighborhood of the empirical distribution. In this case, large deviations results imply that privacy is guaranteed for the true distribution with a certain probability (depending on the neighborhood). Since privacy is guaranteed {\it uniformly} on a neighborhood of the empirical distribution, we name these privacy mechanisms as \emph{uniform privacy mechanisms}.

The paper is organized as follows. In the remainder of this section we present our main contributions, review related work, and introduce notation used in this paper. In Section~\ref{sec:pre_setup} we review preliminaries, including information leakage measures and large deviations results, and recall the framework of privacy-utility trade-offs. Our main results for the discrepancy of privacy-utility guarantees, convergence properties, and uniform privacy mechanisms are presented in Sections~\ref{Section:PointwiseRobustness}, \ref{Section:convergence}, and \ref{Section:UniformRobustness}, respectively. Finally, we illustrate some of the results derived in this paper through two numerical experiments in Section~\ref{sec::Num_Exp} and provide concluding remarks in Section~\ref{Section:Conclusion}.

\begin{figure*}[t!]
\centering
\includegraphics[width=0.92\textwidth]{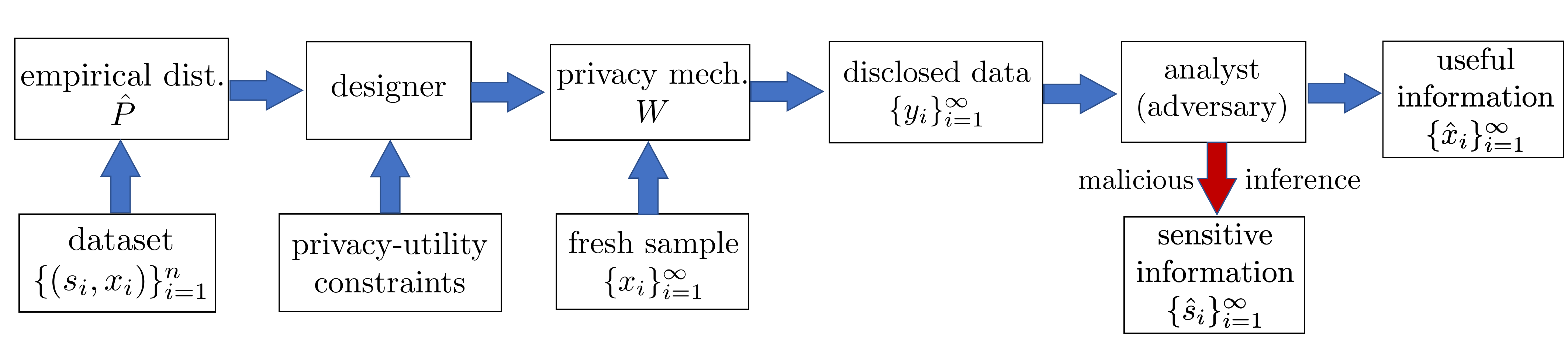}
\caption{\small{Typical setting for the design and deployment of privacy mechanisms.}}
\label{fig:rob_framework}
\end{figure*}

\subsection{Contributions}

In Section~\ref{Section:PointwiseRobustness}, we provide probabilistic upper bounds for the difference between the privacy-utility guarantees for the empirical and the true distributions under five different information metrics: probability of correctly guessing, $f$-information with $f$ locally Lipschitz, Arimoto's mutual information ($\alpha$-leakage) of order $\alpha>1$, Sibson's mutual information of order $\alpha>1$, and maximal $\alpha$-leakage of order $\alpha>1$. These bounds, which scale as $O(1/\sqrt{n})$ with $n$ being the sample size, hold uniformly across all privacy mechanisms and, hence, can be used even for privacy mechanisms that do not have explicit descriptions, e.g., the privacy mechanisms implemented using  neural networks in \cite{huang2017context}. Explicit constants depending on the information metrics under consideration are provided. The proofs of these bounds rely on known large deviations results \cite{weissman2003inequalities} and Lipschitz continuity properties of information leakage measures established in this paper. These continuity properties can be combined with results for other estimation frameworks different from the large deviations one, e.g., the $\ell_1$-minimax setting in \cite{kamath2015learning} and references therein. In those cases, the role of the empirical estimator is naturally replaced by other estimators, e.g., the add-constant estimator in \cite{kamath2015learning}.

We study the convergence properties of optimal privacy mechanisms in Section~\ref{Section:convergence}. Specifically, we consider the case when a privacy mechanism designer constructs a sequence optimal privacy mechanisms for a sequence of joint distributions converging to the true distribution. In this context, we show that while the sequence of optimal privacy mechanisms do not necessarily form a Cauchy sequence, the distance between the privacy mechanisms in the sequence and the set of optimal privacy mechanisms with respect to the true distribution does go to zero. This convergence is analyzed for information measures which satisfy three technical conditions outlined in Section~\ref{Section:convergence}. These conditions are satisfied by probability of correctly guessing, $f$-information with $f$ locally Lipschitz, and Arimoto's mutual information ($\alpha$-leakage) of order $\alpha>1$. As a by-product, we prove that, under these conditions, the privacy-utility function as a function of the joint distribution is continuous except possibly at the boundary of its domain.

In Section~\ref{Section:UniformRobustness}, we introduce the notion of uniform privacy mechanism and prove the existence of {\it optimal uniform privacy mechanisms}, i.e., uniform privacy mechanisms that attain the best utility in a max-min sense. Optimal uniform privacy mechanisms are considerably harder to design than their non-uniform counterparts as privacy has to be guaranteed for any distribution within a neighborhood of the empirical distribution. Nonetheless, we show that they can be approximated by certain mechanisms designed to deliver privacy \emph{only} for the empirical distribution. This approximation result circumvents the highly non-trivial task of designing optimal uniform privacy mechanisms.

\subsection{Related Work}

Differential privacy is a popular measure of privacy which aims at answering queries while simultaneously ensuring privacy of individual records in the database \cite{dwork2011differential,dwork2014algorithmic}. It does not take into account the distribution of the entire dataset which results in its robustness with respect to the dataset distribution, a property that is not satisfied \textit{a priori} by information-theoretic privacy mechanisms. In this work we analyze the degree of robustness to variations of the dataset distribution of certain information-theoretic measures of privacy. 

The study of privacy from an information-theoretic point of view heavily depends on the chosen information metric. One commonly-used information leakage metric is Shannon's mutual information \cite{shannon1949communication} and its generalizations \cite{verdu2015alpha}. In particular, both $\alpha$-leakage and  $f$-information --- a special case of Csisz{\'a}r's $f$-divergence \cite{ciszar1967information} ---  have been used to formulate the privacy-utility trade-offs problem, see e.g., Liao~\etal \cite{liao2018privacy}. These generalizations seek to study privacy using information metrics that have appeared in the information theory literature and carry some operational (statistical) meaning. There are also information measures which have been introduced specifically in the context of privacy and information leakage. For example, Issa~\etal \cite{issa2016maximal,issa2016operational} introduced a metric called maximal leakage to measure privacy leakage, later extended by Liao~\etal \cite{liao2018tunable} in a tunable measure called $\alpha$-leakage. Duchi~\etal \cite{duchi2014privacy} presented locally differential privacy. We refer the reader to Wagner and Eckhoff~\cite{wagner2018technical} for a survey of privacy metrics and remark that information metrics also shed light on many secrecy problems, e.g., cryptosystems with short keys \cite{li2018maximal}, side-channel attack \cite{issa2018operational}, and entropic security \cite{russell2002fool,dodis2005entropic}. 

Many methods to design information-theoretic privacy mechanisms (see Figure~\ref{fig:rob_framework}) have been proposed in the past, see, e.g., \cite{rebollo2010t,du2012privacy,sankar2013utility,asoodeh2017estimation,liao2017hypothesis,issa2018operational}. However, all these works assume that the designer has full knowledge of the data distribution. In practice, the designer may only have access to $n$ i.i.d.~samples, creating a mismatch between the distribution used to design the privacy mechanism and the true one. We significantly generalize the work of Wang and Calmon \cite{wang2017estimation}, which studies the robustness of privacy mechanisms when privacy and utility are measured using $\chi^2$-information, to encompass a wide range of information leakage measures. To be more specific, we show that for every information leakage measure considered (see Theorem~\ref{thm:meta_upper_bounds}), the discrepancy between the privacy-utility guarantees for the empirical and true distributions scales as $O(1/\sqrt{n})$. It is important to remark that, in the context of the information bottleneck, Shamir~\etal \cite{shamir2010learning} established similar upper bounds for mutual information that scale as $O(\log(n)/\sqrt{n})$. Hence, mutual information appears to fall outside the scope of the techniques presented in this paper. A related problem is that of estimating information-theoretic measures from a limited number of samples, as studied by, for example, Jiao~\etal \cite{jiao2015minimax}, Wu and Yang \cite{wu2016minimax} and Issa~\etal \cite{issa2018operational} in the context of privacy. Our work differs from those as we do not focus on the construction of estimators suited for a specific information leakage measure, instead we analyze the performance of the plug-in estimator under different leakage measures. Given the ease of implementation of this estimator, we believe our contributions might be of interest to practitioner.

The fundamental trade-offs between privacy guarantees and statistical utility in data disclosure have been studied in several papers from an information-theoretic view. For example, Yamamoto \cite{yamamoto1983source} developed the trade-off between rate, distortion, and equivocation for a specific source coding model. Rebollo-Monedero~\etal \cite{rebollo2010t}, Calmon and Fawaz \cite{du2012privacy}, and Sankar~\etal \cite{sankar2013utility} characterized privacy-utility trade-offs using tools from rate-distortion theory. Varodayan and Khisti \cite{varodayan2011smart} and Tan~\etal \cite{tan2013increasing} considered the privacy-energy efficiency trade-off in smart meter systems. Makhdoumi~\etal \cite{makhdoumi2014information} showed the connection between the privacy funnel and the information bottleneck proposed by Tishby~\etal \cite{tishby2000information}. The privacy-utility function, as presented in Section~\ref{Subsection:ProblemSetup}, has been analyzed to find the fundamental limits of privacy-utility trade-offs under different metrics, see, e.g., Calmon~\etal \cite{calmon2013bounds} and Asoodeh~\etal \cite{asoodeh2016information}. It has been shown that this function exhibits some common properties across different metrics. For example, it is continuous, concave, and strictly increasing with respect to the privacy parameter, see also Wang and Calmon \cite{wang2017estimation} and Asoodeh~\etal \cite{asoodeh2017estimation}.
%

The study of robustness has appeared in many different areas, such as optimization \cite{ben2009robust}, statistics \cite{huber2011robust}, artificial intelligence \cite{valiant2000robust}, and machine learning \cite{charikar2017learning}. The notion of robustness considered in this work is closely related to generalization in machine learning \cite{shalev2014understanding}, which aim at analyzing different performances of the learner's output between training and testing phases. More specifically, a learner may output a classifier $h$ which minimizes a given empirical risk over a sample $S=(z_1,\cdots,z_m)$, namely,
\begin{align}
    L_S(h)\defined \frac{1}{m}\sum_{i=1}^m \ell(h,z_i),
\end{align}
where $\ell$ is a given loss function. Nonetheless, in practice, the same classifier $h$ might exhibit a very different true risk, defined as
\begin{align}
\label{eq:risk_func_defn}
    L(h) \defined \EE{\ell(h,Z)},
\end{align}
where $Z$ is a random variable distributed according to the dataset distribution. The goal of generalization is to analyze the difference between the empirical risk and the true risk and derive bounds for this discrepancy. This problem has been studied in computer science \cite{kearns1994introduction} and, recently, has been analyzed using tools from information theory by Xu and Raginsky \cite{xu2017information} and Russo and Zou \cite{russo2015much}. Since most information-theoretic metrics do not seem to be expressible in terms of a loss function as in \eqref{eq:risk_func_defn}, classical generalization results cannot be directly applied for analyzing the robustness of information-theoretic privacy mechanisms. Recently, the robustness of privacy mechanisms has been investigated using tools from information theory in the works of Wang and Calmon \cite{wang2017estimation} and Issa~\etal \cite{issa2018operational}. It is important to mention that some techniques developed for privacy preservation, such as differential privacy, can be useful to analyze generalization guarantees, as shown by Dwork~\etal \cite{dwork2015preserving,dwork2015generalization,dwork2015reusable}.

\subsection{Notation}
\label{Subsection:Notation}

Given random variables $U$ and $U'$ supported over a finite alphabet $\mathcal{U}$ and another random variable $V$ supported over a finite alphabet $\mathcal{V}$, we let $P_{U'} \cdot P_{V|U}$ be the joint distribution over $\mathcal{U}\times\mathcal{V}$ determined by
\begin{equation}
    \label{eq:ConventionProductDistributionTransitionMatrix}
    (P_{U'} \cdot P_{V|U})(u,v) = P_{U'}(u) P_{V|U}(v|u).
\end{equation}
Observe that with this notation, $P_U \cdot P_{V|U} = P_{U,V}$. Given another random variable $\tilde{V}$ supported over a finite alphabet $\tilde{\mathcal{V}}$, we let $P_{\tilde{V}|U} P_{V|\tilde{V}}$ be the transition probability matrix determined by
\begin{equation}
    \label{eq:ConventionProductTransitionMatrices}
    (P_{\tilde{V}|U} P_{V|\tilde{V}})(v|u) = \sum_{\tilde{v}\in\tilde{\mathcal{V}}} P_{\tilde{V}|U}(\tilde{v}|u) P_{V|\tilde{V}}(v|\tilde{v}).
\end{equation}

We let $S$ and $X$ be two random variables with discrete supports $\mathcal{S}$ and $\mathcal{X}$, respectively. We let $P$ denote the joint distribution $P_{S,X}$ and $\hat{P}$ be a generic estimate of $P$. The empirical distribution obtained from $n$ i.i.d.~samples drawn from $P$ is denoted by $\hat{P}_n$. We let $(P_n)_{n=1}^{\infty}$ be any sequence of joint distributions converging to the joint distribution $P$. Any generic distribution over $\mathcal{S}\times \mathcal{X}$ is denoted by $Q$. Also, any sequence of distributions over $\mathcal{S}\times \mathcal{X}$ is denoted by $(Q_n)_{n=1}^{\infty}$.

Let $(\mathcal{A},d_{\mathcal{A}})$ and $(\mathcal{B},d_{\mathcal{B}})$ be two metric spaces. We say that $f:\mathcal{A}\to\mathcal{B}$ is Lipschitz over $\mathcal{A}'\subseteq\mathcal{A}$ if there exists $L\geq0$ such that, for all $a_1,a_2\in\mathcal{A}'$,
\begin{equation}
    d_{\mathcal{B}}(f(a_1),f(a_2)) \leq L d_{\mathcal{A}}(a_1,a_2).
\end{equation}
For convenience, we summarize some common notation in Table~\ref{table:Notation}.
\begin{table*}
\small
\centering
\renewcommand{\arraystretch}{1.5}
\begin{tabular}{lll}
\toprule
Category & Notation & Meaning \\  
\toprule

\multirow{2}{*}{Norms} 
& $\|a\|_{\alpha}$ & $\alpha$-norm with $\alpha\in[1,\infty)$: $\left(\sum_i |a_i|^\alpha\right)^{1/\alpha}$\\
\cline{2-3}
& $\|a\|_{\infty}$ & $\infty$-norm: $\max_{i} |a_i|$\\
\midrule

\multirow{6}{*}{Joint Distributions} 
&$P$ & True joint distribution: $P_{S,X}$\\
\cline{2-3}
&$\hat{P}$ & Estimated distribution: $\hat{P}_{S,X}$\\
\cline{2-3}
&$\hat{P}_n$ & Empirical distribution obtained from $n$ i.i.d. samples\\
\cline{2-3}
&$(P_n)_{n=1}^{\infty}$ &Any sequence of joint distributions converging to $P$\\
\cline{2-3}
&$Q$ &Any joint distribution over $\mathcal{S}\times\mathcal{X}$\\
\cline{2-3}
&$(Q_n)_{n=1}^{\infty}$ &Any sequence of joint distributions over $\mathcal{S}\times\mathcal{X}$\\
\midrule

\multirow{1}{*}{Channels} 
&$W$ & Any privacy mechanism (channel) from $\mathcal{X}$ to $\mathcal{Y}$\\
\midrule

\multirow{4}{*}{Sets} 
& $\mathcal{W}$ & Set of all privacy mechanisms\\
\cline{2-3}

& $\mathcal{W}_N$ & Set of all privacy mechanisms from $\mathcal{X}$ to $\{1\,\ldots,N\}$\\
\cline{2-3}
& $\mathcal{P}$ & Set of all joint distributions over $\mathcal{S}\times \mathcal{X}$\\
\cline{2-3}
& $\mathcal{Q}$ & Any closed subset of $\mathcal{P}$\\
\bottomrule 
\end{tabular}
\\[10pt]
\caption{List of notation used in this paper.}
\label{table:Notation}
\end{table*}
\section{Preliminaries and Problem Setup}
\label{sec:pre_setup}

In this section, we review preliminary material regarding information leakage measures and large deviations bounds for distribution estimation. In addition, we formally introduce the problems addressed in this paper.

\subsection{Information Leakage Measures}

Here we review the definition and some basic properties of information leakage measures which are commonly used in the literature. An information leakage measure $\mathcal{L}(U\to V)$ quantifies how much information $V$ leaks about $U$. Specifically, $U$ and $V$ can represent raw and disclosed datasets, respectively, and, in this case, $\mathcal{L}(U \to V)$ measures the utility of the disclosed dataset. On the other hand, when $U$ and $V$ represent private information and information available to an adversary, then $\mathcal{L}(U\to V)$ quantifies the adversary's ability of inferring $U$ from the observation of $V$. In the privacy literature it is valuable to relate information leakage measures with a specific loss/gain function (implicitly) adopted by an adversary. In what follows, we review a range of information leakage measures along with their operational meanings in terms of loss/gain functions.

\subsubsection{$\alpha$-Leakage}

In an effort to unify several information leakage measures within a single framework, Liao \etal recently introduced an information leakage measure called $\alpha$-leakage \cite{liao2018tunable}.

\begin{defn}[Definition~4, \cite{liao2018tunable}]
Let $U$ and $V$ be two random variables supported over finite sets $\mathcal{U}$ and $\mathcal{V}$, respectively. For $\alpha\in(1,\infty)$, the $\alpha$-leakage from $U$ to $V$ is defined as
\begin{equation}
\label{eq:DefAlphaLeakage}
\begin{aligned}
    \mathcal{L}_\alpha(U \to V) 
    \defined \frac{\alpha}{\alpha-1} \log \frac{\max\limits_{\hat{U}:U \to V \to \hat{U}} \mathbb{E}\left[\Pr(\hat{U} = U | U,V)^{\frac{\alpha-1}{\alpha}} \right]}{\max\limits_{\hat{U} : U \indep \hat{U}} \mathbb{E}\left[\Pr(\hat{U} = U | U)^{\frac{\alpha-1}{\alpha}} \right]}.
\end{aligned}
\end{equation}
The value of $\mathcal{L}_\alpha(U \to V)$ is extended by continuity to $\alpha=1$ and $\alpha=\infty$.
\end{defn}

In terms of the adversarial setting described at the beginning of this section, $\alpha$-leakage with $\alpha\in(1,\infty]$ can be interpreted as the multiplicative increase, upon observing $V$, of the expected gain of an adversary. In order to make this observation precise, note that the expected value in the denominator of the logarithmic term in \eqref{eq:DefAlphaLeakage} equals
\begin{equation}
\label{eq:DenAlphaLeakage}
    \mathbb{E}\left[\Pr(\hat{U} = U | U)^{\frac{\alpha-1}{\alpha}} \right] = \sum_{u\in\mathcal{U}} P_U(u) P_{\hat{U}}(u)^{\frac{\alpha-1}{\alpha}}.
\end{equation}
In particular, the RHS of \eqref{eq:DenAlphaLeakage} equals expected gain of $P_{\hat{U}}$ w.r.t.~the gain function 
\begin{align}
\label{eq:GainAlphaLeakage}
    \mathsf{gain}(u,P_{\hat{U}})=P_{\hat{U}}(u)^{\frac{\alpha-1}{\alpha}}.
\end{align}
Thus, the denominator of the logarithmic term in \eqref{eq:DefAlphaLeakage} is the largest expected gain of an adversary with no additional information apart from the distribution of $U$. Also, the expected value in the numerator of the logarithmic term in \eqref{eq:DefAlphaLeakage} equals
\begin{equation}
\label{eq:NumAlphaLeakage}
\begin{aligned}
    \mathbb{E}\left[\Pr(\hat{U} = U | U,V)^{\frac{\alpha-1}{\alpha}} \right] 
    = \sum_{u\in\mathcal{U}} \sum_{v\in\mathcal{V}} P_{U,V}(u,v) P_{\hat{U}|V}(u|v)^{\frac{\alpha-1}{\alpha}}.
\end{aligned}
\end{equation}
Thus, the numerator of the logarithmic term in \eqref{eq:DefAlphaLeakage} is the largest expected gain of an adversary that has access to $V$ and the joint distribution of $U$ and $V$. From these observations we conclude that $\mathcal{L}_\alpha(U \to V)$ measures the multiplicative increase of the best expected gain upon observing $V$. In a similar way, $\mathcal{L}_1(U \to V)$ can be interpreted as additive increase, upon observing $V$, of the expectation of the gain function
\begin{equation}
    \mathsf{gain}(u,P_{\hat{U}})=\log P_{\hat{U}}(u).
\end{equation}

In addition to its operational definition, $\alpha$-leakage can be related to Arimoto's mutual information \cite{arimoto1977information,fehr2014conditional} whose definition is provided below.

\begin{defn}
\label{Def:ArimotoMI}
Let $U$ and $V$ be two random variables supported over finite sets $\mathcal{U}$ and $\mathcal{V}$, respectively. Their Arimoto's mutual information of order $\alpha\in(1,\infty)$ is defined as
\begin{equation}
    I^{\rm A}_{\alpha}(P_{U,V}) \defined \frac{\alpha}{\alpha-1} \log \frac{\sum_v \|P_{U,V}(\cdot,v)\|_{\alpha}}{\|P_U(\cdot)\|_{\alpha}}.
\end{equation}
Also, by continuous extension,
\begin{align}
    I^{\rm A}_1(P_{U,V}) \defined I(P_{U,V}) \quad \text{and}\quad
    I^{\rm A}_\infty(P_{U,V}) \defined \log \dfrac{\sum_v \max_u P_{U,V}(u,v)}{\max_u P_U(u)},
\end{align}
where $I(P_{U,V})$ denotes Shannon's mutual, i.e., $I(P_{U,V}) \defined \sum_{u\in\mathcal{U}} \sum_{v\in\mathcal{V}} P_{U,V}(u,v) \log \dfrac{P_{U,V}(u,v)}{P_U(u)P_V(v)}$.
\end{defn}

In \cite{liao2018tunable}, Liao \etal proved that $\alpha$-leakage is indeed equal to Arimoto's mutual information of order $\alpha$.

\begin{prop}[Theorem~1, {\cite{liao2018tunable}}]
\label{Prop:alphaLeakageArimoto}
Let $U$ and $V$ be two random variables supported over finite sets $\mathcal{U}$ and $\mathcal{V}$, respectively. For $\alpha\in[1,\infty]$, $\alpha$-leakage satisfies
\begin{equation}
    \mathcal{L}_\alpha(U \to V) = I^{\rm A}_\alpha(P_{U,V}).
\end{equation}
\end{prop}

The previous proposition shows that $\alpha$-leakage recovers Shannon's mutual information in the extremal case $\alpha = 1$. It is important to point out that Shannon's mutual information has been extensively used in the context of privacy. Of particular relevance for the present papers are the works of Rebollo-Monedero \etal \cite{rebollo2010t}, Sankar \etal \cite{sankar2013utility}, Calmon and Fawaz \cite{du2012privacy}, Asoodeh \etal \cite{asoodeh2016information}, and references therein. On the other extreme, when $\alpha=\infty$, $\alpha$-leakage is closely related to another information leakage measure: probability of correctly guessing. This information leakage measure, whose definition is recalled next, has been used recently in the context of privacy-utility trade-offs by Asoodeh \etal \cite{asoodeh2017estimation} and in the broader privacy literature by Smith \cite{smith2009foundations}, Braun \etal \cite{braun2009quantitative}, Barthe and Kopf \cite{barthe2011information}, and references therein.

\begin{defn}
\label{Def:Pc}
Let $U$ and $V$ be two random variables supported over finite sets $\mathcal{U}$ and $\mathcal{V}$, respectively. The probability of correctly guessing $U$ and the probability of correctly guessing $U$ given $V$ are given by
\begin{align}
    P_c(U) &\defined \max_{u\in\mathcal{U}} P_U(u),\\
    P_c(U|V) &\defined \sum_{v\in\mathcal{V}} \max_{u\in\mathcal{U}} P_{U,V}(u,v).
\end{align}
\end{defn}

As its name suggests, $P_c(U)$ is equal to the (largest) probability of correctly guessing $U$ without any side information but the distribution of $U$. Similarly, $P_c(U|V)$ is equal to the (largest) probability of correctly guessing $U$ given $V$ and the joint distribution $P_{U,V}$. This interpretation can be made precise using the adversarial setting and the gain function defined in \eqref{eq:GainAlphaLeakage} with $\alpha=\infty$. Indeed, a simple manipulation shows that
\begin{align}
    P_c(U) = \max_{\hat{U}:U\indep\hat{U}} \mathbb{E}\left(\mathbbm{1}_{U = \hat{U}} \right) \quad \text{and}\quad
    P_c(U|V) = \max_{\hat{U}: U \to V \to \hat{U}} \mathbb{E}\left(\mathbbm{1}_{U = \hat{U}} \right),
\end{align}
which corresponds to the adversarial setting under the 0-1 loss. Observe that under this loss, the optimal adversary strategy is the same as the maximum a posteriori (MAP) decoder in the communication literature. Finally, observe that
\begin{equation}
\label{eq:RelationLinftyPc}
    \mathcal{L}_{\infty}(U \to V) = \log \frac{P_c(U|V)}{P_c(U)},
\end{equation}
evidencing the intrinsic relation between $\infty$-leakage and probability of correctly guessing.

\begin{rem}
Despite the close relation between $\infty$-leakage and probability of correctly guessing, we always treat them separately. On the one hand, probability of correctly guessing is of interest in its own right. On the other hand, the results for probability of correctly guessing are easier to derive than their $\infty$-leakage counterparts. As a result, working explicitly with probability of correctly guessing allows us to present the key techniques used in this paper while keeping the technical difficulties at a minimum.
\end{rem}

In \cite{issa2016operational,issa2018operational}, Issa \etal introduced the notion of maximal leakage in order to quantify the adversary's capability to infer any (randomized) function of $U$ from $V$. Motivated by this notion, Liao \etal introduced {\it maximal $\alpha$-leakage}, a tunable information leakage measure that is equal to maximal leakage when $\alpha=\infty$. 
\begin{defn}[Definition~5, \cite{liao2018tunable}]
Let $U$ and $V$ be two random variables supported over finite sets $\mathcal{U}$ and $\mathcal{V}$, respectively. For $\alpha\in[1,\infty]$, the maximal $\alpha$-leakage from $U$ to $V$ is defined as
\begin{equation}
\label{eq:DefMaximalAlphaLeakage}
    \mathcal{L}_\alpha^{\max}(U \to V) \defined \sup_{T:T \to U \to V} \mathcal{L}_\alpha(T \to V),
\end{equation}
where $T$ represents any (randomized) function of $U$ that takes values on a finite but arbitrary alphabet.
\end{defn}

Arimoto's mutual information can be regarded as a generalization of Shannon's mutual information, see, e.g., \cite{verdu2015alpha}. Another such generalization is Sibson's mutual information \cite{sibson1969information} which, as shown below, is also related to maximal $\alpha$-leakage.

\begin{defn}[Definition~4, \cite{verdu2015alpha}]
\label{Def:SibsonMI}
Let $U$ and $V$ be two random variables supported over finite sets $\mathcal{U}$ and $\mathcal{V}$, respectively. Their Sibson's mutual information of order $\alpha\in(1,\infty)$ is defined as 
\begin{equation}
    I^{\rm S}_\alpha(P_{U,V}) \defined \frac{\alpha}{\alpha-1} \log \sum_{v\in\mathcal{V}} \|P_U(\cdot)^{1/\alpha} P_{V|U}(v|\cdot)\|_\alpha.
\end{equation}
Also, $I^{\rm S}_1(P_{U,V}) = I(P_{U,V})$ and $I^{\rm S}_\infty(P_{U,V}) = \log \left( \sum_{v\in\mathcal{V}} \max_{u\in\mathcal{U}} P_{V|U}(v|u)\right)$.
\end{defn}

The following proposition shows that maximal $\alpha$-leakage is the Arimoto's capacity under an input support constraint which, in turn, is equal to Sibson's capacity under the same constraint.

\begin{prop}[Theorem~2, {\cite{liao2018tunable}}]
\label{Prop:MaximalAlphaLeakageExpression}
Let $U$ and $V$ be two random variables supported over finite sets $\mathcal{U}$ and $\mathcal{V}$, respectively. For $\alpha\in(1,\infty]$, maximal $\alpha$-leakage satisfies
\begin{align}
    \mathcal{L}_\alpha^{\max}(U \to V) = \sup_{P_{\tilde{U}}} I^{\rm A}_\alpha(P_{\tilde{U}}\cdot P_{V|U})
    = \sup_{P_{\tilde{U}}} I^{\rm S}_\alpha(P_{\tilde{U}}\cdot P_{V|U}),
\end{align}
where the support of $P_{\tilde{U}}$ is a subset of the support of $P_U$\footnote{In \cite[Theorem~2]{liao2018tunable}, $P_{\tilde{U}}$ ranges over the distributions with the same support as $P_U$. However, their proof readily implies  Proposition~\ref{Prop:MaximalAlphaLeakageExpression} as stated in this work.}. Also, $\mathcal{L}_1^{\max}(U \to V) = I(P_{U,V})$.
\end{prop}
In particular, maximal $\alpha$-leakage of order infinity (i.e., maximal leakage) is the Sibson's mutual information of order infinity.
\begin{prop}[Corollary~1, {\cite{issa2016operational}}]
\label{Prop:MaxL_SMI_inf}
Let $U$ and $V$ be two random variables supported over finite sets $\mathcal{U}$ and $\mathcal{V}$, respectively. Then,
\begin{equation}
    \mathcal{L}_{\infty}^{\max}(U \to V) = I_{\infty}^S(P_{U,V}).
\end{equation}
\end{prop}

\subsubsection{$f$-Information}

We finish our review on information leakage measures by recalling the definition and some basic properties of $f$-information. This information leakage measure is defined in terms of Csisz{\'a}r's $f$-divergence whose definition is recalled next for the reader's convenience.

\begin{defn}[Definition~1.1, {\cite{ciszar1967information}}]
Let $f:(0,\infty)\to\Reals$ be a convex function with $f(1)=0$. Assume that $Q_1$ and $Q_2$ are two probability distributions over a finite set $\mathcal{Z}$ and that $Q_1 \ll Q_2$. The $f$-divergence between $Q_1$ and $Q_2$ is given by
\begin{equation}
D_f(Q_1 \| Q_2)\defined \sum_{z\in\mathcal{Z}} Q_2(z) f\left(\frac{Q_1(z)}{Q_2(z)}\right).
\end{equation}
\end{defn}

Given a function $f$ as in the previous definition, its corresponding $f$-information is defined as follows.

\begin{defn}
Let $f:(0,\infty)\to\Reals$ be a convex function with $f(1)=0$. Furthermore, let $U$ and $V$ be two random variables supported over finite sets $\mathcal{U}$ and $\mathcal{V}$, respectively. Their $f$-information is defined by
\begin{equation}
\begin{aligned}
I_f(P_{U,V}) &\defined D_f(P_{U,V} \| P_U P_V)\\
&=\sum_{u\in\mathcal{U}} \sum_{v\in\mathcal{V}} P_{U}(u)P_{V}(v) f\left(\frac{P_{U,V}(u,v)}{P_U(u)P_V(v)}\right).
\end{aligned}
\end{equation}
\end{defn}

While an $f$-information may not have a straightforward interpretation in operational terms for a specific function $f$, many functions $f$ do, e.g., mutual information and $\chi^2$-information discussed below. Hence, a general treatment allows us to simultaneously handle those $f$-information for which there is a concrete operational meaning and those whose operational interpretation is yet to be discovered. More recent developments about the properties of $f$-divergence can be found in Raginsky \cite{raginsky2016strong}, Calmon \etal \cite{calmon2017strong}, and the references therein.

In the context of privacy, $\chi^2$-information is a fine example of an $f$-information with a tangible operational interpretation. Following the standard convention, the $\chi^2$-information between two random variables $U$ and $V$ is defined as $\chi^2(P_{U,V}) \defined I_f(P_{U,V})$ with $f(t)=(t-1)^2$. Relying on the so-called \textit{principal inertia components} (PICs) \cite{greenacre1984theory}, Calmon \etal \cite{du2017principal} showed that if $\chi^2(U;V)<\epsilon$ for some $0<\epsilon<1$, then the \textit{minimum mean-squared error} (MMSE) of reconstructing \textit{any} zero-mean unit-variance function of $U$ given $V$ is lower bounded by $1-\epsilon$, i.e., no function of $U$ can be reconstructed with small MMSE given an observation of $V$. Thus, $\chi^2$-information measures an adversary's ability to reconstruct functions of $U$ from $V$ under an MMSE loss. Recently, $\chi^2$-information has been used in the context of privacy-utility trade-offs by Wang and Calmon in \cite{wang2017estimation}.

\subsection{Large Deviations Bounds for Distribution Estimation}
\label{Subsection:EmpiricalDistribution}

Now we review some preliminary results regarding the distance between the empirical and true distributions. Throughout this work, we measure the distance between two probability distributions over $\mathcal{Z}$, say $Q_1$ and $Q_2$, by their $\ell_1$-distance which is given by
\begin{equation}
    \|Q_1 - Q_2\|_1 \defined \sum_{z\in\calZ} |Q_1(z) - Q_2(z)|.
\end{equation}
A result by Weissman \etal \cite[Theorem~2.1]{weissman2003inequalities} establishes that, for all $\epsilon > 0$,
\begin{equation}
    \Pr\left(\|\hat{P}_n-P\|_1 \geq \epsilon\right) 
    \leq (2^{|\mathcal{Z}|}-2)\exp(-n\phi(\pi_P)\epsilon^2/4),
\end{equation}
where $P$ is a probability distribution over the finite set $\mathcal{Z}$, $\hat{P}_n$ is the empirical distribution obtained from $n$ i.i.d.~samples from $P$, $\displaystyle \pi_P \defined \max_{\mathcal{A}\subseteq\mathcal{Z}} \min\{ P(\mathcal{A}),1-P(\mathcal{A})\}$, and
\begin{equation}
    \phi(p)\defined 
    \begin{cases}
    \frac{1}{1-2p}\log\frac{1-p}{p} & p\in[0,1/2), \\ 
    2 & p=1/2.
    \end{cases}
\end{equation}
Note that $\phi(p)\geq 2$ for all $p\in[0,1/2]$. Hence,
\begin{equation}
\label{equa:L1bound_weissman}
\begin{aligned}
    \Pr\left(\|\hat{P}_n-P\|_1 \geq \epsilon\right)
    &\leq \exp(|\mathcal{Z}|)\exp(-n\epsilon^2/2).
\end{aligned}
\end{equation}
By taking $\mathcal{Z}=\mathcal{S}\times\mathcal{X}$ and $\epsilon = \sqrt{\frac{2}{n}\left(|\mathcal{S}|\cdot|\mathcal{X}|-\log\beta\right)}$, inequality~\eqref{equa:L1bound_weissman} implies that, with probability at least $1-\beta$,
\begin{equation}
\label{equa:L1bound_Devroye}
    \|\hat{P}_n - P\|_1 \leq \sqrt{\frac{2}{n}\left(|\mathcal{S}|\cdot|\mathcal{X}|-\log\beta\right)}.
\end{equation}
Even though in this paper we focus on large deviations results, it is worth pointing out that the order $O(\sqrt{|\mathcal{S}|\cdot|\mathcal{X}|/n})$ is present in other fundamental settings, e.g., the minimax expected loss framework in \cite[Corollary~9]{kamath2015learning}.

\subsection{Problem Setup}
\label{Subsection:ProblemSetup}

Suppose that $S$ is a variable to be hidden (e.g., political preference) and $X$ is an observed variable  (e.g., movie ratings) that is correlated with $S$. In order to receive some utility (e.g., personalized recommendations), we would like to disclose as much information about $X$ without compromising $S$. An approach with rigorous privacy guarantees is to release a new random variable $Y$ produced by applying a randomized mapping to $X$, hence forming a Markov chain $S \to X \to Y$. This mapping, called the privacy mechanism, is designed to satisfy a specific privacy constraint.

In the sequel, we assume that $S$ and $X$ are discrete random variables with support sets $\mathcal{S}$ and $\mathcal{X}$, respectively. Let $\mathcal{W}$ be the set of all privacy mechanisms which take $X$ as input and display a discrete random variable $Y$ as output. More specifically, let
\begin{equation}
\label{eq:DefCalW}
    \mathcal{W} \defined \bigcup_{N\geq1} \bigg\{W\in[0,1]^{|\mathcal{X}| \times N} : \forall x\in\mathcal{X}, \sum_{y=1}^N W(x,y) =1\bigg\}.
\end{equation}
Through this work, we denote the joint distribution of $S$ and $X$ by $P$ and the privacy mechanism producing $Y$ from $X$ by $W$, i.e., $W = P_{Y|X}$. The privacy leakage and the utility generated by a mapping $W\in\mathcal{W}$ for the true distribution $P$ are denoted by $\mathcal{L}(P,W)$ and $\mathcal{U}(P,W)$, respectively. Throughout this paper, $\mathcal{L}$ and $\mathcal{U}$ are either probability of correctly guessing, $f$-information with $f$ locally Lipschitz, Arimoto's mutual information ($\alpha$-leakage) of order $\alpha\in(1,\infty]$, Sibson's mutual information of order $\alpha\in(1,\infty]$, or maximal $\alpha$-leakage of order $\alpha\in(1,\infty]$.

In this framework, a natural problem is to characterize the fundamental trade-off between privacy and utility as captured by the following definition. We denote by $\mathcal{P}$ the set of all probability distributions over $\mathcal{S}\times\mathcal{X}$.

\begin{defn}
\label{defn:PUT}
For a given $P\in\mathcal{P}$ and $\displaystyle \epsilon\geq\inf_{W\in\mathcal{W}} \mathcal{L}(P,W)$, the {\it privacy-utility function} is defined as
\begin{equation}
\label{eq:DefH}
    \mathsf{H}(P;\epsilon) \defined \sup_{W\in\mathcal{D}(P;\epsilon)} \mathcal{U}(P,W),
\end{equation}
where $\mathcal{D}(P;\epsilon) \defined \{ W\in\mathcal{W} : \mathcal{L}(P,W) \leq \epsilon\}$. Furthermore, the collection of all optimal privacy mechanisms for $P$ at $\epsilon$ is defined as
\begin{align}
\label{equa:opt_privacy_mehc}
    \mathcal{W}^*(P;\epsilon)\defined \{W\in \mathcal{W} : \mathcal{L}(P,W)\leq \epsilon,\  \mathcal{U}(P,W)=\mathsf{H}(P;\epsilon)\}.
\end{align}
\end{defn}

Observe that, by definition, $\mathcal{D}(P;\epsilon)$ is the set of all privacy mechanisms providing an $\epsilon$-privacy guarantee for $P$. Hence, the privacy-utility function in Definition~\ref{defn:PUT} quantifies the best utility achieved by any privacy mechanism providing an $\epsilon$-privacy guarantee for $P$.

This specific type of privacy-utility trade-off (PUT), and the optimal privacy mechanisms associated to it, has been investigated for several measures of privacy and utility, see, for example, \cite{asoodeh2016information,asoodeh2017estimation,wang2017estimation,rassouli2019optimal}. These investigations often rely on the implicit assumption that the data distribution is, for the most part, known. However, in practice, the data distribution may only be accessed through a limited number of samples. In this work, we revisit this assumption and study its implications in the design and performance of privacy mechanisms. Next we present the problems addressed and the main contributions of this paper.

\subsubsection{Discrepancy of Privacy-Utility Guarantees}
\label{Subsection:PointwiseRobust}

In practice, the designer may not have access to the true distribution $P$, but only to samples drawn from this distribution. In this case, the privacy-utility guarantees for a distribution estimated from the samples, say $\hat{P}$, and the true distribution $P$ might be different. For any given privacy mechanism $W$, these discrepancies are effectively quantified by
\begin{equation}
\label{eq:Discrepancies}
|\mathcal{L}(\hat{P},W) - \mathcal{L}(P,W)| \quad \text{and} \quad |\mathcal{U}(\hat{P},W) - \mathcal{U}(P,W)|.
\end{equation}
In Section~\ref{Section:PointwiseRobustness} we provide probabilistic upper bounds for the discrepancies in \eqref{eq:Discrepancies} when the estimated distribution $\hat{P}$ is the empirical distribution $\hat{P}_n$ of $n$ i.i.d.~samples drawn from $P$. These upper bounds depend on the sample size $n$, the alphabet sizes $|\mathcal{S}|$ and $|\mathcal{X}|$, and, in some cases, the probability of the least likely symbol of the marginals of $\hat{P}_n$. Their derivations rely on the large deviations results recalled in Section~\ref{Subsection:EmpiricalDistribution} and continuity properties of information leakage measures established in this work. We summarize these results in the following meta theorem (see Theorem~\ref{thm:meta_upper_bounds}). For ease of notation, we let
\begin{equation}
    \mathcal{L}_c(Q,W) \defined P_c(S|Y) \quad \text{and} \quad \mathcal{U}_c(Q,W) \defined P_c(X|Y),
\end{equation}
where $S \to X \to Y$ is such that $P_{S,X} = Q$ and $P_{Y|X} = W$. With this notation, we also let $\mathcal{L}_f(Q,W) \defined I_f\left(P_{S,Y}\right)$, $\mathcal{L}_{\alpha}^{\rm A}(Q,W) \defined I^{\rm A}_{\alpha}\left(P_{S,Y}\right)$, $\mathcal{L}_{\alpha}^{\rm S}(Q,W) \defined I^{\rm S}_{\alpha}\left(P_{S,Y}\right)$, and, by abuse of notation, $\calL^{\rm max}_{\alpha}(Q,W) \defined \calL^{\rm max}_{\alpha}(S \to Y)$. The analogues for utility are defined in a similar way.

\begin{thm}
\label{thm:meta_upper_bounds}
Let $\hat{P}_n$ be the empirical distribution obtained from $n$ i.i.d.~samples drawn from $P$. Then, with probability at least $1-\beta$, for any $W\in\mathcal{W}$, we have
\begin{equation}
\label{eq:Meta1}
\begin{aligned}
    &|\calL(\hat{P}_n,W) - \calL(P,W)|\\
    &\leq \sqrt{\frac{2}{n}\left(|\mathcal{S}|\cdot|\mathcal{X}|-\log\beta\right)} \cdot 
    \begin{cases}
        1 & \text{when } \mathcal{L}=\calL_c,\\
        C_{f,\overline{m}_S} & \text{when } \mathcal{L}=\calL_f \text{ with } f \text{ locally Lipschitz},\\
        \frac{2\alpha}{\alpha-1} |\mathcal{S}|^{1-1/\alpha}  & \text{when } \mathcal{L}=\mathcal{L}_{\alpha}^{\rm A}\text{ with }\alpha\in(1,\infty],\\
        \frac{2\alpha+1}{(\alpha-1)\overline{m}_S^{1-1/\alpha}} & \text{when } \mathcal{L}=\mathcal{L}_{\alpha}^{\rm S}\text{ with }\alpha\in(1,\infty),\\
        \frac{2}{\min_{s} \sum_{x} \hat{P}_n(s,x)} &\text{when } \mathcal{L}=\mathcal{L}^{\rm S}_{\infty} \text{ or } \mathcal{L}=\calL^{\rm max}_{\infty},\\
        \frac{4\alpha |\mathcal{S}|^{1-1/\alpha}}{(\alpha-1)\min_{s} \sum_{x} \hat{P}_n(s,x)} &\text{when } \mathcal{L}=\calL^{\rm max}_{\alpha}\text{ with }\alpha\in(1,\infty),
    \end{cases}
\end{aligned}
\end{equation}
\begin{equation}
\label{eq:Meta2}
\begin{aligned}
    &|\calU(\hat{P}_n,W) - \calU(P,W)|\\
    &\leq \sqrt{\frac{2}{n}\left(|\mathcal{S}|\cdot|\mathcal{X}|-\log\beta\right)} \cdot 
    \begin{cases}
        1 & \text{when } \mathcal{U}=\mathcal{U}_c,\\
        C_{f,\overline{m}_X} & \text{when } \mathcal{U}=\mathcal{U}_f \text{ with } f \text{ locally Lipschitz},\\
        \frac{2\alpha}{\alpha-1} |\mathcal{X}|^{1-1/\alpha} & \text{when } \mathcal{U}=\mathcal{U}_{\alpha}^{\rm A}\text{ with }\alpha\in(1,\infty],\\
        \frac{1}{(\alpha-1)\overline{m}_X^{1-1/\alpha}} & \text{when } \mathcal{U}=\mathcal{U}_{\alpha}^{\rm S}\text{ with }\alpha\in(1,\infty),\\
        0 & \text{when } \mathcal{U}=\mathcal{U}^{\rm S}_{\infty} \text{ or } \mathcal{U}=\mathcal{U}^{\rm max}_{\alpha} \text{ with } \alpha\in(1,\infty],
    \end{cases}
\end{aligned}
\end{equation}
where $C_{f,u} = 2 K_{f,u^{-1}} + (2u^{-1} + 1) L_{f,u^{-1}}$ with $K_{g,u}$, $L_{g,u}$, $\overline{m}_S$ and $\overline{m}_X$ as defined in \eqref{eq:DefSupNorm}, \eqref{eq:DefLipschitzConstant}, \eqref{eq:DefmSbar}, and \eqref{eq:DefmXbar}, respectively.
\end{thm}

We illustrate the discrepancies of privacy-utility guarantees in \eqref{eq:Meta1} and \eqref{eq:Meta2} along with the corresponding upper bounds through a synthetic and a real-world datasets in Section~\ref{sec::Num_Exp}.

\subsubsection{Convergence of Optimal Privacy Mechanisms}

Assume there is a sequence of joint distributions $(P_n)_{n=1}^{\infty}$ converging to the true distribution $P$. The results provided in Section~\ref{Section:PointwiseRobustness}
establish that the privacy-utility guarantees for $P_n$ converge to those of $P$ as $n$ goes to infinity. Nonetheless, they cannot be used to study the convergence properties of the optimal privacy mechanisms of $P_n$ as defined in \eqref{equa:opt_privacy_mehc}. In Section~\ref{Section:convergence} we address the convergence of optimal privacy mechanisms by further exploiting some of the Lipschitz continuity properties established for the information leakage measures under consideration. To be more specific, assume that the privacy mechanism designer constructs an optimal privacy mechanism $W^*_n$ for each $P_n$, i.e., $W_n^*\in\mathcal{W}^*(P_n;\epsilon)$. We establish the convergence of the sequence of optimal privacy mechanisms $(W^*_n)_{n=1}^{\infty}$ to the set of optimal privacy mechanisms for $P$, i.e., $\mathcal{W}^*(P;\epsilon)$. Due to technical considerations, this result covers only probability of correctly guessing, $f$-information with $f$ locally Lipschitz, and Arimoto's mutual information of order $\alpha$ with $\alpha\in(1,\infty]$.

\subsubsection{Uniform Privacy Mechanisms}

In applications where privacy is a priority, a specific privacy guarantee for the true distribution $P$ may be required, even though the designer has only access to a distribution $\hat{P}$ estimated from samples. We propose the following procedure to overcome this difficulty: (a) use large deviations results to find a probabilistic upper bound, say $r$, for the distance between $\hat{P}$ and $P$; (b) design privacy mechanisms that deliver the required privacy guarantee for {\it all} distributions at distance less or equal than $r$ from $\hat{P}$. Based on this procedure, we introduce the definition of uniform privacy mechanisms in Section~\ref{Section:UniformRobustness}. We prove that \emph{optimal} uniform privacy mechanisms exist and while their design might be challenging, they can be efficiently approximated by optimal privacy mechanisms for $\hat{P}$ as defined in \eqref{equa:opt_privacy_mehc}. We finish Section~\ref{Section:UniformRobustness} establishing convergence properties of optimal uniform privacy mechanisms similar to those in Section~\ref{Section:convergence}.
\section{Discrepancy of Privacy-Utility Guarantees}
\label{Section:PointwiseRobustness}

We provide probabilistic upper bounds for the difference between the privacy-utility guarantees of the empirical and the true distributions. These upper bounds do not depend on the specific privacy mechanism. In order to simplify the exposition, we assume that privacy leakage and utility are measured using the same information metric. Nonetheless, the case when different information metrics are used can be handled in a straightforward manner. In this section, we study five different information metrics: probability of correctly guessing, $f$-information with $f$ locally Lipschitz, Arimoto's mutual information ($\alpha$-leakage) of order $\alpha>1$, Sibson's mutual information of order $\alpha>1$, and maximal $\alpha$-leakage of order $\alpha>1$.

\subsection{Probability of Correctly Guessing}

The main result of this subsection relies on the following lemma which establishes the Lipschitz continuity of the mappings $\calL_c(\cdot,W)$ and $\calU_c(\cdot,W)$. Recall that for $Q\in\mathcal{P}$ and $W\in\mathcal{W}$,
\begin{equation}
\label{eq:DefLcUc}
    \mathcal{L}_c(Q,W) \defined P_c(S|Y) \quad \text{and} \quad \mathcal{U}_c(Q,W) \defined P_c(X|Y),
\end{equation}
where $S \to X \to Y$ is such that $P_{S,X} = Q$ and $P_{Y|X} = W$.

\begin{lem}
\label{lem:rob_PCG1}
For any $Q_1,Q_2\in\mathcal{P}$ and $W\in\mathcal{W}$, we have
\begin{align}
    \label{eq:PrivacyBoundPc} |\calL_c(Q_1,W)-\calL_c(Q_2,W)| &\leq  \|Q_1-Q_2\|_1,\\
    \label{eq:UtilityBoundPc}|\calU_c(Q_1,W)-\calU_c(Q_2,W)| &\leq  \|Q_1-Q_2\|_1.
\end{align}
\end{lem}
\begin{proof}
See Appendix~\ref{Appendix:ProofPc}.
\end{proof}

In particular, when $Q_1 = \hat{P}$ is an estimate of $Q_2 = P$, the previous lemma implies that any upper bound on $\|\hat{P}-P\|_1$ translates into an upper bound for difference of the corresponding privacy-utility guarantees. The following theorem specializes this observation to the empirical distribution by means of the large deviations inequality in \eqref{equa:L1bound_Devroye}.

\begin{thm}
\label{thm:rob_PCG}
Let $\hat{P}_n$ be the empirical distribution obtained from $n$ i.i.d.~samples drawn from $P$. Then, with probability at least $1-\beta$, for any $W\in\mathcal{W}$, we have
\begin{align}
    \label{eq:MainPcL} |\calL_c(\hat{P}_n,W) - \calL_c(P,W)| &\leq \sqrt{\frac{2}{n}\left(|\mathcal{S}|\cdot|\mathcal{X}|-\log\beta\right)},\\
    |\calU_c(\hat{P}_n,W) - \calU_c(P,W)|
    &\leq  \sqrt{\frac{2}{n}\left(|\mathcal{S}|\cdot|\mathcal{X}|-\log\beta\right)}.
\end{align}
\end{thm}

In Section~\ref{Section:SimulationPc} we illustrate the fitness of the bounds in Theorem~\ref{thm:rob_PCG} when applied to ProPublica's COMPAS dataset \cite{angwin2016machine}.

\begin{rem}
Lemma~\ref{lem:rob_PCG1} and Theorem~\ref{thm:rob_PCG} illustrate the general technique we use to derive (probabilistic) upper bounds for the difference between the privacy-utility guarantees of the empirical and the true distributions. Specifically, we relate the difference between the privacy-utility guarantees of two joint distributions with their $\ell_1$ distance, and then we use large deviations results to provide upper bound for the $\ell_1$ distance between the empirical and the true distributions.
\end{rem}

\subsection{$f$-Information with $f$ Locally Lipschitz}

We start establishing the following notation. For a given function $g:[0,\infty)\to \Reals$ and $u>0$, we let
\begin{equation}
\label{eq:DefSupNorm}
    K_{g,u} \defined \sup \{|g(t)| : t\in[0,u]\}.
\end{equation}
The constant $K_{g,u}$ is the so-called supremum norm of $g$ on $[0,u]$. In addition, if $g$ is Liptschitz on $[0,u]$, we let $L_{g,u}$ be its Lipschitz constant on $[0,u]$, i.e.,
\begin{equation}
\label{eq:DefLipschitzConstant}
\begin{aligned}
    L_{g,u}\defined \min \{L\geq0 : |g(t_1)-g(t_2)| \leq L|t_1-t_2|, \forall t_1,t_2\in[0,u]\}.
\end{aligned}
\end{equation}
A function $g:[0,\infty)\to\mathbb{R}$ is called locally Lipschitz if $g$ is Lipschitz on $[0,u]$ for every $u>0$. Note that a locally Lipschitz function is not necessarily Lipschitz on $[0,\infty)$. For example, the function $g(t) = t^2$ is locally Lipschitz with $L_{g,u} = 2u$ for all $u>0$, but it is not Lipschitz on $[0,\infty)$. For any two distributions $Q_1,Q_2\in\mathcal{P}$, we denote
\begin{align}
\label{eq:DefmS} m_S &\defined \min\left\{\sum_{x\in\mathcal{X}} Q_i(s,x): s\in\mathcal{S},i\in\{1,2\}\right\},\\
\label{eq:DefmX} m_X &\defined \min\left\{\sum_{s\in\mathcal{S}} Q_i(s,x) : x\in\mathcal{X},i\in\{1,2\}\right\}.
\end{align}
Observe that $m_S$ (resp.~$m_X$) equals the probability of the least likely symbol among the marginal distributions over $\mathcal{S}$ (resp.~$\mathcal{X}$) of $Q_1$ and $Q_2$. In other words, if $(S_i,X_i)$ has joint distribution $Q_i$ for each $i\in\{1,2\}$, then
\begin{align}
    m_S &= \min\{P_{S_i}(s) : s\in\mathcal{S},i\in\{1,2\}\},\\
    m_X &= \min\{P_{X_i}(x) : x\in\mathcal{X},i\in\{1,2\}\}.
\end{align}

The following lemma serves as an analogue of Lemma~\ref{lem:rob_PCG1} for an $f$-information with $f$ locally Lipschitz. Recall that for $Q\in\mathcal{P}$ and $W\in\mathcal{W}$,
\begin{equation}
    \mathcal{L}_f(Q,W) \defined I_f(P_{S,Y}) \quad \text{and} \quad \mathcal{U}_f(Q,W) \defined I_f(P_{X,Y}),
\end{equation}
where $S \to X \to Y$ is such that $P_{S,X} = Q$ and $P_{Y|X} = W$.

\begin{lem}
\label{lem:rob_f_inf1}
If $f:[0,\infty)\to\Reals$ is locally Lipschitz, then, for any $Q_1,Q_2\in\mathcal{P}$ and $W\in\mathcal{W}$, we have
\begin{align}
    \label{eq:PrivacyfInformation} |\calL_f(Q_1,W)-\calL_f(Q_2,W)| &\leq C_{f,m_S} \|Q_1-Q_2\|_1,\\
    |\calU_f(Q_1,W)-\calU_f(Q_2,W)| &\leq C_{f,m_X} \|Q_1-Q_2\|_1,
\end{align}
where, for $u\in\{m_S,m_X\}$, $C_{f,u} \defined 2 K_{f,u^{-1}} + (2u^{-1} + 1) L_{f,u^{-1}}$.
\end{lem}
\begin{proof}
See Appendix~\ref{Appendix:Prooff}.
\end{proof}

\begin{rem}
\label{Remark:LocallyLipschitzf}
Despite the similarity between Lemmas~\ref{lem:rob_PCG1} and \ref{lem:rob_f_inf1}, the latter does not imply that the mapping $\mathcal{L}_f(\cdot,W)$ is Lipschitz continuous. Indeed, the factor $C_{f,m_S}$ depends on $Q_1$ and $Q_2$ through $m_S$. Nonetheless, Lemma~\ref{lem:rob_f_inf1} does show that the local Lipschitzianity of $f$ is bequeathed to $\mathcal{L}_f(\cdot,W)$. Specifically, for every $\delta>0$, the mapping $\mathcal{L}_f(\cdot,W)$ is Lipschitz continuous over
\begin{equation}
    \left\{Q\in\mathcal{P} : \delta \leq \min_{s\in\mathcal{S}} \sum_{x\in\mathcal{X}} Q(s,x) \right\},
\end{equation}
and its Lipschitz constant is less than or equal to $C_{f,\delta}$. Indeed, this assertion follows from the fact that $u\mapsto C_{f,u}$ is a non-increasing function, as exhibited in Appendix~\ref{Appendix:ProoffGeneralization}. Hence, any lower bound for $m_S$ translates into an upper bound for $C_{f,m_S}$ in \eqref{eq:PrivacyfInformation}. As shown below, the \emph{local Lipschitzianity} of $\mathcal{L}_f(\cdot,W)$ is enough to derive a bound similar to \eqref{eq:MainPcL}. 
A similar discussion can be held for $\mathcal{U}_{f}(\cdot,W)$.
\end{rem}

We illustrate the value of $C_{f,u}$, defined in Lemma~\ref{lem:rob_f_inf1}, through the following examples:
\begin{itemize}
    \item {\it Total variation distance.} If $f(t) = |t - 1|/2$, then $f$ is a (globally) Lipschitz function with $L_{f,u} = 1/2$ and $K_{f,u} = \max\{1,u-1\}/2$ for all $u>0$. Consequently,
    \begin{equation}
        C_{f,u}=u^{-1}\max\{1+3u/2,2-u/2\};
    \end{equation}
    
    \item {\it $\chi^2$-divergence.} If $f(t) = (t-1)^2$, then $f$ is a locally Lipschitz function with $L_{f,u} = 2 \max\{1,u-1\}$ and $K_{f,u} =\max\{1,(u-1)^2\}$ for all $u>0$. Consequently,
    \begin{equation}
        C_{f,u}=2u^{-2}\max\{2u+2u^2,3-3u\};
    \end{equation}
    
    \item {\it Hellinger divergence.} If $\displaystyle f_{\alpha}(t) = \frac{t^\alpha-1}{\alpha-1}$ with $\alpha>1$, then $f_{\alpha}$ is locally Lipschitz with $\displaystyle L_{f_{\alpha},u} = \frac{\alpha u^{\alpha-1}}{\alpha-1}$ and $K_{f,u} = \frac{1}{\alpha-1}\max\{1,u^{\alpha}-1\}$ for all $u>0$. Consequently,
    \begin{equation}
        C_{f,u}=\frac{u^{-\alpha}}{\alpha-1}\max\{2\alpha  + \alpha u+2u^{\alpha},(2\alpha+2)+\alpha u-2u^{\alpha}\}.
    \end{equation}
\end{itemize}
Note, however, that mutual information cannot be handled by Lemma~\ref{lem:rob_f_inf1} since the function $f(t)=t\log(t)$ is not locally Lipschitz. In fact, mutual information seems to have a different asymptotic behavior w.r.t.~the sample size $n$ when compared to the theorems in this paper, cf. \cite[Theorem~3]{shamir2010learning}.

Relying on Lemma~\ref{lem:rob_f_inf1} and the large deviations inequality in \eqref{equa:L1bound_Devroye}, the following theorem provides a data dependent bound for the discrepancy between the guarantees provided for the empirical and the true distributions. For $x\in\mathbb{R}$, we define $(x)_+ \defined \max\{0,x\}$.

\begin{thm}
\label{thm:main_robust}
Let $\hat{P}_n$ be the empirical distribution obtained from $n$ i.i.d.~samples drawn from $P$. If $f:[0,\infty)\to\Reals$ is locally Lipschitz, then, with probability at least $1-\beta$, for any $W\in\mathcal{W}$, we have
\begin{align}
    \label{eq:PrivacyBoundf} |\calL_f(\hat{P}_n,W)-\calL_f(P,W)| 
    &\leq C_{f,\overline{m}_S} \sqrt{\frac{2}{n}\left(|\mathcal{S}|\cdot|\mathcal{X}|-\log\beta\right)},\\
    \label{eq:UtilityBoundf} |\calU_f(\hat{P}_n,W)-\calU_f(P,W)| 
    &\leq  C_{f,\overline{m}_X} \sqrt{\frac{2}{n}\left(|\mathcal{S}|\cdot|\mathcal{X}|-\log\beta\right)},
\end{align}
where, for $u\in\{\overline{m}_S,\overline{m}_X\}$, $C_{f,u} \defined 2 K_{f,u^{-1}} + (2u^{-1} + 1) L_{f,u^{-1}}$,
\begin{align}
\overline{m}_S &\defined \Bigg(\min\left\{\sum_{x\in\mathcal{X}}\hat{P}_{n}(s,x) : s\in\mathcal{S}\right\} 
 -\sqrt{\frac{2}{n}\left(|\mathcal{S}|\cdot|\mathcal{X}|-\log\beta\right)}\Bigg)_+,\label{eq:DefmSbar}\\
\overline{m}_X &\defined \Bigg(\min\left\{\sum_{s\in\mathcal{S}}\hat{P}_{n}(s,x): x\in\mathcal{X}\right\}
-\sqrt{\frac{2}{n}\left(|\mathcal{S}|\cdot|\mathcal{X}|-\log\beta\right)}\Bigg)_+.\label{eq:DefmXbar}
\end{align}
\end{thm}

\begin{proof}
See Appendix~\ref{Appendix:ProoffGeneralization}.
\end{proof}

\begin{rem}
Observe that up to an additive factor of $(2(|\mathcal{S}|\cdot|\mathcal{X}|-\log\beta)/n)^{1/2}$, which is negligible in the large $n$ regime, $\overline{m}_S$ equals the probability of the least likely symbol of the marginal distribution over $\mathcal{S}$ of $\hat{P}_n$. As a consequence, the bound in \eqref{eq:PrivacyBoundf} is data-dependent, while the bound in \eqref{eq:MainPcL} is data-independent. This discrepancy comes mainly from the fact that $\mathcal{L}_c(\cdot,W)$ is Lipschitz continuous, as established in Lemma~\ref{lem:rob_PCG1}, while $\mathcal{L}_f(\cdot,W)$ is {\it locally Lipschitz} in the sense of Remark~\ref{Remark:LocallyLipschitzf}. A similar discussion holds for utility.
\end{rem}

In most practical scenarios, the alphabet of $X$ is significantly larger than the alphabet of $S$, e.g., $S$ might be political preference while $X$ is movie ratings, $S$ might be private household information while $X$ is smart meter data, $S$ might be gender while $X$ is a profile picture, etc. In particular, when the sample size is limited, the bounds in Theorem~\ref{thm:main_robust} might become ineffective due to the potentially small value of $\overline{m}_X$. In order to alleviate this issue, we propose the following pre-processing technique which combines the symbols of $\mathcal{X}$ with less observations in the dataset.

Given $\gamma\geq0$ and a symbol $x_0$ not belonging to $\calX$, we let $\Pi_\gamma$ be the mapping with input alphabet $\calX$, output alphabet
\begin{equation}
    \calX_\gamma \defined \{x_0\} \cup \left\{x\in\calX : \sum_{s\in\mathcal{S}} \hat{P}(s,x) \geq \gamma\right\},
\end{equation}
and determined by
\begin{equation}
\label{eq:PreProcessing}
\Pi_\gamma(x) = 
\begin{cases}
x   & \text{if } \sum_{s} \hat{P}(s,x) \geq \gamma, \\
x_0 & \text{otherwise}. 
\end{cases}
\end{equation}
The proposed pre-processing technique consists in applying this mapping to the samples $\{(s_i,x_i)\}_{i=1}^n$ to obtain the modified samples $\{(s_i,\Pi_\gamma(x_i))\}_{i=1}^n$. Clearly, this pre-processing technique improves the bounds in Theorem~\ref{thm:main_robust} by increasing $\overline{m}_X$. However, this improvement might come at a price in utility, as shown in the following proposition.

\begin{prop}
\label{prop:PreProcessingUtilityReduction}
Let $\gamma\geq0$ be given and let $\hat{P}_n$ be the empirical distribution of $n$ samples $\{(s_i,x_i)\}_{i=1}^n$. If $\hat{P}_{\gamma}$ is the empirical distribution of the modified samples $\{(s_i,\Pi_{\gamma}(x_i))\}_{i=1}^n$, then, for any $\epsilon\in \Reals$ such that $\mathcal{W}^*(\hat{P}_{\gamma};\epsilon)\neq\emptyset$,
\begin{equation}
    \mathsf{H}_{f}(\hat{P}_{\gamma};\epsilon) \leq \mathsf{H}_{f}(\hat{P}_n;\epsilon),
\end{equation}
where $\mathsf{H}_{f}$ denotes the privacy-utility function when both privacy leakage and utility are measured using $f$-information.
\end{prop}

\begin{proof}
See Appendix~\ref{Appendix:ProofPreprocessing}.
\end{proof}

The proof of Proposition~\ref{prop:PreProcessingUtilityReduction} relies on the following lemma.

\begin{lem}
\label{lem:cut_lesslikely_symbols}
Let $\gamma\geq0$ be given. If $X \to X_0 \to Y_0$ is a Markov chain with $X_0 = \Pi_\gamma(X)$, then, for every $f$-information,
\begin{equation}
\label{equa:utility_lesslikely}
I_f(P_{X,Y_0})=I_f(P_{X_0,Y_0}).
\end{equation}
\end{lem}

Observe that Lemma~\ref{lem:cut_lesslikely_symbols} is an immediate consequence of the data processing inequality (DPI) for $f$-information \cite[Remark 2.3]{polyanskiy2014lecture} and the fact that $X_0$ is a deterministic function of $X$. Indeed, if $\Pi:\mathcal{X}\to\mathcal{X}_0$ is any deterministic function and $X_0 \defined \Pi(X)$, then any random variable $Y_0$ satisfies $X_0 \to X \to Y_0$. In particular, the DPI implies that $I_f(P_{X,Y_0}) \geq I_f(P_{X_0,Y_0})$. If, in addition, $Y_0$ is a randomized function of $X_0$, i.e., $X \to X_0 \to Y_0$, then the DPI leads to $I_f(P_{X,Y_0}) \leq I_f(P_{X_0,Y_0})$.

\subsection{Arimoto's Mutual Information, Sibson's Mutual Information, and Maximal $\alpha$-Leakage}

We now establish continuity properties of Arimoto's mutual information, Sibson's mutual information, and maximal $\alpha$-leakage similar to those proved in Lemmas~\ref{lem:rob_PCG1} and \ref{lem:rob_f_inf1} for probability of correctly guessing and $f$-information, respectively. Upon these properties, we state two theorems regarding the discrepancy of privacy-utility guarantees for the information measures at hand.

Recall that for $Q\in\mathcal{P}$ and $W\in\mathcal{W}$,
\begin{equation}
    \mathcal{L}^{\rm A}_{\alpha}(Q,W) \defined I^{\rm A}_{\alpha}\left(P_{S,Y}\right) \quad \text{and} \quad \mathcal{U}^{\rm A}_{\alpha}(Q,W) \defined I^{\rm A}_{\alpha}\left(P_{X,Y}\right),
\end{equation}
where $S \to X \to Y$ is such that $P_{S,X} = Q$ and $P_{Y|X} = W$. The next lemma shows the Lipschitz continuity of $\calL^{\rm A}_{\alpha}(\cdot,W)$ and $\calU^{\rm A}_{\alpha}(\cdot,W)$.

\begin{lem}
\label{lem:LipschitzArimotoMI}
Let $\alpha\in(1,\infty]$. For any $Q_1,Q_2\in\mathcal{P}$ and $W\in\mathcal{W}$,
\begin{align}
    \label{eq:PrivacyBoundAMI} |\calL^{\rm A}_{\alpha}(Q_1,W)-\calL^{\rm A}_{\alpha}(Q_2,W)| &\leq  \frac{2\alpha}{\alpha-1} |\mathcal{S}|^{1-1/\alpha} \|Q_1-Q_2\|_1,\\
    \label{eq:UtilityBoundAMI} |\calU^{\rm A}_{\alpha}(Q_1,W)-\calU^{\rm A}_{\alpha}(Q_2,W)| &\leq  \frac{2\alpha}{\alpha-1} |\mathcal{X}|^{1-1/\alpha} \|Q_1-Q_2\|_1.
 \end{align}
\end{lem}

\begin{proof}
See Appendix~\ref{Appendix:ProofAMI}.
\end{proof}

\begin{rem}
Recall that probability of correctly guessing and Arimoto's mutual information of order $\infty$ are related through the formula
\begin{equation}
\label{eq:RelationAMIinftyPc}
    I^{\rm A}_\infty(U;V) = \log \frac{P_c(U|V)}{P_c(U)}.
\end{equation}
Observe that, despite this relation, the bound for probability of correctly guessing in \eqref{eq:PrivacyBoundPc} and the corresponding bound for Arimoto's mutual information of order $\infty$ in \eqref{eq:PrivacyBoundAMI} differ by a factor of $2|\mathcal{S}|$. As shown in the proof of Lemma~\ref{lem:LipschitzArimotoMI}, the $|\mathcal{S}|$ factor comes from the $\log$ in \eqref{eq:RelationAMIinftyPc} via the minimum in the following inequality
\begin{equation}
\phantom{a,b>0.} \quad \quad \quad \left| \log \frac{a}{b} \right| \leq \frac{|a-b|}{\min\{a,b\}}, \quad \quad \quad a,b>0.
\end{equation}
A similar argument explains the extra factor of $2|\mathcal{X}|$ appearing in \eqref{eq:UtilityBoundAMI} w.r.t.~its counterpart in \eqref{eq:UtilityBoundPc}.
\end{rem}

Now we consider Sibson's mutual information. Recall that for $Q\in\mathcal{P}$ and $W\in\mathcal{W}$,
\begin{equation}
    \mathcal{L}^{\rm S}_{\alpha}(Q,W) \defined I^{\rm S}_{\alpha}\left(P_{S,Y}\right) \quad \text{and} \quad \mathcal{U}^{\rm S}_{\alpha}(Q,W) \defined I^{\rm S}_{\alpha}\left(P_{X,Y}\right),
\end{equation}
where $S \to X \to Y$ is such that $P_{S,X} = Q$ and $P_{Y|X} = W$. The following lemma establishes that the mappings $\mathcal{L}^{\rm S}_{\alpha}(\cdot,W)$ and $\mathcal{U}^{\rm S}_{\alpha}(\cdot,W)$ have similar continuity properties to those of $\mathcal{L}_{f}(\cdot,W)$ and $\mathcal{U}_{f}(\cdot,W)$ with $f$ locally Lipschitz.

\begin{lem}
\label{lem:LipschitzSibsonMI}
For $Q_1,Q_2\in\mathcal{P}$, we define $m_S$ and $m_X$ as in \eqref{eq:DefmS} and \eqref{eq:DefmX}, respectively. If $\alpha\in(1,\infty)$, then, for any $W\in\mathcal{W}$,
\begin{align}
    \label{eq:PrivacyBoundSMI} |\calL^{\rm S}_{\alpha}(Q_1,W)-\calL^{\rm S}_{\alpha}(Q_2,W)| &\leq  \frac{2\alpha+1}{\alpha-1} \frac{\|Q_1 - Q_2\|_1}{m_S^{1-1/\alpha}},\\
    \label{eq:UtilityBoundSMI} |\calU^{\rm S}_{\alpha}(Q_1,W)-\calU^{\rm S}_{\alpha}(Q_2,W)| &\leq  \frac{1}{\alpha-1} \frac{\|Q_1-Q_2\|_1}{m_X^{1-1/\alpha}}.
\end{align}
If $\alpha=\infty$, then, for any $W\in\mathcal{W}$,
\begin{align}
   \label{eq:PrivacyLeakageBoundMaximalLeakage} |\calL^{\rm S}_{\infty}(Q_1,W)-\calL^{\rm S}_{\infty}(Q_2,W)| &\leq  \frac{2\|Q_1-Q_2\|_1}{\min_s \sum_x Q_1(s,x)},\\  \label{eq:UtilityBoundMaximalLeakage} |\calU^{\rm S}_{\infty}(Q_1,W)-\calU^{\rm S}_{\infty}(Q_2,W)| &= 0.
\end{align}
\end{lem}

\begin{proof}
See Appendix~\ref{Appendix:ProofSMI}.
\end{proof}

\begin{rem}
A straightforward manipulation shows that
\begin{align}
    \label{eq:ExpArimotoMI} \exp\left(I^{\rm A}_\infty(P_{U,V})\right) &= \sum_{v\in\mathcal{V}} \max_{u\in\mathcal{U}} \frac{P_U(u)}{\max_{u'} P_U(u')} P_{V|U}(v|u),\\
    \label{eq:ExpSibsonMI} \exp\left(I^{\rm S}_\infty(P_{U,V})\right) &= \sum_{v\in\mathcal{V}} \max_{u\in\mathcal{U}} P_{V|U}(v|u).
\end{align}
From \eqref{eq:ExpArimotoMI} and \eqref{eq:ExpSibsonMI}, we observe that Sibson's mutual information of order infinity is independent of the input distribution $P_U$, while Arimoto's mutual information of the same order is not. By comparing \eqref{eq:PrivacyBoundAMI} and \eqref{eq:PrivacyLeakageBoundMaximalLeakage}, we also observe that the privacy leakage bound for Sibson's mutual information of order infinity \emph{does} depend on the probability of the least likely symbol of the input\footnote{In recent work \cite{issa2018operational}, Issa \etal established a bound similar to \eqref{eq:ThmRobustSMIInfty} which also depends on the probability of the least likely symbol of the input. Thus, this dependency seems unavoidable at the moment.}, while its Arimoto's counterpart does not. These seemingly contradicting facts have a rather intuitive cause: it is hard to estimate $P_{V|U}(\cdot|u)$ reliably if $P_U(u)$ is small. Observe that while $P_{V|U}(v|u)$ appears in both \eqref{eq:ExpArimotoMI} and \eqref{eq:ExpSibsonMI}, the factor $P_U(u) / \max_{u'} P_U(u')$ in \eqref{eq:ExpArimotoMI} makes Arimoto's mutual information of order infinity less dependent on symbols $u$ with small probabilities $P_U(u)$. The difficulty in estimating Sibson's mutual information in comparison to its Arimoto counterpart is also natural from a privacy perspective. As proved by Issa \etal in \cite{issa2016operational}, Sibson's mutual information guarantees privacy for \emph{any} (possibly randomized) function of $S$, while Arimoto's mutual information guarantees privacy only for $S$ itself.
\end{rem}

We end up this section showing that maximal $\alpha$-leakage ($\alpha>1$) behaves in a similar way to Sibson's mutual information of order $\infty$. Recall the notation in Section~\ref{Subsection:Notation} and that, for $Q\in\mathcal{P}$ and $W\in\mathcal{W}$,
\begin{align}
    \mathcal{L}^{\rm max}_\alpha(Q,W) \defined \sup_{P_{\tilde{S}}} I^{\rm A}_\alpha\left(P_{\tilde{S}}\cdot (P_{X|S} W)\right)\quad \text{and}\quad
    \mathcal{U}^{\rm max}_\alpha(Q,W) \defined \sup_{P_{\tilde{X}}} I^{\rm A}_\alpha\left(P_{\tilde{X}}\cdot W\right),
\end{align}
where $S \to X \to Y$ is such that $P_{S,X} = Q$ and $P_{Y|X} = W$.

\begin{lem}
\label{Lemma:LipschitzMaximalAlphaLeakage}
Let $\alpha\in(1,\infty)$. For any $Q_1,Q_2\in\mathcal{P}$ and $W\in\mathcal{W}$,
\begin{align}
    &|\calL^{\rm max}_{\alpha}(Q_1,W)-\calL^{\rm max}_{\alpha}(Q_2,W)|
    \leq \frac{4\alpha}{\alpha-1} \frac{|\mathcal{S}|^{1-1/\alpha}}{\min_s \sum_x Q_1(s,x)} \|Q_1-Q_2\|_1, \label{eq:MaxAlphaLeakageLBound}\\
    &|\calU^{\rm max}_{\alpha}(Q_1,W)-\calU^{\rm max}_{\alpha}(Q_2,W)| = 0.\label{eq:MaxAlphaLeakageUBound}
\end{align}
\end{lem}

\begin{proof}
See Appendix~\ref{Appendix:ProofMAL}.
\end{proof}

We summarize the probabilistic upper bounds derived in this subsection in the following two theorems. They follow immediately from Lemma~\ref{lem:LipschitzArimotoMI}, \ref{lem:LipschitzSibsonMI}, \ref{Lemma:LipschitzMaximalAlphaLeakage} and the large deviations inequality in \eqref{equa:L1bound_Devroye}.

\begin{thm}
\label{thm:bounds_AMI_SMI}
Let $\hat{P}_n$ be the empirical distribution obtained from $n$ i.i.d.~samples drawn from $P$. Then, with probability at least $1-\beta$, for any $W\in\mathcal{W}$, we have
\begin{align}
    &|\calL(\hat{P}_n,W) - \calL(P,W)|
    \leq \sqrt{\frac{2}{n}\left(|\mathcal{S}|\cdot|\mathcal{X}|-\log\beta\right)} \cdot
    \begin{cases}
        \frac{2\alpha}{\alpha-1} |\mathcal{S}|^{1-1/\alpha}  & \text{when } \mathcal{L}=\mathcal{L}_{\alpha}^{\rm A}\text{ with }\alpha\in(1,\infty],\\
        \frac{2\alpha+1}{(\alpha-1)\overline{m}_S^{1-1/\alpha}} & \text{when } \mathcal{L}=\mathcal{L}_{\alpha}^{\rm S}\text{ with }\alpha\in(1,\infty),
    \end{cases}
\end{align}
\begin{align}
    &|\calU(\hat{P}_n,W) - \calU(P,W)|
    \leq \sqrt{\frac{2}{n}\left(|\mathcal{S}|\cdot|\mathcal{X}|-\log\beta\right)} \cdot
    \begin{cases}
        \frac{2\alpha}{\alpha-1} |\mathcal{X}|^{1-1/\alpha} & \text{when } \mathcal{U}=\mathcal{U}_{\alpha}^{\rm A}\text{ with }\alpha\in(1,\infty],\\
        \frac{1}{(\alpha-1)\overline{m}_X^{1-1/\alpha}} & \text{when } \mathcal{U}=\mathcal{U}_{\alpha}^{\rm S}\text{ with }\alpha\in(1,\infty),
    \end{cases}
\end{align}
where $\overline{m}_S$ and $\overline{m}_X$ are defined in \eqref{eq:DefmSbar} and \eqref{eq:DefmXbar}, respectively.
\end{thm}
\begin{thm}
\label{Thm:bound_Sinf_ML}
Let $\hat{P}_n$ be the empirical distribution obtained from $n$ i.i.d.~samples drawn from $P$. Then, with probability at least $1-\beta$, for any $W\in\mathcal{W}$, we have
\begin{equation}
\label{eq:ThmRobustSMIInfty}
\begin{aligned}
    &|\mathcal{L}(\hat{P}_n,W)-\mathcal{L}(P,W)| 
    \leq \frac{2\sqrt{\frac{2}{n}\left(|\mathcal{S}|\cdot|\mathcal{X}|-\log\beta\right)}}{\min_{s} \sum_{x} \hat{P}_n(s,x)} 
    \cdot
    \begin{cases}
        1 &\text{when } \mathcal{L}=\mathcal{L}^{\rm S}_{\infty} \text{ or } \mathcal{L}=\calL^{\rm max}_{\infty},\\
        \frac{2\alpha |\mathcal{S}|^{1-1/\alpha}}{\alpha-1} &\text{when } \mathcal{L}=\calL^{\rm max}_{\alpha}\text{ with }\alpha\in(1,\infty),
    \end{cases}
\end{aligned}
\end{equation}
and $|\mathcal{U}(\hat{P}_n,W)-\mathcal{U}(P,W)|=0$ when $\mathcal{U}=\mathcal{U}^{\rm S}_{\infty}$ or $\mathcal{U}=\mathcal{U}^{\rm max}_{\alpha}$ with $\alpha\in(1,\infty]$.
\end{thm}

\section{Convergence of Optimal Privacy Mechanisms}
\label{Section:convergence}

Now we study the consistency of the optimal privacy mechanisms using, among other results, the (local) Lipschitz continuity of the mappings $\mathcal{L}(\cdot,W)$ and $\mathcal{U}(\cdot,W)$ established in the previous section.

Assume that the (privacy) mechanism designer is given an $\epsilon\in\mathbb{R}$ and a sequence of joint distributions $(P_n)_{n=1}^{\infty}$ converging to the true distribution $P$. Furthermore, the designer constructs an optimal $\epsilon$-private mechanism $W^*_n$ for each $P_n$. In this section, we analyze some convergence properties of the sequence of optimal privacy mechanisms $(W^*_n)_{n=1}^\infty$.

In the sequel, we assume that prior knowledge about the true distribution $P$ might be available. Accordingly, we let $\mathcal{Q}\subseteq\mathcal{P}$ be the set of all joint distributions compatible with such prior knowledge. Note that when no prior knowledge is available, we simply let $\mathcal{Q} = \mathcal{P}$. The main results of this and the following section are established under the following conditions. 

There exist a closed set $\mathcal{Q}\subseteq\mathcal{P}$ and $N\in\mathbb{N}$ such that
\begin{enumerate}[label=(\textbf{C.}\arabic*),ref=\arabic*]
    \item \label{cond:cont}
    \emph{Continuity.} For every $Q\in\mathcal{Q}$, the mappings $\mathcal{L}(Q,\cdot)$ and $\mathcal{U}(Q,\cdot)$ are continuous over
    \begin{equation}
        \mathcal{W}_N \defined \mathcal{W} \cap \mathbb{R}^{|\mathcal{X}|\times N},
    \end{equation}
    where $\mathcal{W}$, as defined in \eqref{eq:DefCalW}, denotes the set of all privacy mechanisms. Furthermore, the mapping $\mathsf{H}(Q;\cdot)$ is continuous over $[\epsilon_{\min}(Q),\infty)$, where
    \begin{align}
    \label{equa:epsilon_min_Def}
        \epsilon_{\min}(Q) &\defined \inf \left\{\mathcal{L}(Q,W) :\  W\in\mathcal{W}\right\}.
    \end{align}
    
    \item \label{cond:lipcont}
    \emph{Lipschitz Continuity.} There exist positive constants $C_L$ and $C_U$ such that, for all $Q_1,Q_2\in\mathcal{Q}$ and $W\in\mathcal{W}_N$,
    \begin{equation}
    \label{eq:PrivacyLemmaContinuity}
        |\mathcal{L}(Q_1,W)-\mathcal{L}(Q_2,W)| \leq C_L \|Q_1-Q_2\|_1,
    \end{equation}
    \begin{equation}
    \label{eq:UtilityLemmaContinuity}
        |\mathcal{U}(Q_1,W)-\mathcal{U}(Q_2,W)| \leq C_U \|Q_1-Q_2\|_1.
    \end{equation}
    
    \item \label{cond:comp}
    \emph{Support.} For each $Q\in\mathcal{Q}$ and $\epsilon\geq\epsilon_{\min}(Q)$, the intersection $\mathcal{W}^*(Q;\epsilon)\cap\mathcal{W}_N$ is not empty.
\end{enumerate}

These three conditions might seem restrictive at a first glance. Nonetheless, as shown in Section~\ref{Section:PcIfIAISTechnicalConditions} below, they are satisfied by probability of correctly guessing, $f$-information with $f$ locally Lipschitz, and Arimoto's mutual information of order $\alpha$ with $\alpha\in(1,\infty]$.

Note that condition (\textbf{C.}\ref{cond:cont}) is a continuity requirement. Specifically, it requires the privacy leakage and utility functions to be continuous w.r.t.~the privacy mechanism $W$. Also, it requires the privacy-utility function to be continuous w.r.t.~the privacy parameter $\epsilon$. Similarly, condition (\textbf{C.}\ref{cond:lipcont}) requires the privacy leakage and utility functions to be Lipschitz continuous w.r.t.~the joint distribution $Q$. Observe that the constants $C_L$ and $C_U$ in \eqref{eq:PrivacyLemmaContinuity} and \eqref{eq:UtilityLemmaContinuity}, respectively, do not depend on the joint distributions $Q_1$ and $Q_2$ neither on the privacy mechanism $W$. Nonetheless, these constants may depend on $\mathcal{Q}$ and $N$.

\begin{rem}
\label{rem:sup_max_Wn}
Observe that the set $\mathcal{W}_N$ is in correspondence with the set of all privacy mechanisms from $\mathcal{X}$ to $\{1,\ldots,N\}$. In particular, condition (\textbf{C.}\ref{cond:comp}) requires that, for each $Q\in\mathcal{Q}$ and $\epsilon\geq\epsilon_{\min}(Q)$, there exists an optimal $\epsilon$-private mechanism for $Q$ supported over $\{1,\ldots,N\}$. As a consequence,
\begin{equation}
\label{eq:HC1}
    \mathsf{H}(Q;\epsilon) =\max_{W\in\mathcal{D}_N(Q;\epsilon)} \mathcal{U}(Q,W),
\end{equation}
where $\mathcal{D}_N(Q;\epsilon) \defined \{ W\in\mathcal{W}_N : \mathcal{L}(Q,W) \leq \epsilon\}$. Furthermore, the intersection of $\mathcal{W}_N$ and all optimal privacy mechanisms is defined as 
\begin{equation}
\label{eq:opt_pri_mech_W_N_star_defn}
\begin{aligned}
    \mathcal{W}_N^*(P;\epsilon)
    \defined \{W\in \mathcal{W}_N : \mathcal{L}(P,W)\leq \epsilon,\  \mathcal{U}(P,W)=\mathsf{H}(P;\epsilon)\}.
\end{aligned}
\end{equation}

By comparing \eqref{eq:DefH} and \eqref{eq:HC1}, we can observe that condition (\textbf{C.}\ref{cond:comp}) allows us to replace the unbounded set $\mathcal{W}$ with the compact space $\mathcal{W}_N$ in the optimization defining the privacy-utility function $\mathsf{H}$. This technical difference is crucial in the proofs of the subsequent results. In a similar spirit, observe that condition (\textbf{C.}\ref{cond:comp}) also implies that
\begin{equation}
\label{eq:epsilonminReduced}
    \epsilon_{\min}(Q) = \min \left\{\mathcal{L}(Q,W) :\  W\in\mathcal{W}_N\right\}.
\end{equation}
\end{rem}

A two-variable function might be continuous in each argument without being jointly continuous, see, e.g., \cite[Ex.~4.7]{rudin1976principles}. The following lemma is an immediate, yet important, consequence of conditions (\textbf{C.}\ref{cond:cont}--\ref{cond:lipcont}), as it establishes the joint continuity of the privacy leakage and utility functions.

\begin{lem}
\label{lem:cont_UandL}
Assume that conditions (\textbf{C.}\ref{cond:cont}--\ref{cond:lipcont}) hold true for a given closed set $\mathcal{Q}\subseteq\mathcal{P}$ and a given $N\in\mathbb{N}$. The functions $\mathcal{L}(\cdot,\cdot)$ and $\mathcal{U}(\cdot,\cdot)$ are continuous over $\mathcal{Q}\times\mathcal{W}_N$.
\end{lem}

\begin{proof}
See Appendix~\ref{Appendix:Proofcont_UandL}.
\end{proof}

Building upon Lemma~\ref{lem:cont_UandL}, we establish the following proposition which plays a key role in the proofs of the main results of this section.

\begin{prop}
\label{Prop:ContinuityH}
Assume that conditions (\textbf{C.}\ref{cond:cont}--\ref{cond:comp}) hold true for a given closed set $\mathcal{Q}\subseteq\mathcal{P}$ and a given $N\in\mathbb{N}$. For any given $\epsilon\in\mathbb{R}$, the mapping $\mathsf{H}(\cdot;\epsilon)$ is continuous over the set $\{Q\in\mathcal{Q}: \epsilon_{\min}(Q)< \epsilon\}$.
\end{prop}

\begin{proof}
See Appendix~\ref{Appendix:ProofContinuousH_general}.
\end{proof}

\begin{rem}
The mapping $\mathsf{H}(\cdot;\epsilon)$ might {\it not} be continuous over $\{Q\in\mathcal{Q}: \epsilon_{\min}(Q)\leq \epsilon\}$. In order to prove this claim, consider the following example.
\end{rem}

\begin{example}
\label{Example:LackContinuity}
For each $\zeta\in[0,1/2]$, let $Q_\zeta$ be the joint distribution of $(S,X_\zeta)$ where $S \sim \text{Ber}(1/2)$ and $P_{X_\zeta|S} = \text{BSC}(\zeta)$. By definition, $\mathsf{H}_c$ denotes the privacy-utility function when both privacy leakage and utility are measured using probability of correctly guessing. By Theorem~2 in \cite{asoodeh2017estimation}, we have that, for every $\zeta\in[0,1/2)$,
\begin{equation}
    \epsilon_{\min}(Q_\zeta)=1/2 \quad \text{and} \quad \mathsf{H}_c(Q_\zeta;1/2) = 1/2.
\end{equation}
Since $S$ and $X_{1/2}$ are independent, we also have that
\begin{equation}
    \epsilon_{\min}(Q_{1/2}) = 1/2 \quad \text{and} \quad \mathsf{H}_c(Q_{1/2};1/2) = 1.
\end{equation}
Therefore, we have that $\displaystyle \lim_{\zeta\uparrow1/2} \mathsf{H}_c(Q_\zeta;1/2) \neq \mathsf{H}_c(Q_{1/2};1/2)$ although $\displaystyle \lim_{\zeta\uparrow1/2} Q_\zeta = Q_{1/2}$.
\end{example}

Despite that $\mathsf{H}(\cdot;\epsilon)$ might be discontinuous at the boundary $\{Q\in\mathcal{Q}: \epsilon_{\min}(Q) = \epsilon\}$, as presented in Example~\ref{Example:LackContinuity}, in some situations such pathological behavior is not exhibited. In such cases, the following corollary provides conditions for which the pointwise convergence established in Proposition~\ref{Prop:ContinuityH} upgrades into uniform convergence. This technical result will be useful in explaining the numerical experiments in Section~\ref{sec::Num_Exp}.

\begin{cor}
\label{cor::unif_conv_PUF}
Assume that conditions (\textbf{C.}\ref{cond:cont}--\ref{cond:comp}) hold true for a given closed set $\mathcal{Q}\subseteq\mathcal{P}$ and a given $N\in\mathbb{N}$. In addition, assume that there exists $\epsilon_0\in\mathbb{R}$ such that $\epsilon_{\min}(Q) = \epsilon_0$ for all $Q\in\mathcal{Q}$. If $P_n \in \mathcal{Q}$ for each $n\in\mathbb{N}$, $\lim_n P_n = P$, and $\lim_{n}\mathsf{H}(P_n;\epsilon_0)=\mathsf{H}(P;\epsilon_0)$, then $\mathsf{H}(P_n;\cdot)$ converges uniformly to $\mathsf{H}(P;\cdot)$ over $[\epsilon_0,\infty)$, i.e.,
\begin{equation}
    \lim_{n\to\infty} \sup_{\epsilon\in[\epsilon_0,\infty)} |\mathsf{H}(P_n;\epsilon)-\mathsf{H}(P,\epsilon)| = 0.
\end{equation}
\end{cor}
\begin{proof}
See Appendix~\ref{Appendix:Proof_cor::unif_conv_PUF}.
\end{proof}

The next theorem shows that if $W^*_n$ is an optimal $\epsilon$-private mechanism for $P_n$ and $(P_n)_{n=1}^\infty$ converges to $P$, then the sequence of optimal $\epsilon$-private mechanisms $(W^*_n)_{n=1}^\infty$ converges to the set of optimal $\epsilon$-private mechanisms for $P$. For $W\in\mathcal{W}_N$ and a subset $\mathcal{W}'\subseteq\mathcal{W}_N$, the distance between $W$ and $\mathcal{W}'$ is defined as
\begin{align}
    \mathsf{dist}(W,\mathcal{W}')\defined
    \inf\{\|W-W'\|_1:W'\in\mathcal{W}'\},
\end{align}
where $\displaystyle \|W-W'\|_1 = \sum_{x\in\mathcal{X}} \sum_{y\in\mathcal{Y}} |W(x,y) - W'(x,y)|$.

\begin{thm}
\label{thm:optimaltooptimal_mainthm}
Assume that conditions (\textbf{C.}\ref{cond:cont}--\ref{cond:comp}) hold true for a given closed set $\mathcal{Q}\subseteq\mathcal{P}$ and a given $N\in\mathbb{N}$. Let $\epsilon\in\mathbb{R}$ be given. If $P_n \in \mathcal{Q}$ for each $n\in\mathbb{N}$, $\lim_n P_n = P$, and $\epsilon_{\min}(P)<\epsilon$, then, for any sequence $(W^*_n)_{n=1}^\infty\subset\mathcal{W}_N$ such that $W^*_n \in \mathcal{W}_N^*(P_n;\epsilon)$,
\begin{align}
\label{eq:optimaltooptimal_mainthm}
    \lim_{n\to\infty}\mathsf{dist}(W^*_n,\mathcal{W}_N^*(P;\epsilon)) = 0.
\end{align}
Furthermore, if $P_n=\hat{P}_n$ is the empirical distribution obtained from $n$ i.i.d.~samples drawn from $P$, then
\begin{align}
    \Pr\left(\lim_{n\to\infty}\mathsf{dist}(W^*_n,\mathcal{W}_N^*(P;\epsilon)) = 0\right)=1.
\end{align}
\end{thm}

\begin{proof}
See Appendix~\ref{Appendix:Proofoptimaltooptimal_mainthm}.
\end{proof}

\begin{rem}
Under the assumptions in Theorem~\ref{thm:optimaltooptimal_mainthm}, it may be possible that $\mathcal{W}_N^*(P_n;\epsilon)$ is empty for some values of $n$. Nonetheless, under the same assumptions, $\mathcal{W}_N^*(P_n;\epsilon)$ is non-empty for $n$ sufficiently large. This fact can be easily derived from the continuity of the mapping $\epsilon_{\min}(\cdot)$ which is the content of Lemma~\ref{lem:epsilonmin} in Appendix~\ref{Appendix:Proofoptimaltooptimal_mainthm}.
\end{rem}

The following corollary follows directly from Theorem~\ref{thm:optimaltooptimal_mainthm} and the compactness of $\mathcal{W}_N^*(P;\epsilon)$ established in Lemma~\ref{lem:W_star_compact} in Appendix~\ref{Appendix:Proofoptimaltooptimal_mainthm}. It shows that the limit of optimal $\epsilon$-private mechanisms is also an optimal $\epsilon$-private mechanism.

\begin{cor}
\label{Corollary:OptimaltoOptimal}
In addition to the assumptions in Theorem~\ref{thm:optimaltooptimal_mainthm}, assume that $\lim_n W^*_n = W_0$ for some $W_0\in\mathcal{W}_N$. Then, $W_0\in \mathcal{W}_N^*(P;\epsilon)$.
\end{cor}

It is important to note that Theorem~\ref{thm:optimaltooptimal_mainthm} does not imply that the sequence of optimal privacy mechanisms $(W^*_n)_{n=1}^\infty$ converges to a privacy mechanism. Indeed, it has been noticed in simulation that a small perturbation to the joint distribution might lead to a big change to the optimal privacy mechanism returned by some optimization algorithms. The following theorem, whose proof relies on standard results from {\it set-valued analysis} \cite{aubin2009set}, shows that given a sequence $(P_n)_{n=1}^\infty$ with $\lim_n P_n = P$, it is possible to choose $W^*_n\in \mathcal{W}_N^*(P_n;\epsilon)$ such that $(W^*_n)_{n=1}^\infty$ is convergent. However, we prove this property only for a residual set which, by definition, is a countable intersection of dense open subsets of $\mathcal{Q} \subset \mathbb{R}^{|\mathcal{S}|\times|\mathcal{X}|}$.

\begin{thm}
\label{thm:setvalued_optimalprivacymechanism}
Assume that conditions (\textbf{C.}\ref{cond:cont}--\ref{cond:comp}) hold true for a given closed set $\mathcal{Q}\subseteq\mathcal{P}$ and a given $N\in\mathbb{N}$. For any given $\epsilon\in\mathbb{R}$ and any $\delta>0$, there exists a residual set $\mathcal{Q}' \subseteq \{Q\in\mathcal{Q}:\epsilon_{\min}(Q) + \delta \leq \epsilon\}$ which satisfies that:
\begin{itemize}[leftmargin=10pt, rightmargin=10pt]
    \item[] For any joint distribution $Q_0\in \mathcal{Q}'$, any sequence of joint distributions $(Q_n)_{n=1}^\infty \subset \{Q\in\mathcal{Q}:\epsilon_{\min}(Q) + \delta \leq \epsilon\}$ with $\lim_n Q_n = Q_0$, and any optimal privacy mechanism $W^*_0\in \mathcal{W}_N^*(Q_0;\epsilon)$, there exists a sequence of privacy mechanisms $(W^*_n)_{n=1}^\infty$ such that $W^*_n\in \mathcal{W}_N^*(Q_n;\epsilon)$ for each $n\in\mathbb{N}$ and $\lim_n W^*_n = W^*_0$.
\end{itemize}
\end{thm}

\begin{proof}
See Appendix~\ref{Appendix:setvalued_optimalprivacymechanism}.
\end{proof}

\begin{rem}
Baire's theorem states that any residual subset of a complete metric space is dense, see, e.g., \cite[Thm.~5.6]{rudin1987real}. In this sense, residual sets are considered to be {\it big} in topological terms. However, residual sets might be negligible in measure theoretic terms, see, e.g., \cite[Exercise~5.3.31]{folland1984real}.
\end{rem}

\subsection{Continuity and Compactness Conditions}
\label{Section:PcIfIAISTechnicalConditions}

In this section we show that probability of correctly guessing, $f$-information with $f$ locally Lipschitz, Arimoto's mutual information of order $\alpha$ with $\alpha\in(1,\infty]$ satisfy conditions (\textbf{C.}\ref{cond:cont}--\ref{cond:comp}).

\subsubsection{Probability of Correctly Guessing}
\label{Subsection:PcC123}

We choose $\mathcal{Q} = \mathcal{P}$ and $\mathcal{Y}=\{1,\ldots,N\}$ with $N \geq |\mathcal{X}|+1$. Note that $\mathcal{Q}=\mathcal{P}$ corresponds to the case where no prior information about the true distribution $P$ is available. Recall that, by definition,
\begin{align}
\mathcal{L}_c(Q,W) &= \sum_{y\in\mathcal{Y}} \max_{s\in\mathcal{S}} \sum_{x\in\mathcal{X}} Q(s,x) W(x,y),\\
\mathcal{U}_c(Q,W) &= \sum_{y\in\mathcal{Y}} \max_{x\in\mathcal{X}} \sum_{s\in\mathcal{S}} Q(s,x) W(x,y).
\end{align}
Since the maximum and the sum of continuous functions are continuous functions as well, we have that, for every $Q\in\mathcal{P}$, the mappings $\mathcal{L}_c(Q,\cdot)$ and $\mathcal{U}_c(Q,\cdot)$ are continuous over $\mathcal{W}_N$. In \cite{asoodeh2017estimation}, Asoodeh \etal showed that $\mathsf{H}_c(Q;\cdot)$, the privacy-utility function associated to $\mathcal{L}_c$ and $\mathcal{U}_c$, is continuous and piecewise linear over\footnote{Indeed, in this setting $\epsilon_{\min}(Q) = \max_s \sum_x Q(s,x)$.} $[\epsilon_{\min}(Q),\infty)$. Therefore condition (\textbf{C.}\ref{cond:cont}) is satisfied. Furthermore, Lemma~\ref{lem:rob_PCG1} shows that, for any $Q_1,Q_2\in\mathcal{P}$ and $W\in\mathcal{W}$,
\begin{align}
    |\calL_c(Q_1,W)-\calL_c(Q_2,W)| &\leq  \|Q_1-Q_2\|_1,\\
    |\calU_c(Q_1,W)-\calU_c(Q_2,W)| &\leq  \|Q_1-Q_2\|_1,
\end{align}
which implies that condition (\textbf{C.}\ref{cond:lipcont}) is satisfied with $C_L=1$ and $C_U=1$. It was established in \cite{asoodeh2017estimation} that, for every $\epsilon\geq\epsilon_{\min}(Q)$, there is always an optimal $\epsilon$-private mechanism using at most $|\mathcal{X}|+1$ symbols. As a consequence, condition (\textbf{C.}\ref{cond:comp}) holds true as $N\geq|\mathcal{X}|+1$ by assumption.

\subsubsection{$f$-Information with $f$ Locally Lipschitz}

Assume that the function $f:[0,\infty)\to\mathbb{R}$ is locally Lipschitz and convex with $f(1) = 0$. Given $\gamma>0$, we choose $N\geq |\mathcal{X}|+1$ and
\begin{equation}
\label{eq:DefPgamma}
\mathcal{Q} = \left\{Q\in\mathcal{P} : \gamma\leq\min_{s\in\mathcal{S}} \sum_{x\in\mathcal{X}} Q(s,x), \gamma\leq\min_{x\in\mathcal{X}} \sum_{s\in\mathcal{S}} Q(s,x)\right\}.
\end{equation}
In this case, the prior information about the true distribution $P = P_{S,X}$ comes in the form of the assumption that the probability mass functions of $S$ and $X$ are bounded away from 0. Recall that, by definition,
\begin{align}
    \mathcal{L}_f(Q,W) &= \sum_{s\in\mathcal{S}} \sum_{y\in\mathcal{Y}} P_{S}(s)P_{Y}(y) f\left(\frac{P_{S,Y}(s,y)}{P_S(s)P_Y(y)}\right),\\
    \mathcal{U}_f(Q,W) &= \sum_{x\in\mathcal{X}} \sum_{y\in\mathcal{Y}} P_{X}(x)P_{Y}(y) f\left(\frac{P_{X,Y}(x,y)}{P_X(x)P_Y(y)}\right),
\end{align}
where $S \to X \to Y$ is such that $P_{S,X} = Q$ and $P_{Y|X} = W$. Note that if $Q\in \mathcal{Q}$, then, for all $(s,x,y)\in\mathcal{S}\times\mathcal{X}\times\mathcal{Y}$,
\begin{equation}
\max\left\{\frac{P_{S,Y}(s,y)}{P_S(s)P_Y(y)},\frac{P_{X,Y}(x,y)}{P_X(x)P_Y(y)}\right\} \leq \gamma^{-1}.
\end{equation}
Upon this fact, it is straightforward to verify that the mappings $\mathcal{L}_f(Q,\cdot)$ and $\mathcal{U}_f(Q,\cdot)$ are continuous over $\mathcal{W}_N$. In \cite{hsu2018generalizing}, Hsu \etal established the continuity of $\mathsf{H}_f(Q;\cdot)$ over $[0,\infty)$. Hence, condition (\textbf{C.}\ref{cond:cont}) is satisfied. Recall the definitions of $K_{g,u}$ and $L_{g,u}$ in \eqref{eq:DefSupNorm} and \eqref{eq:DefLipschitzConstant}, respectively. Lemma~\ref{lem:rob_f_inf1} shows that, for any $Q_1,Q_2\in\mathcal{Q}$ and $W\in\mathcal{W}$,
\begin{equation}
\begin{aligned}
    |\calL_f(Q_1,W)-\calL_f(Q_2,W)| 
    \leq 
    (2 K_{f,\gamma^{-1}} + (2\gamma^{-1}+1)L_{f,\gamma^{-1}}) \|Q_1-Q_2\|_1,
\end{aligned}
\end{equation}
\begin{equation}
\begin{aligned}
    |\calU_f(Q_1,W)-\calU_f(Q_2,W)| 
    \leq 
    (2 K_{f,\gamma^{-1}} + (2\gamma^{-1}+1)L_{f,\gamma^{-1}})
    \|Q_1-Q_2\|_1,
\end{aligned}
\end{equation}
which implies that condition (\textbf{C.}\ref{cond:lipcont}) is satisfied with
\begin{equation}
    C_L = C_U = 2 K_{f,\gamma^{-1}} + (2\gamma^{-1}+1)L_{f,\gamma^{-1}}.
\end{equation}
For example, if $f(x) = |x-1|$, then $C_L = C_U \leq 4\gamma^{-1}+1$; and if $f(x)=x^2-1$, then $C_L = C_U \leq 8\gamma^{-2}$. As with probability of correctly guessing, there is always an optimal $\epsilon$-private mechanism using at most $|\mathcal{X}|+1$ symbols \cite{hsu2018generalizing}. From this fact, condition (\textbf{C.}\ref{cond:comp}) follows immediately.

\subsubsection{Arimoto's Mutual Information}

Let $\alpha\in(1,\infty]$. We choose $N\geq|\mathcal{X}|+1$ and $\mathcal{Q} = \mathcal{P}$, i.e., no prior information is assumed. As with the previous information measures, the continuity of the mappings $\mathcal{L}^{\rm A}_\alpha(Q,\cdot)$ and $\mathcal{U}^{\rm A}_\alpha(Q,\cdot)$ is evident. We denote the privacy-utility function by $\mathsf{H}^{\rm A}_\alpha$ when both privacy leakage and utility are measured using Arimoto's mutual information of order $\alpha$. Note that the graph of $\mathsf{H}^{\rm A}_\alpha(Q;\cdot)$ corresponds to the upper boundary of the set
\begin{equation}
    \mathcal{A} \defined \left\{(\mathcal{L}^{\rm A}_\alpha(Q,W),\mathcal{U}^{\rm A}_\alpha(Q,W)) : W\in\mathcal{W}\right\}.
\end{equation}
More specifically, $\mathsf{H}^{\rm A}_\alpha(Q;\epsilon) = \sup \{ u : (p,u) \in \calA, p\leq\epsilon\}$. Consider the transformation
\begin{equation}
    (p,u) \mapsto \left(\Big\|\sum_{x\in\mathcal{X}} Q(\cdot,x)\Big\|_\alpha e^{\alpha p /(\alpha-1)},\Big\|\sum_{s\in\mathcal{S}} Q(s,\cdot)\Big\|_\alpha e^{\alpha u /(\alpha-1)}\right).
\end{equation}
It is straightforward to verify that this transformation is one-to-one, continuous, and monotone coordinatewise. Using this transformation, it can be shown that (the upper boundary of) $\mathcal{A}$ is homeomorphic to (the upper boundary of)
\begin{equation}
    \mathcal{B} \defined \left\{(\phi(Q,W),\psi(Q,W)) : W\in\mathcal{W}\right\},
\end{equation}
where $\phi(Q,W) \defined \sum_{y\in\mathcal{Y}} \left\|\sum_{x\in\mathcal{X}} Q(\cdot,x)W(x,y)\right\|_\alpha$ and $\psi(Q,W) \defined \sum_{y\in\mathcal{Y}} \left\|\sum_{s\in\mathcal{S}} Q(s,\cdot)W(\cdot,y)\right\|_\alpha$. By the homogeneity of the $\alpha$-norm, it can be verified that
\begin{equation}
\phi(Q,[\lambda_1W_1,\ldots,\lambda_KW_K] = \sum_{k=1}^K \lambda_k \phi(Q,W_k),
\end{equation}
whenever $K\in\mathbb{N}$, $W_1,\ldots,W_K\in\mathcal{W}$, and $\lambda_1,\ldots,\lambda_K\geq0$ with $\sum_k \lambda_k = 1$. A similar equality holds for $\psi$. Therefore, $\mathcal{B}$ is a convex set and, as a consequence, its upper boundary is the graph of a continuous function. Since the upper boundaries of $\mathcal{A}$ and $\mathcal{B}$ are homeomorphic, we conclude that the mapping $\mathsf{H}^{\rm A}_\alpha(Q;\cdot)$ is continuous. Therefore, condition (\textbf{C.}\ref{cond:cont}) is satisfied. Also, Lemma~\ref{lem:LipschitzArimotoMI} implies that condition (\textbf{C.}\ref{cond:lipcont}) is satisfied with $\displaystyle C_L = \frac{2\alpha}{\alpha-1} |\mathcal{S}|^{1-1/\alpha}$ and $\displaystyle C_U = \frac{2\alpha}{\alpha-1} |\mathcal{X}|^{1-1/\alpha}$. Observe that proving condition (\textbf{C.}\ref{cond:comp}) with $N\geq|\mathcal{X}|+1$ is equivalent to show that any point in the upper boundary of $\mathcal{A}$ can be written as $(\mathcal{L}^{\rm A}_\alpha(Q,W),\mathcal{U}^{\rm A}_\alpha(Q,W))$ for some $W\in\mathcal{W}_N$. Since $\mathcal{A}$ and $\mathcal{B}$ are homeomorphic, the latter property can be established from an analogous property for $\mathcal{B}$ which in turn follows from a minor adaptation of the argument in \cite{witsenhausen1975conditional}.

\begin{rem}
Regarding condition (\textbf{C.}\ref{cond:comp}), both \cite{asoodeh2017estimation} and \cite{hsu2018generalizing} build upon the convex analysis argument employed by  Witsenhausen and Wyner in \cite{witsenhausen1975conditional}. More specifically, they rely on an extension of Carath\'{e}odory's theorem, the so-called Fenchel-Eggleston theorem \cite{eggleston1958cambridge}, in order to find a bound for the size of the output alphabet of an optimal privacy mechanism.
\end{rem}
\section{Uniform Privacy Mechanisms}
\label{Section:UniformRobustness}

In this section we consider the scenario in which delivering privacy is the top priority for the privacy mechanism designer. We introduce privacy mechanisms that, with a certain probability, guarantee privacy for the true distribution despite having access only to an estimate of it. These mechanisms are constructed in the following conceptual way. Recall that large deviations results show that, with a certain probability, the true distribution is within some $\ell_1$ ball of the empirical distribution. Consequently, if a mechanism guarantees privacy \emph{uniformly} for every distribution within such an $\ell_1$ ball, it necessarily guarantees privacy for the true distribution with at least the same probability. In this section, we prove that there is a well-defined notion of optimality for \emph{uniform} privacy mechanisms and that these optimal mechanisms can be approximated by appropriately chosen (non-uniform) optimal privacy mechanisms as previously defined in \eqref{equa:opt_privacy_mehc}.

In order to introduce uniform privacy mechanisms precisely, recall that inequality~\eqref{equa:L1bound_Devroye} shows that, with probability at least $1-\exp(|\mathcal{S}|\cdot|\mathcal{X}|-nr^2/2)$, the true distribution $P$ is within the $\ell_1$ ball of radius $r$ centered at the empirical distribution $\hat{P}_n$,
\begin{equation}
    \mathcal{Q}_r(\hat{P}_n)\defined\{Q\in\mathcal{Q}: \|Q-\hat{P}_n\|_1\leq r\}.
\end{equation}
Based on this observation, we consider uniform privacy mechanisms which guarantee privacy for every joint distribution within $\mathcal{Q}_r(\hat{P})$, where $\hat{P}$ is any estimate of $P$. Despite the fact that $P\in\mathcal{Q}_r(\hat{P})$ with high probability, we do not know the true value of $P$. For this reason, we define optimal uniform privacy mechanisms as those that achieve the best worst-case utility within $\mathcal{Q}_r(\hat{P})$.

\begin{defn}
\label{def:robust-utility-privacy}
Let $\mathcal{Q}\subset\mathcal{P}$ and $N\in\mathbb{N}$. For $\hat{P}\in\mathcal{Q}$, $\epsilon\geq0$, and $r\geq0$, we define the set of uniform privacy mechanisms for $\mathcal{Q}_r(\hat{P})$ at $\epsilon$ as
\begin{equation}
    \mathcal{D}_{\mathcal{Q},N}(\hat{P};\epsilon,r) \defined \hspace{-5pt} \bigcap_{Q \in \mathcal{Q}_r(\hat{P})} \hspace{-5pt} \left\{ W\in \mathcal{W}_N : \mathcal{L}(Q,W) \leq \epsilon\right\}.
\end{equation}
Furthermore, we define the set of optimal uniform privacy mechanisms for $\mathcal{Q}_r(\hat{P})$ at $\epsilon$ as
\begin{equation}
\label{eq:DefParametricRPM}
    \mathcal{W}^{\dagger}_{\mathcal{Q},N}(\hat{P};\epsilon,r) \defined \operatorname*{arg\,max}_{W\in \mathcal{D}_{\mathcal{Q},N}(\hat{P};\epsilon,r)} \calU_r(\hat{P},W),
\end{equation}
where $\displaystyle \calU_r(\hat{P},W) \defined \hspace{-5pt} \min_{Q\in \mathcal{Q}_r(\hat{P})} \hspace{-5pt} \mathcal{U}(Q,W)$.
\end{defn}

Recall that, as defined in Remark~\ref{rem:sup_max_Wn}, $\mathcal{D}_N(Q;\epsilon) = \{W\in\mathcal{W}_N : \mathcal{L}(Q,W) \leq \epsilon\}$ is the set of all privacy mechanisms in $\mathcal{W}_N$ delivering an $\epsilon$-privacy guarantee for $Q$. Thus,
\begin{equation}
\label{eq:UniformPrivacyMechanismsIntersection}
    \mathcal{D}_{\mathcal{Q},N}(\hat{P};\epsilon,r) = \bigcap_{Q\in\mathcal{Q}_r(\hat{P})} \mathcal{D}_N(Q;\epsilon)
\end{equation}
is the set of all privacy mechanisms in $\mathcal{W}_N$ that deliver an $\epsilon$-privacy guarantee uniformly for all the distributions in $\mathcal{Q}_r(\hat{P})$, i.e., all the distributions at a distance less than or equal to $r$ from $\hat{P}$. For a given privacy mechanism $W$, $\mathcal{U}_r(\hat{P},W)$ quantifies the least utility $\mathcal{U}(Q,W)$ attained by $W$ over all the distributions in $\mathcal{Q}_r(\hat{P})$. Thus, by definition, $\mathcal{W}^\dagger_{\mathcal{Q},N}(\hat{P};\epsilon,r)$ is the set of all uniform privacy mechanisms for $\mathcal{Q}_r(\hat{P})$ at $\epsilon$ with the best worst-case utility.

The following lemma shows that Definition~\ref{def:robust-utility-privacy} is well-defined under conditions (\textbf{C.}\ref{cond:cont}--\ref{cond:lipcont}). Specifically, it shows that the infimum defining $\calU_r(\hat{P},W)$ and the supremum defining $\mathcal{W}^{\dagger}_{\mathcal{Q},N}(\hat{P};\epsilon,r)$ are attainable.

\begin{lem}
\label{lem:uni_pri_mech_well_defined}
Assume that conditions (\textbf{C.}\ref{cond:cont}--\ref{cond:lipcont}) hold true for a given closed set $\mathcal{Q}\subseteq\mathcal{P}$ and a given $N\in\mathbb{N}$. Then
\begin{itemize}
    \item[\textnormal{(i)}] the infimum $\inf \{\mathcal{U}(Q,W) : Q\in\mathcal{Q}_r(\hat{P})\}$ is attainable for every $W\in\mathcal{W}_N$;
    \item[\textnormal{(ii)}] the supremum $\sup \{\mathcal{U}_r(\hat{P},W) : W\in\mathcal{D}_{\mathcal{Q},N}(\hat{P};\epsilon,r)\}$ is attainable whenever $\mathcal{D}_{\mathcal{Q},N}(\hat{P};\epsilon,r)$ is not empty.
\end{itemize}
\end{lem}

\begin{proof}
See Appendix~\ref{Appendix:Proofuni_pri_mech_well_define}.
\end{proof}

Observe that \textnormal{(ii)} shows that optimal uniform privacy mechanisms do exists as long as $\mathcal{D}_{\mathcal{Q},N}(\hat{P};\epsilon,r)$ is not empty. Nonetheless, their construction might be challenging as the construction of optimal privacy mechanisms in the non-uniform sense, as defined in \eqref{equa:opt_privacy_mehc}, is already non-trivial in most cases. The following theorem shows that optimal privacy mechanism can be modified to deliver a uniform privacy guarantee without incurring in a big cost in terms of utility. Throughout this section, we consistently use $W^*$ to denote optimal privacy mechanisms, i.e., elements in $\mathcal{W}_N^*$ as defined in \eqref{eq:opt_pri_mech_W_N_star_defn}, and $W^\dagger$ to denote their uniform counterparts as defined in \eqref{eq:DefParametricRPM}.

\begin{thm}
\label{thm:desn_unfpri_mech}
Assume that conditions (\textbf{C.}\ref{cond:cont}--\ref{cond:comp}) hold true for a given closed set $\mathcal{Q}\subseteq\mathcal{P}$ and a given $N\in\mathbb{N}$. If $P\in\mathcal{Q}_r(\hat{P})$ and $\epsilon-C_Lr\geq \epsilon_{\min}(\hat{P})$, then
\begin{equation}
\label{eq:OptimalNonUniformSubsetUniform}
    \emptyset \neq \mathcal{D}_N(\hat{P};\epsilon-C_Lr)\subseteq \mathcal{D}_{\mathcal{Q},N}(\hat{P};\epsilon,r).
\end{equation}
Furthermore, for every $W^*\in \mathcal{W}_N^*(\hat{P};\epsilon-C_Lr)$ and every $W^\dagger\in\mathcal{W}_{\mathcal{Q},N}^\dagger(\hat{P};\epsilon,r)$,
\begin{align}
    \mathcal{U}(P,W^*) 
    \geq \mathcal{U}(P,W^\dagger)-\left(\mathsf{H}(\hat{P};\epsilon+C_L r)-\mathsf{H}(\hat{P};\epsilon-C_L r) +2C_Ur\right). \label{eq:LowerBoundPerformanceNonUniformAsUniform}
\end{align}
\end{thm}

\begin{proof}
See Appendix~\ref{Appendix:Proof_lem:desn_unfpri_mech}.
\end{proof}

Observe that \eqref{eq:OptimalNonUniformSubsetUniform} establishes that privacy mechanisms for $\hat{P}$ at $\epsilon-C_L r$ are, in fact, uniform privacy mechanisms at $\epsilon$. In other words, \eqref{eq:OptimalNonUniformSubsetUniform} shows that by imposing a slightly stronger privacy requirement for $\hat{P}$, we obtain \emph{uniform} privacy mechanisms for $\mathcal{Q}_r(\hat{P})$. Furthermore, under condition (\textbf{C.}\ref{cond:cont}),
\begin{equation}
    \lim_{r\downarrow0} \left(\mathsf{H}(\hat{P};\epsilon+C_L r)-\mathsf{H}(\hat{P};\epsilon-C_L r) +2C_Ur\right) = 0.
\end{equation}
Therefore, \eqref{eq:LowerBoundPerformanceNonUniformAsUniform} shows that, when $r$ is small, any optimal privacy mechanism for $\hat{P}$ at $\epsilon-C_L r$ performs almost as well as any optimal uniform privacy mechanism for $\mathcal{Q}_r(\hat{P})$ at $\epsilon$. Of course, all these conclusions hold as long as $P\in\mathcal{Q}_r(\hat{P})$ but, as pointed out before, this is the case (with high probability) in the large sample size regime.

\begin{rem}
Note that if $\mathsf{H}(\hat{P};\cdot)$ is Lipschitz continuous with Lipschitz constant $L$, then \eqref{eq:LowerBoundPerformanceNonUniformAsUniform} becomes
\begin{equation}
\label{eq:LowerBoundNonRobustAsRobust}
    \mathcal{U}(P,W^*) \geq \mathcal{U}(P,W^\dagger) - 2(C_U+LC_L)r.
\end{equation}
Observe that if $\mathsf{H}(\hat{P};\cdot)$ is a convex function differentiable at $\epsilon_0 \defined \epsilon_{\min}(\hat{P})$, then $L = \mathsf{H}'(\hat{P};\epsilon_0)$. The value of $\mathsf{H}'(\hat{P};\epsilon_0)$ is closely related with the notion of \emph{reverse} strong data processing inequality \cite{calmon2015fundamental}, see also \cite{asoodeh2016information}.
\end{rem}

Now we consider a specific example to illustrate the above results.

\begin{example}
\label{Example:ParametricPc}
Recall that $\mathsf{H}_c$ is the privacy-utility function associated to $\mathcal{L}_c$ and $\mathcal{U}_c$ as given in \eqref{eq:DefLcUc}. For ease of notation, we define
\begin{equation}
    p\#q \defined \left(\begin{matrix} (1-p)(1-q) & (1-p)q \\ pq & p(1-q) \end{matrix}\right).
\end{equation}
Let $\mathcal{Q} = \{p\#q : p\in[1/2,1],q\in[0,1-p]\}$ and $N\geq2$. This selection of $\mathcal{Q}$ captures the case where $S \sim \text{Ber}(p)$ with $p\in[1/2,1]$ and $P_{X|S} = \text{BSC}(q)$ with $q \in [0,1-p]$. By Theorem~2 in \cite{asoodeh2017estimation}, for all $Q = p\#q \in\mathcal{Q}$,
\begin{equation}
    \mathsf{H}_{c}(Q;\epsilon) = 1 - \frac{1-q}{p-q} (p+q-2pq) + \epsilon \frac{p+q-2pq}{p-q},
\end{equation}
whenever $\epsilon\in[p,1-q]$. In particular, $\mathsf{H}_c(Q;\cdot)$ is Lipschitz continuous with Lipschitz constant $\dfrac{p+q-2pq}{p-q}$. Recall that, for probability of correctly guessing, $C_L = C_U = 1$ as established in Section~\ref{Subsection:PcC123}. Hence, under the assumptions of Theorem~\ref{thm:desn_unfpri_mech}, \eqref{eq:LowerBoundNonRobustAsRobust} becomes
\begin{equation}
    \mathcal{U}(P,W^*) \geq \mathcal{U}(P,W^\dagger) - \frac{2\hat{p}(1-\hat{q})}{\hat{p}-\hat{q}} r,
\end{equation}
where $W^*\in \mathcal{W}_N^*(\hat{P};\epsilon-r)$ and $W^\dagger\in\mathcal{W}_{\mathcal{Q},N}^\dagger(\hat{P};\epsilon,r)$ with $\hat{P} \defined \hat{p}\#\hat{q} \in \mathcal{Q}$. Furthermore, by taking $\hat{P} = \hat{P}_n$ and $r=(2(4-\log\beta)/n)^{1/2}$, inequality~\eqref{equa:L1bound_Devroye} implies that, with probability at least $1-\beta$,
\begin{align}
    \mathcal{U}(P,W^*) \geq \mathcal{U}(P,W^\dagger) - \frac{2\hat{p}(1-\hat{q})}{\hat{p}-\hat{q}}\sqrt{\frac{2}{n}\left(4-\log\beta\right)}.
\end{align}
\end{example}

We finish this section by studying some convergence properties of uniform privacy mechanism, similar to those studied in Theorem~\ref{thm:optimaltooptimal_mainthm}. More specifically, the next theorem shows that although a sequence $(W^{\dagger}_n)_{n=1}^{\infty}$ with $W_n^\dagger\in\mathcal{W}_{\mathcal{Q},N}^\dagger(P_n;\epsilon,r_n)$ may not be convergent, the distance between each $W^{\dagger}_n$ and $\mathcal{W}_N^*(P;\epsilon)$ converges to zero as long as $\lim_n P_n = P$ and $\lim_n r_n = 0$. Furthermore, it also shows that, under certain conditions, this convergence can be guaranteed almost surely for the empirical distribution estimator.

\begin{thm}
\label{thm:rob_mech_converge_opt_mech}
Assume that conditions (\textbf{C.}\ref{cond:cont}--\ref{cond:comp}) hold true for a given closed set $\mathcal{Q}\subseteq\mathcal{P}$ and a given $N\in\mathbb{N}$. Let $\epsilon\in\Reals$ and $P\in\mathcal{Q}$ be given. If $P_n\in\mathcal{Q}$ with $\|P_n-P\|_1 \leq r_n$ for each $n\in\mathbb{N}$, $\lim_n r_n = 0$, and $\epsilon > \epsilon_{\min}(P)$, then, for any sequence $(W^{\dagger}_n)_{n=1}^{\infty}\subset \mathcal{W}_N$ such that $W^{\dagger}_n\in \mathcal{W}^{\dagger}_{\mathcal{Q},N}(P_n;\epsilon,r_n)$ for all $n\geq1$, 
\begin{equation}
\label{equa:W_ntoF_lim}
    \lim_{n \to \infty} \mathsf{dist}(W^{\dagger}_n,\mathcal{W}_N^*(P;\epsilon)) = 0.
\end{equation}
Furthermore, if $P_n=\hat{P}_n$ is the empirical estimator obtained from $n$ i.i.d.~samples drawn from $P$ and, for some $p>1$, $r_n \geq \sqrt{\frac{2p\log(n)}{n}}$ for all $n\geq1$, then,
\begin{equation}
    \Pr\left(\lim_{n \to \infty} \mathsf{dist}(W^{\dagger}_n, \mathcal{W}_N^*(P;\epsilon)) = 0\right)=1.
\end{equation}
\end{thm}

\begin{proof}
See Appendix~\ref{Appendix:thm:rob_mech_to_mech_prob}.
\end{proof}

The following corollary follows immediately from Thereom~\ref{thm:rob_mech_converge_opt_mech} by taking $P = P_n = \hat{P}$.

\begin{cor}
\label{cor:W_r_cont}
Assume that conditions (\textbf{C.}\ref{cond:cont}--\ref{cond:comp}) hold true for a given closed set $\mathcal{Q}\subseteq\mathcal{P}$ and a given $N\in\mathbb{N}$. If $\hat{P}\in\mathcal{Q}$ and $\epsilon>\epsilon_{\min}(\hat{P})$, then
\begin{align}
    \limsup_{r\downarrow0} \mathcal{W}^{\dagger}_{\mathcal{Q},N}(\hat{P};\epsilon,r) \subseteq \mathcal{W}_N^*(\hat{P};\epsilon),
\end{align}
where 
\begin{align}
    \limsup_{r\downarrow0} \mathcal{W}^{\dagger}_{\mathcal{Q},N}(\hat{P};\epsilon,r)
    \defined \left\{W\in\mathcal{W}_N:\liminf_{r\downarrow0}\mathsf{dist}\left(W,\mathcal{W}^{\dagger}_{\mathcal{Q},N}(\hat{P};\epsilon,r)\right)=0\right\}.
\end{align}
\end{cor}
Theorem~\ref{thm:desn_unfpri_mech} shows that in terms of performance, optimal privacy mechanisms approximate optimal uniform privacy mechanisms in the small $r$ regime. The previous corollary shows that, in the same regime, optimal uniform privacy mechanisms are {\it geometrically close} to the set of optimal privacy mechanisms. These two results evidence an intrinsic relation between optimal privacy mechanisms and their uniform counterparts.
\section{Numerical Experiments}
\label{sec::Num_Exp}

We illustrate some of the results derived in the previous sections through two numerical experiments. The first experiment, conducted on a synthetic dataset, illustrates the convergence of the empirical privacy-utility function and optimal privacy mechanisms to their corresponding limits as the sample size increases. The second experiment, performed on a real-world dataset, displays the discrepancy between the privacy-utility guarantees during the design and testing of a privacy mechanism.

\subsection{Synthetic Dataset}

Here we illustrate the convergence of the empirical privacy-utility function and optimal privacy mechanisms as the sample size increases. The joint distribution matrix $\mathbf{P}$, determined by $\mathbf{P}_{s,x}\defined P(s,x)$, is chosen to be
\begin{equation}
    \mathbf{P} = \left(\begin{matrix} 0.42 & 0.18 \\ 0.16 & 0.24 \end{matrix}\right),
\end{equation}
and both privacy leakage and utility are measured using $\chi^2$-information, i.e., the $f$-information associated to the function $f(t) = (t-1)^2$. For any given value of $n$, we randomly select $n$ i.i.d.~samples $\{(s_i,x_i)\}_{i=1}^n$ drawn from the joint distribution $P$ and compute the empirical distribution $\hat{P}_n$.

In Figure~\ref{fig:PUFunc_chi_square} (top left), we depict the privacy-utility function $\mathsf{H}_{\chi^2}(P;\cdot)$ and the empirical privacy-utility function $\mathsf{H}_{\chi^2}(\hat{P}_n;\cdot)$ for three different values of $n$. Note that the privacy-utility function of the empirical distribution converges (pointwise) to the corresponding function for the true distribution, i.e., for any given $\epsilon\geq0$ we have that $\mathsf{H}_{\chi^2}(\hat{P}_n;\epsilon)$ approaches $\mathsf{H}_{\chi^2}(P;\epsilon)$ as $n$ grows. This corroborates the conclusion of Proposition~\ref{Prop:ContinuityH} which establishes that, for any given $\epsilon\in\Reals$, $\mathsf{H}(\cdot;\epsilon)$ is continuous over $\{Q\in\mathcal{Q}: \epsilon_{\min}(Q)< \epsilon\}$. In Figure~\ref{fig:PUFunc_chi_square} (top right), we depict the corresponding optimal privacy mechanisms. In order to visualize a privacy mechanism $P_{Y|X}$ in the xy-plane, we let the x-axis and y-axis to be the values of $P_{Y|X}(0|0)$ and $P_{Y|X}(1|1)$, respectively. Observe that optimal privacy mechanisms are not unique. Hence, depending on the algorithm used to obtain such mechanisms, a sequence of empirical optimal privacy mechanisms may not converge to a fixed mechanism. However, as observed from Figure~\ref{fig:PUFunc_chi_square} (top right), the distance between any such sequence and the set of optimal privacy mechanisms for the true distribution will necessarily converge to zero which echoes Theorem~\ref{thm:optimaltooptimal_mainthm}.

Finally we scatter plot the largest (signed) gap, denoted by $\Delta_n$, between $\mathsf{H}_{\chi^2}(P_n;\cdot)$ and $\mathsf{H}_{\chi^2}(P;\cdot)$. In particular, the absolute value of $\Delta_n$ is equal to the uniform norm of the function $\mathsf{H}_{\chi^2}(P_n;\cdot)-\mathsf{H}_{\chi^2}(P;\cdot)$, i.e.,
\begin{align}
|\Delta_n| = \sup_{\epsilon\in[0,\infty)} |\mathsf{H}_{\chi^2}(P_n;\epsilon)-\mathsf{H}_{\chi^2}(P;\epsilon)|.
\end{align}
We depict the value of $\Delta_n$ in Figure~\ref{fig:PUFunc_chi_square} (bottom). As shown, when the number of samples increases, $|\Delta_n|$ tends to decrease which illustrates the uniform convergence established in Corollary~\ref{cor::unif_conv_PUF}. Due to the mismatch between the empirical and true distributions, privacy and utility might be over or under-estimated, leading to the variable sign of $\Delta_n$.
\begin{figure}[t!]
\centering
\includegraphics[width=0.48\textwidth]{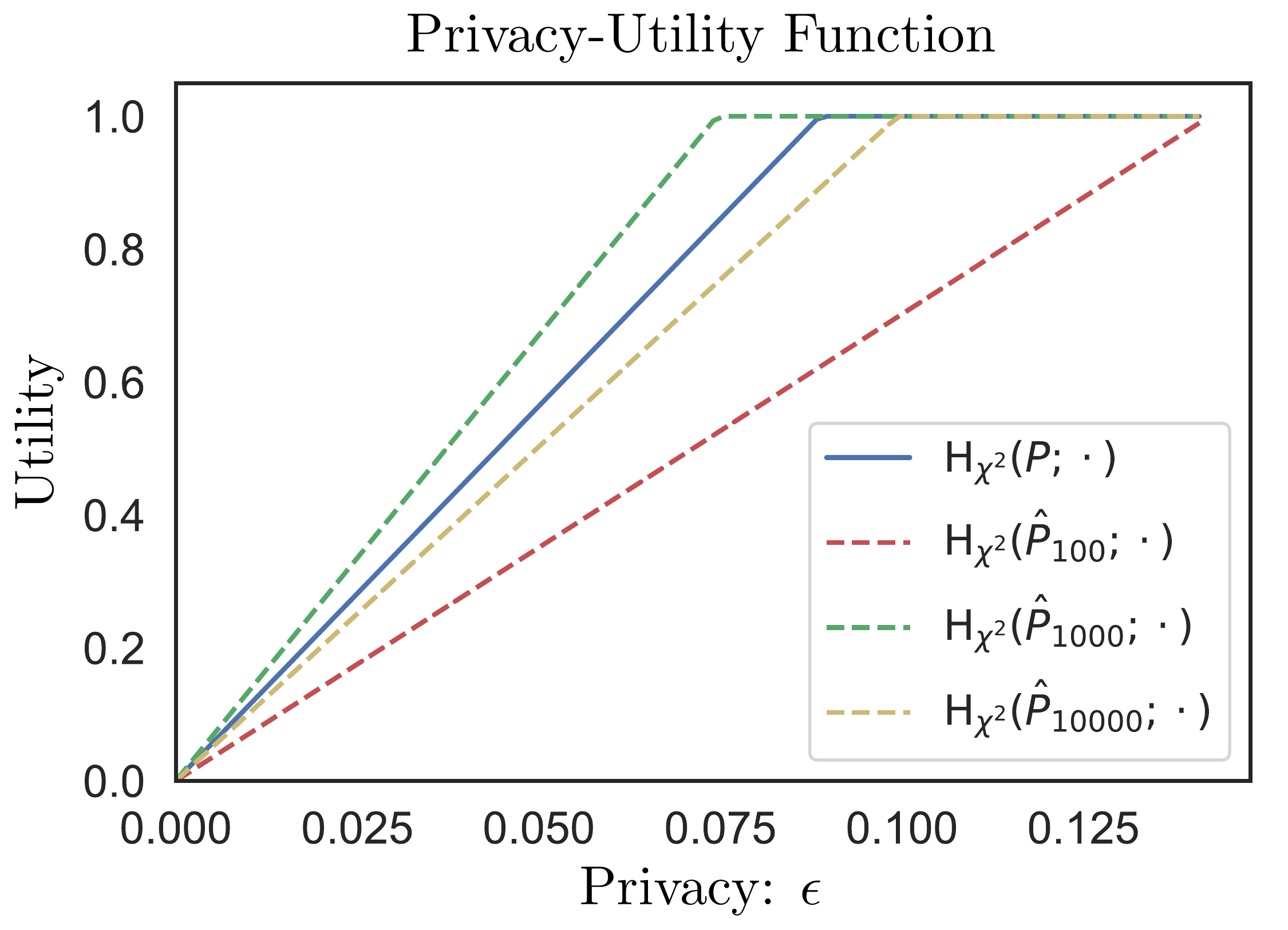}
\includegraphics[width=0.48\textwidth]{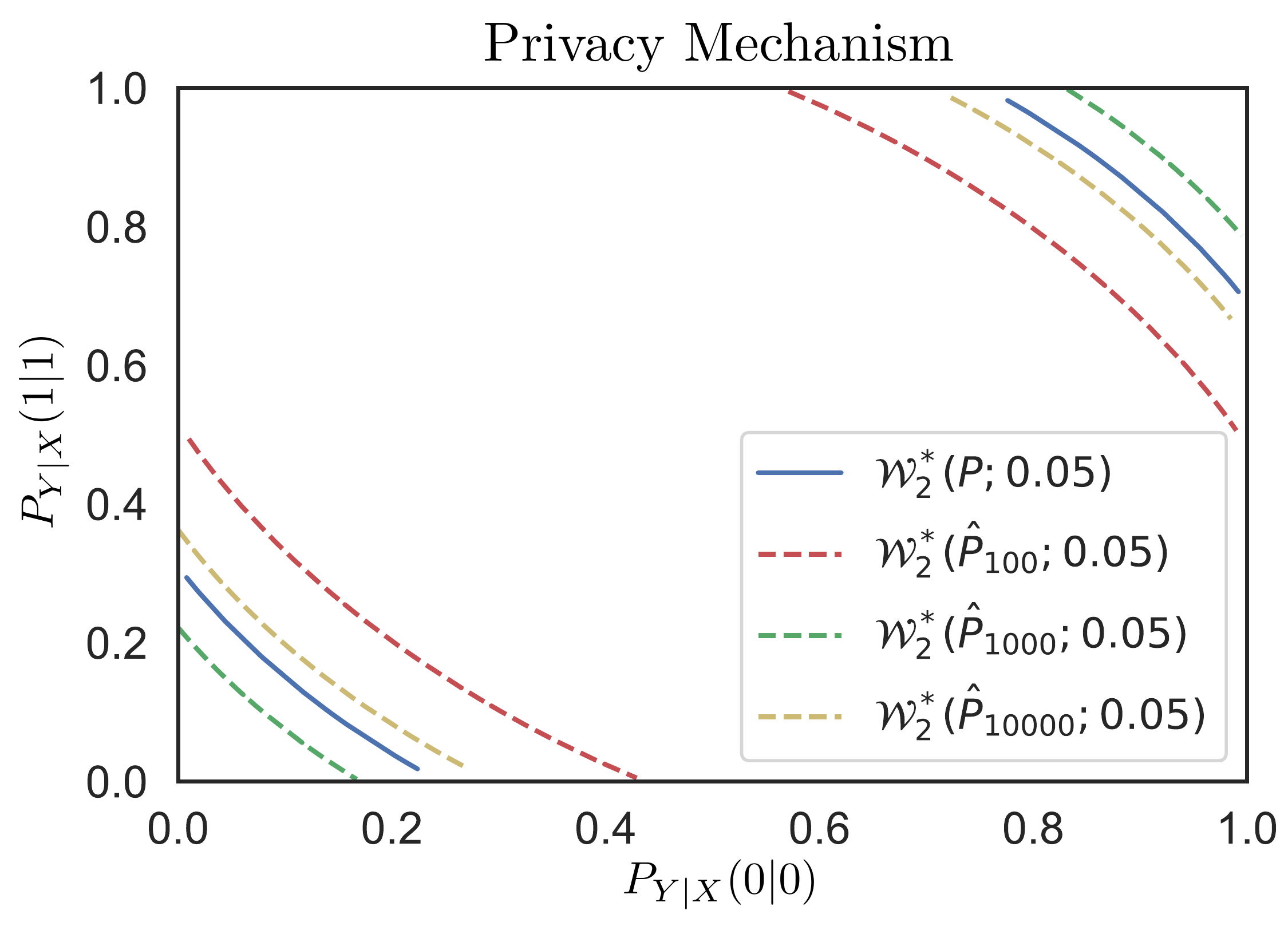}
\includegraphics[width=0.5\textwidth]{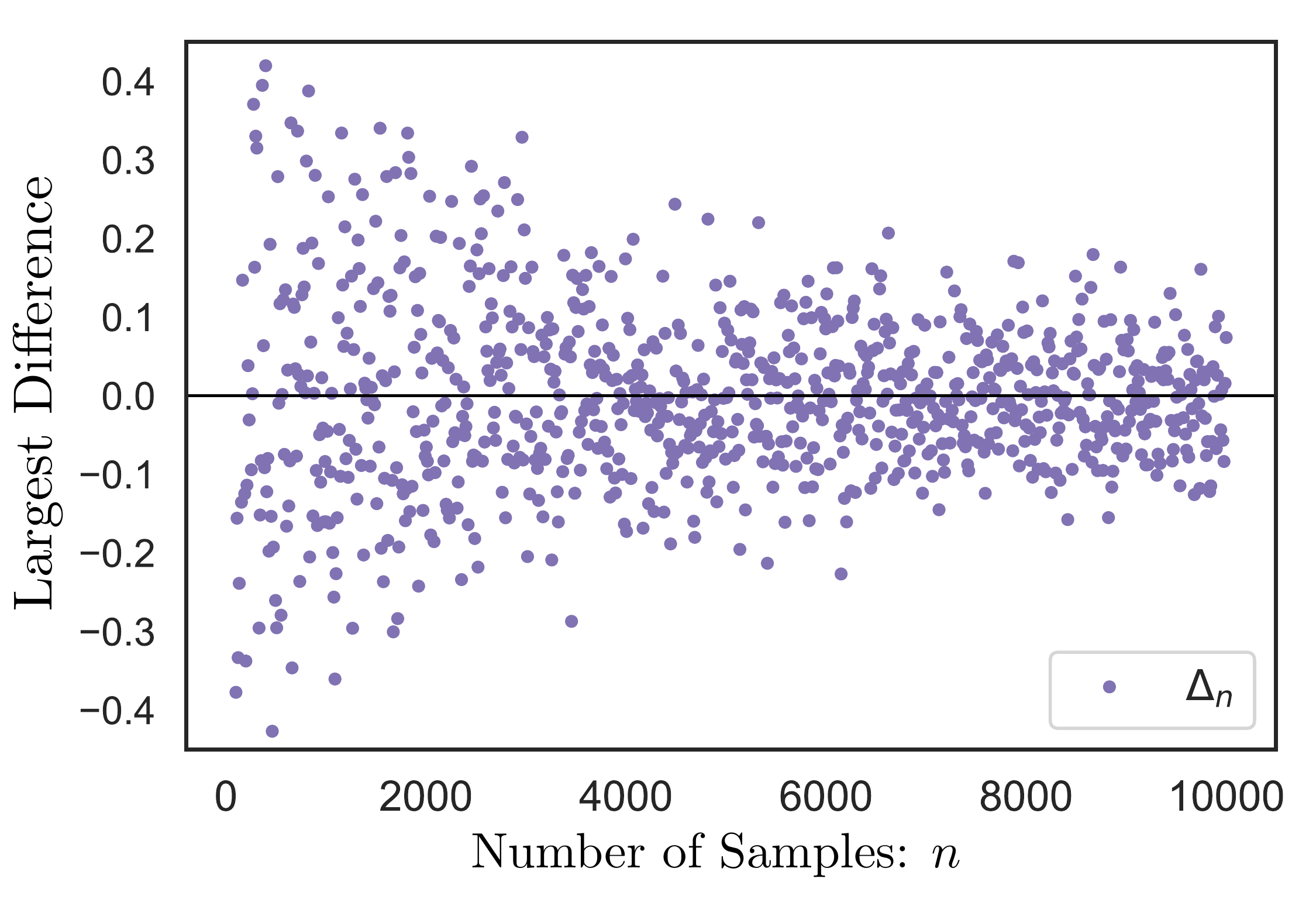}
\caption{\small{Both privacy leakage and utility are measured using the $f$-information with $f(t)=(t-1)^2$. Top left: privacy-utility function $\mathsf{H}_{\chi^2}(P;\cdot)$ and empirical privacy-utility function $\mathsf{H}_{\chi^2}(\hat{P}_n;\cdot)$. Top right: corresponding optimal privacy mechanisms for $\epsilon=0.05$. Bottom: largest (signed) difference between $\mathsf{H}_{\chi^2}(P_n;\cdot)$ and $\mathsf{H}_{\chi^2}(P;\cdot)$.
}}
\label{fig:PUFunc_chi_square}
\end{figure}

\subsection{ProPublica's COMPAS Recidivism Dataset}
\label{Section:SimulationPc}

We now illustrate Theorems~\ref{thm:rob_PCG} and \ref{thm:bounds_AMI_SMI} in Section~\ref{Section:PointwiseRobustness} through ProPublica's COMPAS dataset \cite{angwin2016machine}, which contains the criminal history, jail and prison time, demographics and COMPAS risk scores for defendants in Broward County from 2013 and 2014. We process the original dataset by dropping records with missing information and quantizing some of the interest variables. Our final dataset contains 5278 records. We choose the private variable (S): $\textfn{Race}\in\{\textfn{Caucasian}, \textfn{African-American}\}$; the useful variable (X): $\textfn{PriorCounts}\in\{0, 1 - 3, >3\}$ and $\textfn{AgeCategory}\in\{{<25}, {25 - 45}, {>45}\}$. Hence, the size of support sets are $|\mathcal{S}|=2$ and $|\mathcal{X}|=9$. 

Given $n\in\mathbb{N}$, we choose $n$ records from the dataset as our training set and another $n$ different records as our testing set. We use probability of correctly guessing to measure both privacy and utility. We compute the empirical distribution $\hat{P}_{n,\textfn{train}}$ based on the training set and let $W_n$ be the privacy mechanism that solves the following optimization problem:
\begin{align}
    \max_{W\in\mathcal{W}_R}\  &\calU_c(\hat{P}_{n,\textfn{train}},W)\\
    \sto~& \calL_c(\hat{P}_{n,\textfn{train}},W)\leq0.65,
\end{align}
where $\mathcal{W}_R$ denotes the set of {\it randomized response mechanisms} \cite{warner1965randomized,kairouz2014extremal}, i.e., the set of all privacy mechanisms $W$ with output alphabet $\mathcal{X}$ such that, for some $\rho\geq0$,
\begin{align}
    W(i,j) =
    \begin{cases}
        \dfrac{e^\rho}{e^\rho+|\mathcal{X}|-1} & \text{if } i = j,\\
        \dfrac{1}{e^\rho+|\mathcal{X}|-1} & \text{if } i\neq j.
    \end{cases}
\end{align}
Then, we let $\hat{P}_{n,\textfn{test}}$ be the empirical distribution of the testing set and compute the privacy-utility guarantees attained by $W_n$ for this distribution. Finally, we evaluate the discrepancy between the privacy-utility guarantees provided for $\hat{P}_{n,\textfn{train}}$ and $\hat{P}_{n,\textfn{test}}$:
\begin{align}
    \Delta_{L,n} &\defined |\calL_c(\hat{P}_{n,\textfn{test}},W_n)-\calL_c(\hat{P}_{n,\textfn{train}},W_n)|,\\
    \Delta_{U,n} &\defined |\calU_c(\hat{P}_{n,\textfn{test}},W_n)-\calU_c(\hat{P}_{n,\textfn{train}},W_n)|.
\end{align}
By the triangle inequality and Theorem~\ref{thm:rob_PCG}, with probability at least $1-\beta$, these two discrepancies are upper bounded by
\begin{align}
\label{eq:UpperBoundExperiment}
    \mathsf{UpperBound}_n \defined 2\sqrt{\frac{2}{n}\left(|\mathcal{S}|\cdot|\mathcal{X}|-\log\beta\right)}.
\end{align}
Figure~\ref{fig:Disc_pc} depicts the discrepancies $\Delta_{L,n}$ and $\Delta_{U,n}$, as functions of $n$, along with the upper bound $\mathsf{UpperBound}_n$ for $\beta=10^{-1}$. Observe that when number of samples increases, the discrepancy between the privacy-utility guarantees tends to decrease.

\begin{figure}[t!]
\centering
\includegraphics[width=0.96\textwidth]{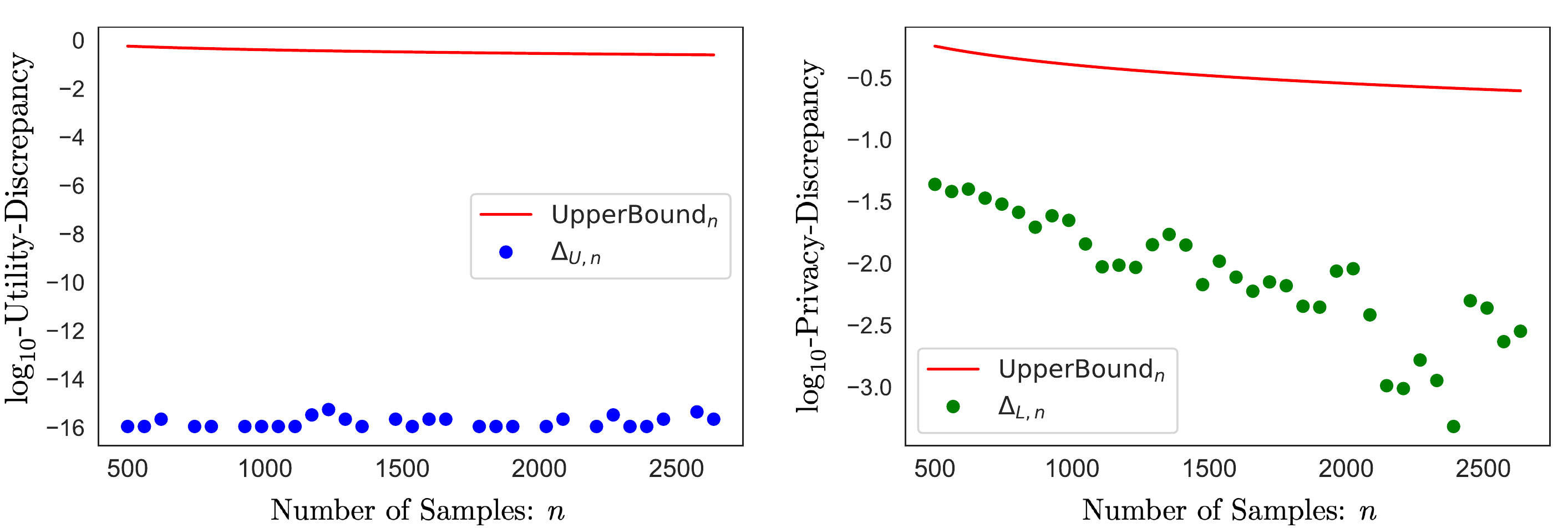}
\caption{\small{\rev{Both privacy leakage and utility are measured using probability of correctly guessing. Left: discrepancy between utility guarantees for the training and testing sets. Right: discrepancy between privacy guarantees for the training and testing sets. In both pictures the theoretical upper bound in \eqref{eq:UpperBoundExperiment} is shown in red.}}}
\label{fig:Disc_pc}
\end{figure}

To further illustrate our results, we perform a similar experiment using Arimoto's mutual information of order $2$. Specifically, we compute the privacy mechanism $W_n$ by solving the following optimization problem:
\begin{align}
    \max_{W\in\mathcal{W}_Z}\  &\calU_{2}^{A}(\hat{P}_{n,\textfn{train}},W)\\
    \sto~& \calL_{2}^{A}(\hat{P}_{n,\textfn{train}},W) \leq 0.05,
\end{align}
where $\mathcal{W}_Z$ denotes the set of privacy mechanisms $W$ with output alphabet $\mathcal{X}$ such that, for some $\bar{x}\in\mathcal{X}$ and $\zeta\geq0$,
\begin{align}
    W(i,j) =
    \begin{cases}
        1 & \text{if } i = j = \bar{x},\\
        1-\zeta & \text{if } i=j \neq \bar{x},\\
        \zeta & \text{if } i\neq \bar{x}, j=\bar{x}.
    \end{cases}
\end{align}
In the binary case, $\mathcal{W}_Z$ is nothing but the collection of Z-channels. It has been proved \cite{asoodeh2017estimation} that these channels are capable to achieve optimal privacy-utility trade-offs in some cases. Let $\Delta_{L,n}$ (resp.~$\Delta_{U,n}$) be the discrepancy between the privacy (resp.~utility) guarantees for $\hat{P}_{n,\textfn{train}}$ and $\hat{P}_{n,\textfn{test}}$. By Theorem~\ref{thm:bounds_AMI_SMI}, with probability at least $1-\beta$, $\Delta_{L,n}$ and $\Delta_{U,n}$ are upper bounded by
\begin{align}
    &\text{Privacy:} \mathsf{UpperBound}_n \defined 8\sqrt{\frac{2}{n}\left(|\mathcal{S}|\cdot|\mathcal{X}|-\log\beta\right) \cdot |\mathcal{S}|}, \label{eq:AMI_pri_UpperBoundExperiment}\\
    &\text{Utility:} \mathsf{UpperBound}_n \defined 8\sqrt{\frac{2}{n}\left(|\mathcal{S}|\cdot|\mathcal{X}|-\log\beta\right) \cdot |\mathcal{X}|},\label{eq:AMI_uti_UpperBoundExperiment}
\end{align}
respectively. Figure~\ref{fig:Disc_AMI} depicts the discrepancies $\Delta_{L,n}$ and $\Delta_{U,n}$, as functions of $n$, along with their corresponding upper bounds for $\beta=10^{-1}$.

\begin{figure}[t!]
\centering
\includegraphics[width=0.96\textwidth]{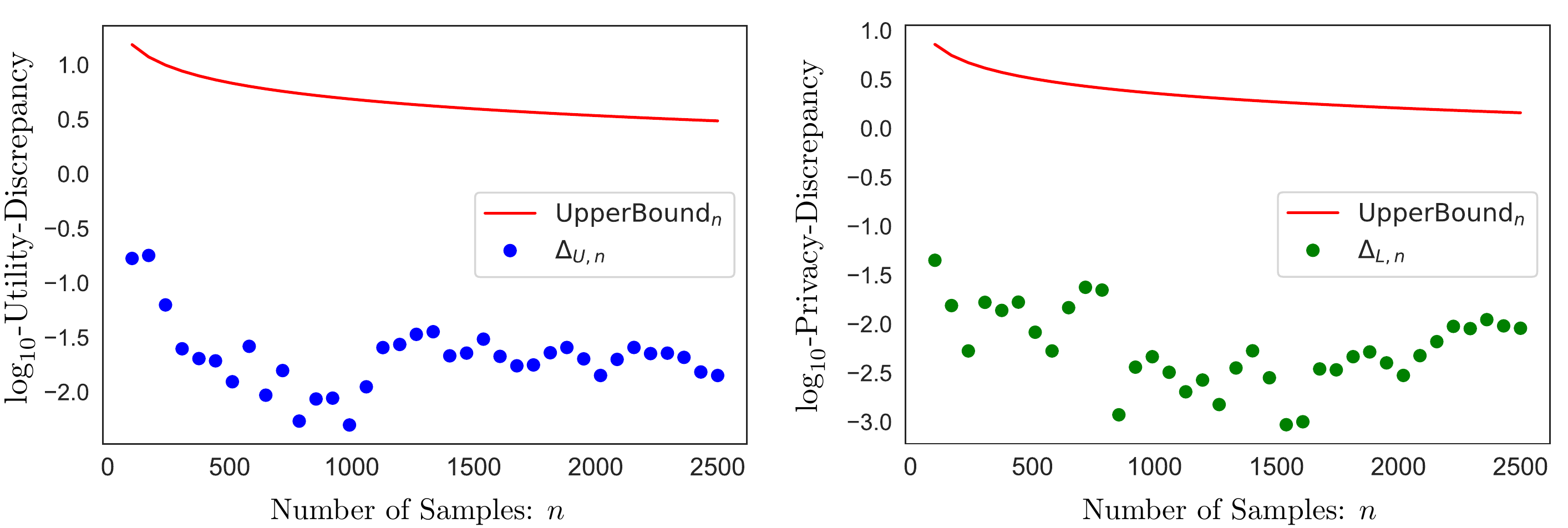}
\caption{\small{Both privacy leakage and utility are measured using Arimoto's mutual information of order $2$. Left: discrepancy between utility guarantees for the training and testing sets. Right: discrepancy between privacy guarantees for the training and testing sets. The theoretical upper bounds, in \eqref{eq:AMI_pri_UpperBoundExperiment} and \eqref{eq:AMI_uti_UpperBoundExperiment} respectively, are shown in red.}}
\label{fig:Disc_AMI}
\end{figure}

\section{Concluding Remarks}
\label{Section:Conclusion}

In this work, we analyzed the effect of a limited sample size on the disclosure of data under privacy constraints. We considered a setting where data is released upon observing non-private features correlated with a set of private features. We evaluated privacy and utility using one out of five of the most commonly used information leakage and utility measures in information-theoretic privacy: probability of correctly guessing, $f$-information, Arimoto's mutual information, Sibson's mutual information, and maximal $\alpha$-leakage.

The effect of a limited sample size was assessed via probabilistic upper bounds for the difference between the privacy-utility guarantees for the empirical and true distributions. An important feature of these bounds is that they are completely independent of the privacy mechanism at hand. On a technical level, the proofs of these bounds depend on large deviations results already available in the literature and continuity properties of information leakage measures established in this work. Furthermore, we have established new continuity properties of privacy-utility functions. Using these properties, we have shown that the limit of a convergent sequence of optimal privacy mechanisms is an optimal privacy mechanism itself.

In order to mitigate the effect of a limited sample size on the privacy guarantees delivered to the true distribution, we introduced the notion of \emph{uniform} privacy mechanisms. By definition, these mechanisms provide a specific privacy guarantee for \emph{every} distribution in a given subset of the probability simplex. In particular, when this subset is a neighborhood of the empirical distribution, large deviations results imply that privacy is guaranteed for the true distribution with high probability. While the construction of optimal uniform privacy mechanisms might be challenging, we proved that these mechanisms can be approximated in a natural way by optimal privacy mechanisms in the non-uniform sense. More specifically, we have proved that an optimal privacy mechanism for the empirical distribution delivers a slightly weaker privacy guarantee for a whole neighborhood of the empirical distribution and performs almost as well as any optimal uniform privacy mechanism. By establishing convergence results regarding optimal uniform privacy mechanisms, we have further exhibited the intrinsic relation between optimal privacy mechanisms and their uniform counterparts.

While this work predominantly focused on large deviations bounds, the continuity properties derived in this paper can be applied together with contemporary results to derive similar upper bounds in other estimation frameworks, e.g., $\ell_1$ minimax estimation.

\section*{Acknowledgements}

The authors would like to thank Dr. Shahab Asoodeh (Harvard University) for his keen comments regarding the proof of Lemma~\ref{lem:cut_lesslikely_symbols}. The authors would also like to thank the anonymous reviewers and the associate editor for their valuable and constructive feedback.

\appendix
\section{Proofs from Section~\ref{Section:PointwiseRobustness}}
\label{Appendix:ProofsIII}

\subsection{Lemma~\ref{lem:rob_PCG1}}
\label{Appendix:ProofPc}

Consider the following elementary lemma whose proof is included for the sake of completeness.

\begin{lem}
\label{lem:VariationSY2SX}
Let $Q_1,Q_2\in\mathcal{P}$ and $W\in\mathcal{W}$ be given. For each $i\in\{1,2\}$, let $S_i \to X_i \to Y_i$ be such that $P_{S_i,X_i} = Q_i$ and $P_{Y_i|X_i} = W$. Then,
\begin{equation}
    \|P_{S_1,Y_1} - P_{S_2,Y_2}\|_1 \leq \|Q_1 - Q_2\|_1.
\end{equation}
\end{lem}

\begin{proof}
By the definition of $\|\cdot\|_1$, we have that
\begin{equation}
    \|P_{S_1,Y_1} - P_{S_2,Y_2}\|_1 = \sum_{s\in\mathcal{S}} \sum_{y\in\mathcal{Y}} |P_{S_1,Y_1}(s,y) - P_{S_2,Y_2}(s,y)|.
\end{equation}
By assumption, for each $i\in\{1,2\}$, $P_{S_i,Y_i}(s,y) = \sum_{x} Q_i(s,x) W(x,y)$. Hence, by the triangle inequality,
\begin{align}
    \|P_{S_1,Y_1} - P_{S_2,Y_2}\|_1 &\leq \sum_{y\in\mathcal{Y}} \sum_{s\in\mathcal{S}} \sum_{x\in\mathcal{X}} |Q_1(s,x) - Q_2(s,x)| W(x,y)\\
    &= \|Q_1-Q_2\|_1,
\end{align}
where we used the fact that $\sum_y W(x,y) =1$ for all $x\in\mathcal{X}$.
\end{proof}

The proof of Lemma~\ref{lem:rob_PCG1} relies on the following elementary observation: Let $n\in\mathbb{N}$. If $a_i,b_i\in\mathbb{R}$ for all $i=1,\ldots,n$, then
\begin{equation}
\label{eq:DifMaxMaxDif}
    \left|\max_{i=1,\ldots,n} a_i - \max_{i=1,\ldots,n} b_i\right| \leq \max_{i=1,\ldots,n} |a_i - b_i| \leq \sum_{i=1}^n |a_i - b_i|.
\end{equation}

\begin{proof}
For ease of notation, let $\Delta_L \defined |\mathcal{L}_c(Q_1,W)-\mathcal{L}_c(Q_2,W)|$. For each $i\in\{1,2\}$, let $S_i\to X_i \to Y_i$ be such that $P_{S_i,X_i} = Q_i$ and $P_{Y_i|X_i}=W$. With this notation, $\mathcal{L}_c(Q_i,W) = P_c(S_i|Y_i)$ where
\begin{equation}
    P_c(S_i|Y_i) = \sum_{y\in\mathcal{Y}} \max_{s\in\mathcal{S}} P_{S_i,Y_i}(s,y).
\end{equation}
By the triangle inequality, we have that
\begin{equation}
    \Delta_L \leq \sum_{y\in\mathcal{Y}} \left| \max_{s\in \mathcal{S}} P_{S_1,Y_1}(s,y) - \max_{s\in \mathcal{S}} P_{S_2,Y_2}(s,y)\right|.
\end{equation}
An immediate application of \eqref{eq:DifMaxMaxDif} leads to
\begin{align}
    \Delta_L &\leq \sum_{y\in\mathcal{Y}} \sum_{s\in\mathcal{S}} |P_{S_1,Y_1}(s,y)-P_{S_2,Y_2}(s,y)|\\
    &= \|P_{S_1,Y_1} - P_{S_2,Y_2}\|_1\\
    &\leq \|Q_1 - Q_2\|_1,
\end{align}
where the last inequality follows from Lemma~\ref{lem:VariationSY2SX}. Mutatis mutandis, it can be shown that
\begin{align}
    |\mathcal{U}_c(Q_1,W)-\mathcal{U}_c(Q_2,W)|
    &\leq \|Q_1-Q_2\|_1,
\end{align}
as required.
\end{proof}
\subsection{Lemma~\ref{lem:rob_f_inf1}}
\label{Appendix:Prooff}

The proof of the following lemma is similar to the one of Lemma~\ref{lem:VariationSY2SX}. The details are left to the reader.

\begin{lem}
\label{lem:VariationSXY2SX}
Let $Q_1,Q_2\in\mathcal{P}$ and $W\in\mathcal{W}$ be given. For each $i\in\{1,2\}$, let $S_i \to X_i \to Y_i$ be such that $P_{S_i,X_i} = Q_i$ and $P_{Y_i|X_i} = W$. Then,
\begin{equation}
    \|P_{Y_1} - P_{Y_2}\|_1 \leq \max\{\|P_{S_1} - P_{S_2}\|_1,\|P_{X_1} - P_{X_2}\|_1\} \leq \|Q_1 - Q_2\|_1.
\end{equation}
\end{lem}

Now we are in position to prove Lemma~\ref{lem:rob_f_inf1}.

\begin{proof}
For ease of notation, let $\Delta_L \defined |\mathcal{L}_f(Q_1,W)-\mathcal{L}_f(Q_2,W)|$. For each $i\in\{1,2\}$, let $S_i\to X_i \to Y_i$ be such that $P_{S_i,X_i} = Q_i$ and $P_{Y_i|X_i}=W$. With this notation, $\mathcal{L}_f(Q_i,W)=I_f(P_{S_i,Y_i})$. By the definition of $f$-information, we have that
\begin{equation}
    \Delta_L = \bigg| \sum_{s\in\mathcal{S}} \sum_{y\in\mathcal{Y}} \left(P_{S_1}(s) P_{Y_1}(y) f\left(\frac{P_{S_1,Y_1}(s,y)}{P_{S_1}(s)P_{Y_1}(y)}\right) 
    - P_{S_2}(s) P_{Y_2}(y) f\left(\frac{P_{S_2,Y_2}(s,y)}{P_{S_2}(s)P_{Y_2}(y)}\right) \right)\bigg|.
\end{equation}
An application of the triangle inequality leads to $\Delta_L \leq {\rm I} + {\rm II}$, where
\begin{align}
\label{eq:DefIII1}
    {\rm I} = \sum_{s\in\mathcal{S}} \sum_{y\in\mathcal{Y}} |P_{S_1}(s) P_{Y_1}(y) - P_{S_2}(s) P_{Y_2}(y)|
    \left|f\left(\frac{P_{S_1,Y_1}(s,y)}{P_{S_1}(s)P_{Y_1}(y)}\right)\right|,
\end{align}
\begin{align}
\label{eq:DefIII2}
    {\rm II} = \sum_{s\in\mathcal{S}} \sum_{y\in\mathcal{Y}} P_{S_2}(s) P_{Y_2}(y)  \left|f\left(\frac{P_{S_1,Y_1}(s,y)}{P_{S_1}(s)P_{Y_1}(y)}\right) - f\left(\frac{P_{S_2,Y_2}(s,y)}{P_{S_2}(s)P_{Y_2}(y)}\right)\right|.
\end{align}
First we provide an upper bound for ${\rm I}$. Recall the definition of the supremum norm in \eqref{eq:DefSupNorm}. Observe that $P_{S_i,Y_i}(s,y) \leq P_{Y_i}(y)$ for all $i\in\{1,2\}$, $s\in\mathcal{S}$, and $y\in\mathcal{Y}$. Hence, for all $s\in\mathcal{S}$ and $y\in\mathcal{Y}$,
\begin{equation}
\label{eq:LipschitzfInformationsBoundProbabilityQuotient}
\max\left\{\frac{P_{S_1,Y_1}(s,y)}{P_{S_1}(s)P_{Y_1}(y)},\frac{P_{S_2,Y_2}(s,y)}{P_{S_2}(s)P_{Y_2}(y)}\right\} \leq \max\left\{\frac{1}{P_{S_1}(s)},\frac{1}{P_{S_2}(s)}\right\} \leq m_S^{-1},
\end{equation}
and thus $\displaystyle \left|f\left(\frac{P_{S_1,Y_1}(s,y)}{P_{S_1}(s)P_{Y_1}(y)}\right)\right| \leq K_{f,m_S^{-1}}$, as $|f(t)|\leq K_{f,m_S^{-1}}$ for all $t\in [0,m_S^{-1}]$. As a consequence,
\begin{align}
    {\rm I} &\leq K_{f,m_S^{-1}} \sum_{s\in\mathcal{S}} \sum_{y\in\mathcal{Y}} |P_{S_1}(s) P_{Y_1}(y) - P_{S_2}(s) P_{Y_2}(y)|\\
    &\leq K_{f,m_S^{-1}} \sum_{s\in\mathcal{S}} \sum_{y\in\mathcal{Y}} \Big[ P_{Y_1}(y)|P_{S_1}(s) - P_{S_2}(s)|
    +P_{S_2}(s)|P_{Y_1}(y)-P_{Y_2}(y)|\Big]\\
    \label{eq:LipschitzfInformationsBoundI} &= K_{f,m_S^{-1}} \left(\| P_{S_1} - P_{S_2} \|_1 + \| P_{Y_1} - P_{Y_2}\|_1\right).
\end{align}
By Lemma~\ref{lem:VariationSXY2SX}, we conclude that
\begin{equation}
    {\rm I} \leq 2 K_{f,m_S^{-1}} \|Q_1 - Q_2\|_1.
\end{equation}
Now we focus on ${\rm II}$. Recall that $f$ is locally Lipschitz by assumption, and thus it is Lipschitz on $[0,m_S^{-1}]$. As defined in \eqref{eq:DefLipschitzConstant}, let $L_{f,m_S^{-1}}$ be the Lipschitz constant of $f$ on the latter interval. Hence, \eqref{eq:LipschitzfInformationsBoundProbabilityQuotient} implies that
\begin{equation}
    \left|f\left(\frac{P_{S_1,Y_1}(s,y)}{P_{S_1}(s)P_{Y_1}(y)}\right) - f\left(\frac{P_{S_2,Y_2}(s,y)}{P_{S_2}(s)P_{Y_2}(y)}\right)\right| \leq L_{f,m_S^{-1}} \left|\frac{P_{S_1,Y_1}(s,y)}{P_{S_1}(s)P_{Y_1}(y)} - \frac{P_{S_2,Y_2}(s,y)}{P_{S_2}(s)P_{Y_2}(y)}\right|.
\end{equation}
As a result, we have that
\begin{equation}
\label{equa:BoundD}
    {\rm II} \leq L_{f,m_S^{-1}} \sum_{s\in\mathcal{S}} \sum_{y\in\mathcal{Y}} \frac{|P_{S_2}(s)P_{Y_2}(y)P_{S_1,Y_1}(s,y)
    -P_{S_1}(s)P_{Y_1}(y)P_{S_2,Y_2}(s,y)|}{P_{S_1}(s)P_{Y_1}(y)}.
\end{equation}
By the triangle inequality, the numerator of the quotient in \eqref{equa:BoundD} is upper bounded by
\begin{equation}
\begin{aligned}
    P_{S_1,Y_1}(s,y) |P_{S_2}(s)P_{Y_2}(y)-P_{S_1}(s)P_{Y_1}(y)|
    + P_{S_1}(s)P_{Y_1}(y) |P_{S_1,Y_1}(s,y)-P_{S_2,Y_2}(s,y)|.
\end{aligned}
\end{equation}
In particular,
\begin{align}
    {\rm II} &\leq L_{f,m_S^{-1}} \sum_{s\in\mathcal{S}} \sum_{y\in\mathcal{Y}} \frac{P_{S_1,Y_1}(s,y)}{P_{S_1}(s)P_{Y_1}(y)} |P_{S_2}(s)P_{Y_2}(y)-P_{S_1}(s)P_{Y_1}(y)|\\
    &\quad
    + L_{f,m_S^{-1}} \sum_{s\in\mathcal{S}} \sum_{y\in\mathcal{Y}} |P_{S_1,Y_1}(s,y)-P_{S_2,Y_2}(s,y)|\\
    &\leq m_S^{-1} L_{f,m_S^{-1}} \Big(\|P_{Y_1}-P_{Y_2}\|_1+\|P_{S_1}-P_{S_2}\|_1\Big) + L_{f,m_S^{-1}}\|P_{S_1,Y_1}-P_{S_2,Y_2}\|_1,
\end{align}
where the last inequality follows from \eqref{eq:LipschitzfInformationsBoundProbabilityQuotient} and the argument used in \eqref{eq:LipschitzfInformationsBoundI}. Hence, Lemma~\ref{lem:VariationSY2SX} and Lemma~\ref{lem:VariationSXY2SX} imply that
\begin{equation}
    {\rm II} 
    \leq \left(2m_S^{-1} + 1\right) L_{f,m_S^{-1}} \|Q_1 - Q_2\|_1.
\end{equation}
Since $\Delta_L \leq {\rm I} + {\rm II}$, we conclude that
\begin{equation}
    \Delta_L \leq \left(2 K_{f,m_S^{-1}} + \left(2m_S^{-1} + 1\right) L_{f,m_S^{-1}}\right) \|Q_1 - Q_2\|_1.
\end{equation}
Mutatis mutandis, it can be shown that
\begin{align}
    |\mathcal{U}_f(Q_1,W)-\mathcal{U}_f(Q_2,W)|
    \leq \left(2 K_{f,m_X^{-1}} + \left(2m_X^{-1} + 1\right) L_{f,m_X^{-1}}\right) \|Q_1 - Q_2\|_1.
\end{align}
The proof is complete.
\end{proof}
\subsection{Theorem~\ref{thm:main_robust}}
\label{Appendix:ProoffGeneralization}

\begin{proof}
By choosing $Q_1=\hat{P}_n$ and $Q_2=P$ in Lemma~\ref{lem:rob_f_inf1}, we have
\begin{align}
    \label{equa:Lip_proof_lem_pointwise_L} |\calL_f(\hat{P}_n,W)-\calL_f(P,W)| &\leq C_{f,m_S} \|\hat{P}_n-P\|_1,
\end{align}
where
\begin{align}
    m_S &= \min\left\{\left\{\sum_{x\in\mathcal{X}} \hat{P}_n(s,x): s\in\mathcal{S}\right\}\cup\left\{\sum_{x\in\mathcal{X}} P(s,x): s\in\mathcal{S}\right\}\right\}.
\end{align}
Observe that, for all $s\in\mathcal{S}$,
\begin{align}
    \sum_{x\in\mathcal{X}} P(s,x) 
    &\geq \sum_{x\in\mathcal{X}} \hat{P}_n(s,x)-\sum_{x\in\mathcal{X}} |\hat{P}_n(s,x)-P(s,x)|\\
    &\geq \sum_{x\in\mathcal{X}} \hat{P}_n(s,x)-\|\hat{P}_n-P\|_1.
\end{align}
In particular, we have that
\begin{equation}
    m_S \geq \left(\min\left\{\sum_{x\in\mathcal{X}} \hat{P}_n(s,x) : s\in\mathcal{S}\right\} - \|\hat{P}_n - P\|_1\right)_+.
\end{equation}
Inequality~\eqref{equa:L1bound_Devroye} shows that, with probability at least $1-\beta$,
\begin{equation}
\label{equa:L1bound_proof_thm_f_pointwise}
    \|\hat{P}_n - P\|_1 \leq \sqrt{\frac{2}{n}\left(|\mathcal{S}|\cdot|\mathcal{X}|-\log\beta\right)},
\end{equation}
and, thus, $m_S \geq \overline{m}_S$ whenever \eqref{equa:L1bound_proof_thm_f_pointwise} holds true. Recall the definitions of $K_{g,u}$ and $L_{g,u}$ in \eqref{eq:DefSupNorm} and \eqref{eq:DefLipschitzConstant}, respectively. It is straightforward to verify that the mappings $u \mapsto K_{f,u}$ and $u \mapsto L_{f,u}$ are non-decreasing. Since the mapping $u \mapsto u^{-1}$ is non-increasing, we conclude that
\begin{equation}
u \mapsto C_{f,u}=2 K_{f,u^{-1}} + (2u^{-1} + 1) L_{f,u^{-1}}
\end{equation}
is non-increasing. Therefore, under \eqref{equa:L1bound_proof_thm_f_pointwise}, we have
\begin{align}
    \label{equa:C_leq_C_bar} C_{f,m_S} \leq C_{f,\overline{m}_S}.
\end{align}
Combining \eqref{equa:Lip_proof_lem_pointwise_L}, \eqref{equa:L1bound_proof_thm_f_pointwise} and \eqref{equa:C_leq_C_bar}, we have that, with probability at least $1-\beta$,
\begin{align}
    |\calL_f(\hat{P}_n,W)-\calL_f(P,W)| 
    &\leq C_{f,\overline{m}_S} \sqrt{\frac{2}{n}\left(|\mathcal{S}|\cdot|\mathcal{X}|-\log\beta\right)}.
\end{align}
Mutatis mutandis, it can be shown that
\begin{align}
    |\calU_f(\hat{P}_n,W)-\calU_f(P,W)| 
    &\leq  C_{f,\overline{m}_X} \sqrt{\frac{2}{n}\left(|\mathcal{S}|\cdot|\mathcal{X}|-\log\beta\right)},
\end{align}
whenever \eqref{equa:L1bound_proof_thm_f_pointwise} holds true.
\end{proof}

\subsection{Proposition~\ref{prop:PreProcessingUtilityReduction}}
\label{Appendix:ProofPreprocessing}

\begin{proof}
Assume that $\tilde{S}$, $\tilde{X}$ and $\tilde{X}_0$ satisfy that $P_{\tilde{S},\tilde{X}} = \hat{P}_n$ and $\tilde{X}_0 = \Pi_{\gamma}(\tilde{X})$. With this notation, it can be verified that $\tilde{S} \to \tilde{X} \to \tilde{X}_0$ and 
\begin{equation}
\label{eq:PreprocessingProofPigamma}
    P_{\tilde{S},\tilde{X}_0} = \hat{P}_n P_{\tilde{X}_0|\tilde{X}} = \hat{P}_{\gamma},
\end{equation}
where, by definition,
\begin{equation}
    \left(\hat{P}_n P_{\tilde{X}_0|\tilde{X}}\right)(s,x_0) = \sum_{x\in\mathcal{X}} \hat{P}_n(s,x) P_{\tilde{X}_0|\tilde{X}}(x_0|x).
\end{equation}
By assumption, $\mathcal{W}^*(\hat{P}_{\gamma};\epsilon) \neq \emptyset$. Let $\tilde{W}\in \mathcal{W}^*(\hat{P}_\gamma;\epsilon)$, i.e.,
\begin{equation}
\label{eq:PreprocessingProofTildeWProperties}
    \mathcal{L}_{f}(\hat{P}_{\gamma},\tilde{W})\leq \epsilon \quad \text{and} \quad \mathcal{U}_{f}(\hat{P}_{\gamma},\tilde{W})=\mathsf{H}_{f}(\hat{P}_{\gamma};\epsilon).
\end{equation}
Let $\tilde{Y}_0$ be such that $\tilde{S} \to \tilde{X} \to \tilde{X}_0 \to \tilde{Y}_0$ and $P_{\tilde{Y}_0 | \tilde{X}_0} = \tilde{W}$. Observe that $P_{\tilde{Y}_0 | \tilde{X}} = P_{\tilde{X}_0|\tilde{X}}\tilde{W}$. In particular,
\begin{equation}
    P_{\tilde{S},\tilde{Y}_0} = \hat{P}_n P_{\tilde{Y}_0 | \tilde{X}} = \hat{P}_n P_{\tilde{X}_0|\tilde{X}}\tilde{W} = \hat{P}_\gamma \tilde{W},
\end{equation}
where the last equality follows from \eqref{eq:PreprocessingProofPigamma}. Therefore,
\begin{equation}
\label{eq:PreprocessingProofL}
     \mathcal{L}_{f}(\hat{P}_{\gamma},\tilde{W}) = I_f(P_{\tilde{S},\tilde{Y}_0}) = \mathcal{L}_{f}(\hat{P}_n,P_{\tilde{Y}_0 | \tilde{X}}).
\end{equation}
Furthermore, Lemma~\ref{lem:cut_lesslikely_symbols} implies that
\begin{equation}
\label{eq:PreprocessingProofU}
    \mathcal{U}_{f}(\hat{P}_{\gamma},\tilde{W}) = I_f(P_{\tilde{X}_0,\tilde{Y}_0}) = I_f(P_{\tilde{X},\tilde{Y}_0}) = \mathcal{U}_{f}(\hat{P}_n,P_{\tilde{Y}_0 | \tilde{X}}).
\end{equation}
By \eqref{eq:PreprocessingProofTildeWProperties}, \eqref{eq:PreprocessingProofL} and \eqref{eq:PreprocessingProofU}, we have
\begin{equation}
\label{eq:PreprocessingProofFeasible}
    \mathcal{L}_{f}(\hat{P}_n,P_{\tilde{Y}_0 | \tilde{X}}) \leq \epsilon \quad \text{and} \quad \mathcal{U}_{f}(\hat{P}_n,P_{\tilde{Y}_0 | \tilde{X}}) = \mathsf{H}_{f}(\hat{P}_{\gamma};\epsilon).
\end{equation}
Recall that, by definition, $\mathsf{H}_{f}(\hat{P}_n;\epsilon) = \sup_W \mathcal{U}_{f}(\hat{P}_n,W)$ where the supremum is taken over all $W\in\mathcal{W}$ such that $\mathcal{L}_f(\hat{P}_n,W) \leq \epsilon$. Thus, \eqref{eq:PreprocessingProofFeasible} implies that
\begin{equation}
    \mathsf{H}_{f}(\hat{P}_n;\epsilon) \geq \mathcal{U}_{f}(\hat{P}_n,P_{\tilde{Y}_0 | \tilde{X}}) = \mathsf{H}_{f}(\hat{P}_{\gamma};\epsilon),
\end{equation}
as we wanted to prove.
\end{proof}
\subsection{Lemma~\ref{lem:LipschitzArimotoMI}}
\label{Appendix:ProofAMI}

\begin{proof}
For ease of notation, let $\Delta_L \defined |\calL^{\rm A}_{\alpha}(Q_1,W)-\calL^{\rm A}_{\alpha}(Q_2,W)|$. For each $i\in\{1,2\}$, let $S_i \to X_i \to Y_i$ be such that $P_{S_i,X_i} = Q_i$ and $P_{Y_i|X_i} = W$. With this notation,
\begin{align}
    \Delta_L  &= |I^{\rm A}_{\alpha}(P_{S_1,Y_1}) - I^{\rm A}_{\alpha}(P_{S_2,Y_2})|\\
    &= \frac{\alpha}{\alpha-1} \left|\log \frac{\sum_y \|P_{S_1,Y_1}(\cdot,y)\|_\alpha}{\sum_y \|P_{S_2,Y_2}(\cdot,y)\|_\alpha} - \log \frac{\|P_{S_1}(\cdot)\|_\alpha}{\|P_{S_2}(\cdot)\|_\alpha}\right|,
\end{align}
where the second equality follows directly from the definition of Arimoto's mutual information. By the triangle inequality, we obtain that $\displaystyle \frac{\alpha-1}{\alpha} \Delta_L \leq {\rm I} + {\rm II}$ where
\begin{align}
    {\rm I} &\defined \left|\log\frac{\sum_y \|P_{S_1,Y_1}(\cdot,y)\|_\alpha}{\sum_y \|P_{S_2,Y_2}(\cdot,y)\|_\alpha}\right|,\\
    {\rm II} &\defined \left|\log\frac{\|P_{S_1}(\cdot)\|_\alpha}{\|P_{S_2}(\cdot)\|_\alpha}\right|.
\end{align}
Notice that if $a,b>0$, then $\displaystyle \left|\log\frac{a}{b}\right| \leq \frac{|a-b|}{\min\{a,b\}}$. Hence,
\begin{equation}
\label{eq:IAlphaLeakage}
    {\rm I} \leq \frac{\left|\sum_y \|P_{S_1,Y_1}(\cdot,y)\|_\alpha-\sum_y \|P_{S_2,Y_2}(\cdot,y)\|_\alpha\right|}{\min\{\sum_y \|P_{S_1,Y_1}(\cdot,y)\|_\alpha,\sum_y \|P_{S_2,Y_2}(\cdot,y)\|_\alpha\}}.
\end{equation}
Observe that, by the triangle inequality,
\begin{equation}
    \bigg|\sum_{y\in\mathcal{Y}} \|P_{S_1,Y_1}(\cdot,y)\|_\alpha-\sum_{y\in\mathcal{Y}} \|P_{S_2,Y_2}(\cdot,y)\|_\alpha\bigg| \leq \sum_{y\in\mathcal{Y}} \Big|\|P_{S_1,Y_1}(\cdot,y)\|_\alpha - \|P_{S_2,Y_2}(\cdot,y)\|_\alpha\Big|.
\end{equation}
By Minkowski's inequality, $| \|a\|_\alpha - \|b\|_\alpha| \leq \|a-b\|_\alpha$ for every $a,b\in\mathbb{R}^n$ and $\alpha \geq 1$. Therefore,
\begin{align}
    \bigg|\sum_{y\in\mathcal{Y}} \|P_{S_1,Y_1}(\cdot,y)\|_\alpha-\sum_{y\in\mathcal{Y}} \|P_{S_2,Y_2}(\cdot,y)\|_\alpha\bigg| &\leq \sum_{y\in\mathcal{Y}} \|P_{S_1,Y_1}(\cdot,y) - P_{S_2,Y_2}(\cdot,y)\|_\alpha\\
    &\leq \sum_{y\in\mathcal{Y}} \|P_{S_1,Y_1}(\cdot,y) - P_{S_2,Y_2}(\cdot,y)\|_1,\\
    \label{eq:NumeratorAlphaLeakage1} &= \|P_{S_1,Y_1} - P_{S_2,Y_2}\|_1,
\end{align}
where the second inequality follows from the fact that $\|a\|_\alpha \leq \|a\|_1$ for all $a\in\mathbb{R}^n$ and $\alpha\geq1$. By assumption, we have that $P_{Y_1|X_1} = W = P_{Y_2|X_2}$. Hence, Lemma~\ref{lem:VariationSY2SX} implies that
\begin{equation}
\label{eq:AlphaLeakageLemmaAux1}
\|P_{S_1,Y_1} - P_{S_2,Y_2}\|_1 \leq \|Q_1-Q_2\|_1.
\end{equation}
This leads to
\begin{align}
    \left|\sum_{y\in\mathcal{Y}} \|P_{S_1,Y_1}(\cdot,y)\|_\alpha-\sum_{y\in\mathcal{Y}} \|P_{S_2,Y_2}(\cdot,y)\|_\alpha\right| \leq  \|Q_1-Q_2\|_1. \label{eq:NumeratorAlphaLeakage}
\end{align}
For $i\in\{1,2\}$, Minkowski's inequality implies that
\begin{equation}
    \sum_{y\in\mathcal{Y}} \|P_{S_i,Y_i}(\cdot,y)\|_\alpha \geq \bigg\|\sum_{y\in\mathcal{Y}} P_{S_i,Y_i}(\cdot,y)\bigg\|_\alpha = \|P_{S_i}(\cdot)\|_\alpha.
\end{equation}
The generalized mean inequality implies that
\begin{equation}
    \|P_{S_i}(\cdot)\|_\alpha \geq |\mathcal{S}|^{1/\alpha-1} \|P_{S_i}(\cdot)\|_1 = |\mathcal{S}|^{1/\alpha-1},
\end{equation}
as $\|P_{S_i}(\cdot)\|_1 = \sum_s P_{S_i}(s) = 1$. Hence,
\begin{equation}
\label{eq:DenominatorAlphaLeakage}
    \sum_{y\in\mathcal{Y}} \|P_{S_i,Y_i}(\cdot,y)\|_\alpha \geq |\mathcal{S}|^{1/\alpha-1}.
\end{equation}
Therefore, by plugging \eqref{eq:NumeratorAlphaLeakage} and \eqref{eq:DenominatorAlphaLeakage} in \eqref{eq:IAlphaLeakage},
\begin{equation}
    {\rm I} \leq |\mathcal{S}|^{1-1/\alpha} \|Q_1 - Q_2\|_1.
\end{equation}
Similarly, we have that
\begin{equation}
    {\rm II} \leq \frac{|\|P_{S_2}(\cdot)\|_\alpha - \|P_{S_1}(\cdot)\|_\alpha|}{\min\{\|P_{S_2}(\cdot)\|_\alpha,\|P_{S_1}(\cdot)\|_\alpha\}}.
\end{equation}
As in \eqref{eq:NumeratorAlphaLeakage1} and \eqref{eq:AlphaLeakageLemmaAux1},
\begin{equation}
    |\|P_{S_2}(\cdot)\|_\alpha - \|P_{S_1}(\cdot)\|_\alpha| \leq \|P_{S_1}-P_{S_2}\|_1 \leq \|Q_1-Q_2\|_1.
\end{equation}
As established in \eqref{eq:DenominatorAlphaLeakage}, for all $i\in\{1,2\}$,
\begin{equation}
    \|P_{S_i}(\cdot)\|_\alpha \geq |\mathcal{S}|^{1/\alpha-1}.
\end{equation}
Thus, $\displaystyle {\rm II} \leq |\mathcal{S}|^{1-1/\alpha} \|Q_1 - Q_2\|_1$. Since $\displaystyle \frac{\alpha-1}{\alpha} \Delta_L \leq {\rm I} + {\rm II}$, we conclude that
\begin{equation}
    |\calL^{\rm A}_{\alpha}(Q_1,W)-\calL^{\rm A}_{\alpha}(Q_2,W)| \leq  \frac{2\alpha}{\alpha-1} |\mathcal{S}|^{1-1/\alpha} \|Q_1 - Q_2\|_1.
\end{equation}
Mutatis mutandis, one can obtain that
\begin{equation}
    |\calU^{\rm A}_{\alpha}(Q_1,W)-\calU^{\rm A}_{\alpha}(Q_2,W)| \leq \frac{2\alpha}{\alpha-1} |\mathcal{X}|^{1-1/\alpha} \|Q_1-Q_2\|_1,
\end{equation}
as required
\end{proof}
\subsection{Lemma~\ref{lem:LipschitzSibsonMI}}
\label{Appendix:ProofSMI}

Consider the following lemma.

\begin{lem}
\label{lem:TotalVariationMismatchedChannel}
If $S_1,S_2$ and $X_1,X_2$ are random variables supported over $\mathcal{S}$ and $\mathcal{X}$ respectively, then,
\begin{equation}
    \|P_{S_1}\cdot P_{X_1|S_1} - P_{S_1}\cdot P_{X_2|S_2}\|_1 \leq 2 \|P_{S_1,X_1} - P_{S_2,X_2}\|_1.
\end{equation}
\end{lem}

\begin{proof}
For ease of notation, let $\Delta \defined \|P_{S_1}\cdot P_{X_1|S_1} - P_{S_1}\cdot P_{X_2|S_2}\|_1$. By the triangle inequality, we have that
\begin{equation}
    \Delta \leq \|P_{S_1}\cdot P_{X_1|S_1} -  P_{S_2}\cdot P_{X_2|S_2}\|_1 + \|P_{S_2}\cdot P_{X_2|S_2} - P_{S_1}\cdot P_{X_2|S_2}\|_1.
\end{equation}
By Lemma~\ref{lem:VariationSXY2SX}, and the fact that $P_{S_i} \cdot P_{X_i|S_i} = P_{S_i,X_i}$ for every $i\in\{1,2\}$,
\begin{equation}
    \Delta \leq \|P_{S_1,X_1}-P_{S_2,X_2}\|_1 + \|P_{S_1}-P_{S_2}\|_1 \leq 2\|P_{S_1,X_1}-P_{S_2,X_2}\|_1,
\end{equation}
as required.
\end{proof}

Now we are in position to prove Lemma~\ref{lem:LipschitzSibsonMI}.

\begin{proof}
For ease of notation, let $\Delta_L \defined |\calL^{\rm S}_{\alpha}(Q_1,W)-\calL^{\rm S}_{\alpha}(Q_2,W)|$. For each $i\in\{1,2\}$, let $S_i \to X_i \to Y_i$ be such that $P_{S_i,X_i} = Q_i$ and $P_{Y_i|X_i} = W$. Assume that $\alpha\in(1,\infty)$. Then
\begin{align}
    \Delta_L  &= |I^{\rm S}_{\alpha}(P_{S_1,Y_1}) - I^{\rm S}_{\alpha}(P_{S_2,Y_2})|\\
    &= \frac{\alpha}{\alpha-1} \left|\log \frac{\sum_y \|P_{S_1}(\cdot)^{1/\alpha} P_{Y_1|S_1}(y|\cdot)\|_\alpha}{\sum_y \|P_{S_2}(\cdot)^{1/\alpha} P_{Y_2|S_2}(y|\cdot)\|_\alpha}\right|,
\end{align}
where the second equality follows directly from the definition of Sibson's mutual information. Notice that if $a,b>0$, then $\displaystyle \left|\log\frac{a}{b}\right| \leq \frac{|a-b|}{\min\{a,b\}}$. Hence,
\begin{equation}
\label{eq:SibsonQuotient}
    \Delta_L \leq \frac{\alpha}{\alpha-1} \frac{|\sum_y \|P_{S_1}(\cdot)^{1/\alpha} P_{Y_1|S_1}(y|\cdot)\|_\alpha - \sum_y \|P_{S_2}(\cdot)^{1/\alpha} P_{Y_2|S_2}(y|\cdot)\|_\alpha|}{\min\{\sum_y \|P_{S_1}(\cdot)^{1/\alpha} P_{Y_1|S_1}(y|\cdot)\|_\alpha,\sum_y \|P_{S_2}(\cdot)^{1/\alpha} P_{Y_2|S_2}(y|\cdot)\|_\alpha\}}.
\end{equation}
By Minkowski's inequality, we have that
\begin{equation}
    \sum_{y\in\mathcal{Y}} \|P_{S_i}(\cdot)^{1/\alpha} P_{Y_i|S_i}(y|\cdot)\|_\alpha \geq \|P_{S_i}(\cdot)^{1/\alpha}\|_\alpha = 1,
\end{equation}
where we used the fact that $\sum_y P_{Y_i|S_i}(y|s) = 1$ for every $s\in\mathcal{S}$ and $i\in\{1,2\}$. Hence, \eqref{eq:SibsonQuotient} becomes
\begin{align}
    \Delta_L &\leq \frac{\alpha}{\alpha-1} \sum_{y\in\mathcal{Y}} \bigg|\|P_{S_1}(\cdot)^{1/\alpha} P_{Y_1|S_1}(y|\cdot)\|_\alpha - \|P_{S_2}(\cdot)^{1/\alpha} P_{Y_2|S_2}(y|\cdot)\|_\alpha\bigg|\\
    &\leq \frac{\alpha}{\alpha-1} \sum_{y\in\mathcal{Y}} \|P_{S_1}(\cdot)^{1/\alpha} P_{Y_1|S_1}(y|\cdot) - P_{S_2}(\cdot)^{1/\alpha} P_{Y_2|S_2}(y|\cdot)\|_\alpha,
\end{align}
where the last inequality follows from another application of Minkowski's inequality. Since $\|a\|_\alpha \leq \|a\|_1$ for all $a\in\mathbb{R}^n$ and $\alpha>1$, we obtain that
\begin{equation}
    \Delta_L \leq \frac{\alpha}{\alpha-1} \sum_{y\in\mathcal{Y}} \|P_{S_1}(\cdot)^{1/\alpha} P_{Y_1|S_1}(y|\cdot) - P_{S_2}(\cdot)^{1/\alpha} P_{Y_2|S_2}(y|\cdot)\|_1.
\end{equation}
Observe that, for each $i\in\{1,2\}$,
\begin{equation}
    P_{S_i}(s)^{1/\alpha}P_{Y_i|S_i}(y|s) = \sum_{x\in\mathcal{X}} P_{S_i}(s)^{1/\alpha} P_{X_i|S_i}(x|s) W(x,y).
\end{equation}
Thus, a straightforward manipulation leads to
\begin{align}
    \Delta_L &\leq \frac{\alpha}{\alpha-1} \sum_{y\in\mathcal{Y}} \sum_{x\in\mathcal{X}} \|P_{S_1}(\cdot)^{1/\alpha} P_{X_1|S_1}(x|\cdot) - P_{S_2}(\cdot)^{1/\alpha} P_{X_2|S_2}(x|\cdot)\|_1 W(x,y)\\
    \label{eq:ProofSibsonIxIIx} &= \frac{\alpha}{\alpha-1} \sum_{x\in\mathcal{X}} \|P_{S_1}(\cdot)^{1/\alpha} P_{X_1|S_1}(x|\cdot) - P_{S_2}(\cdot)^{1/\alpha} P_{X_2|S_2}(x|\cdot)\|_1.
\end{align}
By adding and subtracting the term $P_{S_1}(\cdot)^{1/\alpha}P_{X_2|S_2}(x|\cdot)$ inside the norm in \eqref{eq:ProofSibsonIxIIx}, Minkowski's inequality implies that $\displaystyle \frac{\alpha-1}{\alpha} \Delta_L \leq \sum\nolimits_{x} {\rm I}_x + \sum\nolimits_{x} {\rm II}_x$ where, for each $x\in\mathcal{X}$,
\begin{align}
    {\rm I}_x &\defined \|P_{S_1}(\cdot)^{1/\alpha}(P_{X_1|S_1}(x|\cdot) - P_{X_2|S_2}(x|\cdot))\|_1,\\
    {\rm II}_x &\defined \|(P_{S_1}(\cdot)^{1/\alpha} - P_{S_2}(\cdot)^{1/\alpha}) P_{X_2|S_2}(x|\cdot)\|_1.
\end{align}
Observe that, for each $x\in\mathcal{X}$,
\begin{align}
    {\rm I}_x &= \sum_{s\in\mathcal{S}} P_{S_1}(s)^{1/\alpha} |P_{X_1|S_1}(x|s) - P_{X_2|S_2}(x|s)|\\
    &= \sum_{s\in\mathcal{S}} \frac{P_{S_1}(s)}{P_{S_1}(s)^{1-1/\alpha}} |P_{X_1|S_1}(x|s) - P_{X_2|S_2}(x|s)|.
\end{align}
Since $P_{S_1}(s) \geq m_S$ for all $s\in\mathcal{S}$, we obtain that
\begin{align}
    \sum_{x\in\mathcal{X}} {\rm I}_x &\leq \frac{1}{m_S^{1-1/\alpha}} \sum_{x\in\mathcal{X}} \sum_{s\in\mathcal{S}} P_{S_1}(s) |P_{X_1|S_1}(x|s) - P_{X_2|S_2}(x|s)|\\
    &= \frac{\|P_{S_1}\cdot P_{X_1|S_1} - P_{S_1}\cdot P_{X_2|S_2}\|_1}{m_S^{1-1/\alpha}}.
\end{align}
Therefore, Lemma~\ref{lem:TotalVariationMismatchedChannel} leads to
\begin{equation}
\label{eq:SibsonPartI}
    \sum_{x\in\mathcal{X}} {\rm I}_x \leq \frac{2\|Q_1 - Q_2\|_1}{m_S^{1-1/\alpha}}.
\end{equation}
By definition, ${\rm II}_x = \sum_s |P_{S_1}(s)^{1/\alpha} - P_{S_2}(s)^{1/\alpha}| P_{X_2|S_2}(x|s)$ for every $x\in\mathcal{X}$. Hence,
\begin{equation}
    \sum_{x\in\mathcal{X}} {\rm II}_x = \sum_{s\in\mathcal{S}} |P_{S_1}(\cdot)^{1/\alpha} - P_{S_2}(\cdot)^{1/\alpha}|.
\end{equation}
Since the function $t \mapsto t^{1/\alpha}$ is Lipschitz continuous on $[m_S,1]$ with Lipschitz constant $(\alpha m_S^{1-1/\alpha})^{-1}$,
\begin{equation}
\label{eq:SibsonPartII}
    \sum_{x\in\mathcal{X}} {\rm II}_x \leq \frac{\|P_{S_1} - P_{S_2}\|_1}{\alpha m_S^{1-1/\alpha}} \leq \frac{\|Q_1 - Q_2\|_1}{\alpha m_S^{1-1/\alpha}},
\end{equation}
where the last inequality follows from Lemma~\ref{lem:VariationSXY2SX}. By \eqref{eq:SibsonPartI} and \eqref{eq:SibsonPartII}, we conclude that
\begin{equation}
    \Delta_L \leq \frac{2\alpha+1}{\alpha-1} \frac{\|Q_1 - Q_2\|_1}{m_S^{1-1/\alpha}}.
\end{equation}
Mutatis mutandis, it can be shown that
\begin{equation}
\label{eq:SibsonDeltaUUpperBound}
    |\calU^{\rm S}_{\alpha}(Q_1,W)-\calU^{\rm S}_{\alpha}(Q_2,W)| \leq  \frac{1}{\alpha-1} \frac{\|Q_1-Q_2\|_1}{m_X^{1-1/\alpha}}.
\end{equation}

Recall that $\lim_{\alpha\to\infty} I^{\rm S}_\alpha(P_{X,Y}) = I^{\rm S}_\infty(P_{X,Y})$ \cite{verdu2015alpha}. Therefore, \eqref{eq:UtilityBoundMaximalLeakage} follows after taking the limit $\alpha\to\infty$ in both sides of \eqref{eq:SibsonDeltaUUpperBound}. Using a similar argument as above, it can be shown that
\begin{align}
    \Delta_L &\leq \sum_{s\in\mathcal{S}} \sum_{x\in\mathcal{X}} |P_{X_1|S_1}(x|s) - P_{X_2|S_2}(x|s)|\\
    &\leq \frac{1}{\min_{s'} P_{S_1}(s')} \sum_{s\in\mathcal{S}} \sum_{x\in\mathcal{X}} P_{S_1}(s) |P_{X_1|S_1}(x|s) - P_{X_2|S_2}(x|s)|\\
    &= \frac{\|P_{S_1}\cdot P_{X_1|S_1} - P_{S_1}\cdot P_{X_2|S_2|}\|_1}{\min_{s'} P_{S_1}(s')}.
\end{align}
By applying Lemma~\ref{lem:TotalVariationMismatchedChannel} and noting that $\min_{s'} P_{S_1}(s') = \min_{s'} \sum_x Q_1(s',x)$, the result follows.
\end{proof}
\subsection{Lemma~\ref{Lemma:LipschitzMaximalAlphaLeakage}}
\label{Appendix:ProofMAL}

\begin{proof}
For ease of notation, let $\Delta_L \defined |\calL^{\rm max}_{\alpha}(Q_1,W)-\calL^{\rm max}_{\alpha}(Q_2,W)|$. For each $i\in\{1,2\}$, let $S_i \to X_i \to Y_i$ be such that $P_{S_i,X_i} = Q_i$ and $P_{Y_i|X_i} = W$. Observe that $P_{Y_i|S_i} = P_{X_i|S_i} W$ where, by definition,
\begin{equation}
    \left(P_{X_i|S_i} W\right)(y|s) = \sum_{x\in\mathcal{X}} P_{X_i|S_i}(x|s) W(x,y).
\end{equation}
In particular,
\begin{align}
    \mathcal{L}^{\rm max}_\alpha(Q_i,W) 
    &\defined \sup_{P_{\tilde{S}}} I^{\rm A}_\alpha(P_{\tilde{S}}\cdot P_{Y_i|S_i})\\
    &= \sup_{P_{\tilde{S}}} I^{\rm A}_\alpha(P_{\tilde{S}}\cdot P_{X_i|S_i} W).
\end{align}
By the definition of $\mathcal{L}^{\rm A}_\alpha(Q,W)$, we have that
\begin{equation}
\label{eq:MaxAlphaLeakagetoAlphaLeakage}
    \mathcal{L}^{\rm max}_\alpha(Q_i,W) = \sup_{P_{\tilde{S}}} \mathcal{L}^{\rm A}_\alpha(P_{\tilde{S}}\cdot P_{X_i|S_i},W).
\end{equation}
It is easy to verify that if $\mathcal{I}$ is an arbitrary index set and $a_\iota,b_\iota\in\mathbb{R}$ for all $\iota\in\mathcal{I}$, then
\begin{equation}
  \left|\sup_{\iota\in\mathcal{I}} a_\iota - \sup_{\iota\in\mathcal{I}} b_\iota\right| 
  \leq \sup_{\iota\in\mathcal{I}} |a_\iota - b_\iota|.
\end{equation}
Therefore, \eqref{eq:MaxAlphaLeakagetoAlphaLeakage} implies that
\begin{align}
    \Delta_L &= \left|\sup_{P_{\tilde{S}}} \mathcal{L}^{\rm A}_\alpha(P_{\tilde{S}}\cdot P_{X_1|S_1},W) - \sup_{P_{\tilde{S}}} \mathcal{L}^{\rm A}_\alpha(P_{\tilde{S}}\cdot P_{X_2|S_2},W)\right|\\
    &\leq \sup_{P_{\tilde{S}}} \left|\mathcal{L}^{\rm A}_\alpha(P_{\tilde{S}}\cdot P_{X_1|S_1},W) - \mathcal{L}^{\rm A}_\alpha(P_{\tilde{S}}\cdot P_{X_2|S_2},W)\right|.
\end{align}
By Lemma~\ref{lem:LipschitzArimotoMI}, we obtain that
\begin{equation}
\label{eq:MaximalAlphaLeakageDeltaLReduction}
    \Delta_L \leq \frac{2\alpha}{\alpha-1} |\mathcal{S}|^{1-1/\alpha} \sup_{P_{\tilde{S}}} \|P_{\tilde{S}}\cdot P_{X_1|S_1} - P_{\tilde{S}}\cdot P_{X_2|S_2}\|_1.
\end{equation}
A straightforward computation shows that, for any $P_{\tilde{S}}$,
\begin{align}
    \|P_{\tilde{S}}\cdot P_{X_1|S_1} - P_{\tilde{S}}\cdot P_{X_2|S_2}\|_1 &= \sum_{s\in\mathcal{S}} P_{\tilde{S}}(s) \|P_{X_1|S_1}(\cdot|s) - P_{X_2|S_2}(\cdot|s)\|_1\\
    &\leq \sum_{s\in\mathcal{S}} \|P_{X_1|S_1}(\cdot|s) - P_{X_2|S_2}(\cdot|s)\|_1\\
    &\leq \frac{1}{\min_{s'} P_{S_1}(s')} \sum_{s\in\mathcal{S}} P_{S_1}(s) \|P_{X_1|S_1}(\cdot|s) - P_{X_2|S_2}(\cdot|s)\|_1\\
    &= \frac{\|P_{S_1}\cdot P_{X_1|S_1} - P_{S_1}\cdot P_{X_2|S_2}\|_1}{\min_{s'} P_{S_1}(s')}.
\end{align}
Therefore, Lemma~\ref{lem:TotalVariationMismatchedChannel} implies that
\begin{equation}
\label{eq:MaximalAlphaLeakageSupInq}
    \sup_{P_{\tilde{S}}} \|P_{\tilde{S}}\cdot P_{X_1|S_1} - P_{\tilde{S}}\cdot P_{X_2|S_2}\|_1 \leq \frac{2\|Q_1-Q_2\|}{\min_{s'} P_{S_1}(s')}.
\end{equation}
By \eqref{eq:MaximalAlphaLeakageDeltaLReduction} and \eqref{eq:MaximalAlphaLeakageSupInq}, we conclude that
\begin{equation}
    \Delta_L \leq \frac{4\alpha}{\alpha-1} \frac{|\mathcal{S}|^{1-1/\alpha}}{\min_{s'} P_{S_1}(s')} \|Q_1-Q_2\|_1.
\end{equation}
By noting that $\min_{s'} P_{S_1}(s') = \min_{s'} \sum_x Q_1(s',x)$, \eqref{eq:MaxAlphaLeakageLBound} follows. Finally, note that $\mathcal{U}^{\text{max}}_\alpha(Q,W)$ only depends on $W$. Therefore, $\mathcal{U}^{\text{max}}_\alpha(Q_1,W) = \mathcal{U}^{\text{max}}_\alpha(Q_2,W)$ which trivially leads to \eqref{eq:MaxAlphaLeakageUBound}.
\end{proof}

\section{Proofs from Section~\ref{Section:convergence}}
\label{Appendix:ProofsIV}
\subsection{Lemma~\ref{lem:cont_UandL}}
\label{Appendix:Proofcont_UandL}

Lemma~\ref{lem:cont_UandL} is an immediate consequence of the following standard result whose proof is included for the sake of completeness.

\begin{lem}
Let $(\mathcal{A},d_{\mathcal{A}})$ and $(\mathcal{B},d_{\mathcal{B}})$ be two metric spaces. If $f:\mathcal{A}\times\mathcal{B}\to\Reals$ satisfies that
\begin{itemize}
    \item[\textnormal{(i)}] $|f(a_1,b) - f(a_2,b)| \leq d_{\mathcal{A}}(a_1,a_2)$ for all $a_1,a_2\in\mathcal{A}$ and $b\in\mathcal{B}$;
    \item[\textnormal{(ii)}] $f(a,\cdot)$ is continuous over $\mathcal{B}$ for all $a\in\mathcal{A}$;
\end{itemize}
then $f(\cdot,\cdot)$ is continuous over $\mathcal{A}\times\mathcal{B}$.
\end{lem}

\begin{proof}
In order to prove the continuity of $f(\cdot,\cdot)$ over $\mathcal{A}\times\mathcal{B}$, we will show that for any sequence $\{(a_n,b_n)\}_{n=0}^\infty$ such that $\lim_n (a_n,b_n) = (a_0,b_0)$, it holds true that $\lim_n f(a_n,b_n) = f(a_0,b_0)$.

Let $\{(a_n,b_n)\}_{n=0}^\infty$ be such that $\lim_n (a_n,b_n) = (a_0,b_0)$. By the triangle inequality, for every $n\in\mathbb{N}$,
\begin{align}
    |f(a_n,b_n) - f(a_0,b_0)| &\leq |f(a_n,b_n) - f(a_0,b_n)| + |f(a_0,b_n) - f(a_0,b_0)|\\
    &\leq d_{\mathcal{A}}(a_n,a_0) + |f(a_0,b_n) - f(a_0,b_0)|,
\end{align}
where the last inequality follows from assumption \textnormal{(i)}. By assumption \textnormal{(ii)}, $f(a_0,\cdot)$ is continuous over $\mathcal{B}$. Since $\lim_n (a_n,b_n) = (a_0,b_0)$ is equivalent to $\lim_n a_n = a_0$ and $\lim_n b_n = b_0$, we conclude that
\begin{equation}
    \lim_{n\to\infty} |f(a_n,b_n) - f(a_0,b_0)| = 0,
\end{equation}
as required.
\end{proof}
\subsection{Proposition~\ref{Prop:ContinuityH}}
\label{Appendix:ProofContinuousH_general}

We start recalling an elementary fact about convergent sequences in metric spaces, see, e.g., \cite[Exercise~2.4.11]{royden1968real}.

\begin{thm}
\label{thm:conv_seq_sub}
Let $\mathcal{M}$ be a metric space. A sequence $(a_n)_{n=1}^\infty\subset\mathcal{M}$ converges to $a\in\mathcal{M}$ if and only if every subsequence of $(a_n)_{n=1}^\infty$ has a further subsequence which converges to $a$.
\end{thm}

We proceed with the proof of Proposition~\ref{Prop:ContinuityH}.

\begin{proof}
In order to prove the continuity of $\mathsf{H}(\cdot;\epsilon)$ over $\{Q\in\mathcal{Q}: \epsilon_{\min}(Q)< \epsilon\}$, we will show that for any sequence $(Q_n)_{n=0}^{\infty} \subset\{Q\in\mathcal{Q}: \epsilon_{\min}(Q)< \epsilon \}$ such that
\begin{equation}
\label{eq:ConvergenceProofContinuity}
    \lim_{n\to\infty} \| Q_n - Q_0\|_1 = 0,
\end{equation}
it holds true that
\begin{equation}
\label{eq:ConvergenceHProofContinuity}
    \lim_{n\to\infty} \mathsf{H}(Q_n;\epsilon) = \mathsf{H}(Q_0;\epsilon).
\end{equation}
In order to prove \eqref{eq:ConvergenceHProofContinuity}, Theorem~\ref{thm:conv_seq_sub} implies that it is enough to show that any subsequence $(\mathsf{H}(Q_{n_k};\epsilon))_{k=1}^{\infty}$ has a further subsequence converging to $\mathsf{H}(Q_0;\epsilon)$.

Let $(n_k)_{k=1}^{\infty}\subset\mathbb{N}$ be given. By Remark~\ref{rem:sup_max_Wn}, there exists $(W^*_{n_{k}})_{k=1}^{\infty}\subset\mathcal{W}_N$ such that, for each $k\geq1$, we have that $\mathcal{L}(Q_{n_k},W^*_{n_k})\leq\epsilon$ and
\begin{equation}
\label{eq:HUProofContinuity}
    \mathsf{H}(Q_{n_k};\epsilon) = \mathcal{U}(Q_{n_k},W^*_{n_k}).
\end{equation}
Since the space $\mathcal{W}_N$ is compact, there exists a further subsequence $(n_{k_j})_{j=1}^\infty$ such that
\begin{equation}
\label{eq:ConvergenceSubsequenceProofContinuity}
    \lim_{j\to\infty} W^*_{n_{k_j}} = W_0,
\end{equation}
for some $W_0\in\mathcal{W}_N$. Lemma~\ref{lem:cont_UandL} shows that $\mathcal{U}(\cdot,\cdot)$ is continuous over $\mathcal{Q}\times\mathcal{W}_N$. Then \eqref{eq:ConvergenceProofContinuity} and \eqref{eq:ConvergenceSubsequenceProofContinuity} imply that 
\begin{equation}
    \mathcal{U}(Q_0,W_0) = \lim_{j\to\infty} \mathcal{U}(Q_{n_{k_j}},W^*_{n_{k_j}}) = \lim_{j\to\infty} \mathsf{H}(Q_{n_{k_j}};\epsilon),
\end{equation}
where the second equality follows from \eqref{eq:HUProofContinuity}. A similar reasoning shows that
\begin{equation}
    \mathcal{L}(Q_0,W_0) = \lim_{j\to\infty} \mathcal{L}(Q_{n_{k_j}},W^*_{n_{k_j}}) \leq \epsilon.
\end{equation}
Therefore, by the maximality of $\mathsf{H}(Q_0;\epsilon)$,
\begin{equation}
\label{eq:ProofContinuityHgeq}
    \lim_{j\to\infty} \mathsf{H}(Q_{n_{k_j}};\epsilon) = \mathcal{U}(Q_0,W_0) \leq \mathsf{H}(Q_0;\epsilon).
\end{equation}

Next, we show that $\lim_j \mathsf{H}(Q_{n_{k_j}};\epsilon) \geq \mathsf{H}(Q_0;\epsilon)$. Recall that, by assumption, $\epsilon_{\min}(Q_0)<\epsilon$. Let $\delta>0$ be such that $\epsilon_{\min}(Q_0)\leq\epsilon-\delta$. By condition (\textbf{C.}\ref{cond:comp}), there exists $W_0'\in\mathcal{W}_N$ such that $\mathcal{L}(Q_0,W_0') \leq \epsilon - \delta$ and $\mathcal{U}(Q_0,W_0') = \mathsf{H}(Q_0;\epsilon-\delta)$. By the Lipschitz continuity given in condition (\textbf{C.}\ref{cond:lipcont}) and \eqref{eq:ConvergenceProofContinuity},
\begin{align}
    \lim_{j\to\infty} \mathcal{L}(Q_{n_{k_j}},W_0') &= \mathcal{L}(Q_0,W_0') \leq \epsilon-\delta,\\
    \label{eq:ProofContinuityHepsilondeltaA} \lim_{j\to\infty} \mathcal{U}(Q_{n_{k_j}},W_0') &= \mathcal{U}(Q_0,W_0') = \mathsf{H}(Q_0;\epsilon-\delta).
\end{align}
In particular, for $j$ large enough we have that $\mathcal{L}(Q_{n_{k_j}},W_0') \leq \epsilon$ and, by the maximality of $\mathsf{H}(Q_{n_{k_j}};\epsilon)$,
\begin{align}
\label{eq:ProofContinuityHepsilondelta}
    \lim_{j\to\infty} \mathsf{H}(Q_{n_{k_j}};\epsilon) \geq \lim_{j\to\infty} \mathcal{U}(Q_{n_{k_j}},W_0') = \mathsf{H}(Q_0;\epsilon-\delta),
\end{align}
where the last equality follows from \eqref{eq:ProofContinuityHepsilondeltaA}. Recall that, by condition (\textbf{C.}\ref{cond:cont}), the mapping $\mathsf{H}(Q_0;\cdot)$ is continuous over $[\epsilon_{\min}(Q_0),\infty)$. Hence, by taking limits in \eqref{eq:ProofContinuityHepsilondelta},
\begin{equation}
\label{eq:ProofContinuityHleq}
    \lim_{j\to\infty} \mathsf{H}(Q_{n_{k_j}};\epsilon) \geq \lim_{\delta\downarrow0} \mathsf{H}(Q_0;\epsilon-\delta) = \mathsf{H}(Q_0;\epsilon).
\end{equation}
By combining \eqref{eq:ProofContinuityHgeq} and \eqref{eq:ProofContinuityHleq}, we conclude that
\begin{equation}
    \lim_{j\to\infty} \mathsf{H}(Q_{n_{k_j}};\epsilon) = \mathsf{H}(Q_0;\epsilon),
\end{equation}
as required.
\end{proof}
\rev{
\subsection{Corollary~\ref{cor::unif_conv_PUF}}
\label{Appendix:Proof_cor::unif_conv_PUF}

Before we prove Corollary~\ref{cor::unif_conv_PUF}, let us recall two standard results from analysis on metric spaces, see, e.g., Chapter~0.1 in \cite{resnick2013extreme} and Lemma 14 in \cite{tan2011large}, respectively.

\begin{lem}
\label{lem::uni_conv_seq_func}
Let $f_n: [a,b]\to \Reals$ be a non-decreasing function for each $n\in\mathbb{N}$. If there exists a continuous function $f:[a,b]\to\Reals$ such that $\displaystyle \lim_{n\to\infty} f_n(x) = f(x)$ for all $x\in [a,b]$, then $f_n$ converges uniformly to $f$.
\end{lem}

\begin{lem}
\label{lem:inf_cont_overcomp_cont}
Let $\mathcal{A}$ and $\mathcal{B}$ be two metric spaces with $\mathcal{B}$ compact. If $f: \mathcal{A}\times\mathcal{B} \to \Reals$ is continuous, then the function $g:\calA\to\Reals$ defined by $\displaystyle g(a) \defined \inf_{b\in\mathcal{B}} f(a,b)$ is also continuous.
\end{lem}

Observe that the infimum in the previous lemma can be replaced by a supremum. We proceed with the proof of Corollary~\ref{cor::unif_conv_PUF}.
\begin{proof}
We first introduce 
\begin{align}
    \epsilon_{\max}(Q)\defined \max\{\mathcal{L}(Q,W) :\ W\in \mathcal{W}_N\}.
\end{align}
Observe that, by the maximality of $\epsilon_{\max}(Q)$, for all $\epsilon\geq\epsilon_{\max}(Q)$,
\begin{equation}
\label{eq:MaximalityEpsilonMax}
    \mathsf{H}(Q;\epsilon)=\mathsf{H}(Q;\epsilon_{\max}(Q)).
\end{equation}
By Lemma~\ref{lem:cont_UandL} the function $\mathcal{L}(\cdot,\cdot)$ is continuous over $\mathcal{Q}\times\mathcal{W}_N$. Since $\mathcal{W}_N$ is compact, Lemma~\ref{lem:inf_cont_overcomp_cont} implies that $\epsilon_{\max}(\cdot)$ is continuous over $\mathcal{Q}$. By hypothesis $\lim_{n} P_n = P$, thus $\lim_{n} \epsilon_{\max}(P_n) = \epsilon_{\max}(P)$ and, in particular, there exists $\epsilon_1 > \epsilon_0$ such that $\epsilon_{\max}(P_n) \leq \epsilon_1$ for all $n\in\mathbb{N}$. By \eqref{eq:MaximalityEpsilonMax}, for all $\epsilon\geq\epsilon_1$ and all $n\in\mathbb{N}$,
\begin{equation}
\mathsf{H}(P_n;\epsilon) = \mathsf{H}(P_n;\epsilon_{\max}(P_n)) = \mathsf{H}(P_n;\epsilon_1). \label{eq::H_max_equal_Pn}
\end{equation}
Observe that an analogous equality holds for $P$. Hence, we obtain that, for all $n\in\mathbb{N}$,
\begin{equation}
\label{eq:UniformConvergenceAux}
    \sup_{\epsilon\geq\epsilon_1} |\mathsf{H}(P_n;\epsilon) - \mathsf{H}(P;\epsilon)| = |\mathsf{H}(P_n;\epsilon_1) - \mathsf{H}(P_n;\epsilon_1)|.
\end{equation}
As established by Proposition~\ref{Prop:ContinuityH}, the right hand side of \eqref{eq:UniformConvergenceAux} vanishes as $n$ goes to infinity. Hence, we conclude that $\mathsf{H}(P_n;\cdot)$ converge uniformly to $\mathsf{H}(P;\cdot)$ over $[\epsilon_1,\infty)$. Finally, recall that the function $\mathsf{H}(P_n;\cdot)$ is non-decreasing over $[\epsilon_0,\epsilon_1]$ for each $n\in\mathbb{N}$. Furthermore, by condition~(\textbf{C.}\ref{cond:cont}), the function $\mathsf{H}(P;\cdot)$ is continuous over $[\epsilon_0,\epsilon_1]$. Therefore, Proposition~\ref{Prop:ContinuityH} and Lemma~\ref{lem::uni_conv_seq_func} imply that $\mathsf{H}(P_n;\cdot)$ converges uniformly to $\mathsf{H}(P;\cdot)$ over $[\epsilon_0,\epsilon_1]$. This shows that uniform convergence holds over $[\epsilon_0,\infty)$.
\end{proof}
}

\subsection{Theorem~\ref{thm:optimaltooptimal_mainthm}}
\label{Appendix:Proofoptimaltooptimal_mainthm}

The following lemma follows directly from \rev{Lemma~\ref{lem:inf_cont_overcomp_cont}} by noticing that $\mathcal{L}(\cdot,\cdot)$ is continuous over $\mathcal{Q}\times \mathcal{W}_N$ and $\mathcal{W}_N$ is compact.

\begin{lem}
\label{lem:epsilonmin}
Assume that conditions (\textbf{C.}\ref{cond:cont}--\ref{cond:comp}) hold true for a given closed set $\mathcal{Q}\subseteq\mathcal{P}$ and a given $N\in\mathbb{N}$. Then $\epsilon_{\min}(\cdot)$, as given in \eqref{eq:epsilonminReduced}, is continuous over $\mathcal{Q}$.
\end{lem}

We prove the following lemma which will be used in the proof of Theorem~\ref{thm:optimaltooptimal_mainthm}.
\begin{lem}
\label{lem:W_star_compact}
Assume that conditions (\textbf{C.}\ref{cond:cont}) and (\textbf{C.}\ref{cond:comp}) hold true for a given closed set $\mathcal{Q}\subseteq\mathcal{P}$ and a given $N\in\mathbb{N}$. Then the set $\mathcal{W}_N^*(P;\epsilon)$, defined in \eqref{eq:opt_pri_mech_W_N_star_defn}, is compact.
\end{lem}
\begin{proof}
Observe that
\begin{equation}
\label{eq:CorollaryOptimaltoOptimalSplit}
    \mathcal{W}^*_N(P;\epsilon) = \{W\in\mathcal{W}_N : \mathcal{L}(P,W) \leq \epsilon\} \cap \{W\in\mathcal{W}_N : \mathcal{U}(P,W) = \mathsf{H}(P;\epsilon)\}.
\end{equation}
Recall that, by condition (\textbf{C.}\ref{cond:cont}), the mappings $\mathcal{L}(P,\cdot)$ and $\mathcal{U}(P,\cdot)$ are continuous. Since both sets in the RHS of \eqref{eq:CorollaryOptimaltoOptimalSplit} are the preimage of a closed set under a continuous function, they are closed. Therefore, $\mathcal{W}^*_N(P;\epsilon)$ is closed as the finite intersection of closed sets is closed. Finally, we conclude that $\mathcal{W}^*_N(P;\epsilon)$ is compact since it is a closed subset of the compact set $\mathcal{W}_N$.
\end{proof}

Observe that, by Lemma~\ref{lem:epsilonmin}, $\lim_n \epsilon_{\min}(P_n) = \epsilon_{\min}(P)$ whenever $\lim_n P_n = P$. Hence, if $\epsilon_{\min}(P) < \epsilon$, then $\epsilon_{\min}(P_n) < \epsilon$ for $n$ large enough. Therefore, under the assumptions of Theorem~\ref{thm:optimaltooptimal_mainthm}, for $n$ large enough we have that the set $\mathcal{W}^*(P_n;\epsilon)$ is non-empty and, by condition (\textbf{C.}\ref{cond:comp}), $\mathcal{W}_N^*(P_n;\epsilon)$ is non-empty as well. Now we are in position to prove Theorem~\ref{thm:optimaltooptimal_mainthm}.

\begin{proof}
We assume, without loss of generality, that $\mathcal{W}_N^*(P_n;\epsilon)$ is non-empty for all $n\in\mathbb{N}$. In order to reach contradiction, assume that \eqref{eq:optimaltooptimal_mainthm} does not hold true, i.e., there exists $W_n^*\in\mathcal{W}_N^*(P_n;\epsilon)$ such that
\begin{equation}
    \limsup_{n\to\infty} \textsf{dist}(W_n,\mathcal{W}_N^*(P;\epsilon)) > 0.
\end{equation}
Since $\mathcal{W}_N$ and $\mathcal{W}_N^*(P;\epsilon)$ are compact from Lemma~\ref{lem:W_star_compact}, a routine argument shows that there exist a subsequence $(W^*_{n_k})_{k=1}^\infty$ and $W_0\in\mathcal{W}_N$ such that $\lim_k W^*_{n_k} = W_0$ and $W_0\not\in \mathcal{W}_N^*(P;\epsilon)$. Lemma~\ref{lem:cont_UandL} shows that $\mathcal{L}(\cdot,\cdot)$ is continuous over $\mathcal{Q}\times\mathcal{W}_N$. Therefore, we have
\begin{align}
    \mathcal{L}(P,W_0) 
    =\lim_{k\to\infty} \mathcal{L}(P_{n_k},W^*_{n_k})
    \leq \epsilon.
\end{align}
Hence, by the maximality of the privacy-utility function, $\mathcal{U}(P,W_0) \leq \mathsf{H}(P;\epsilon)$. Using a similar argument, one can show that
\begin{equation}
\label{eq:optimaloptimal_112387}
    \mathcal{U}(P,W_0) = \lim_{k\to\infty} \mathcal{U}(P_{n_k},W^*_{n_k}) = \lim_{k\to\infty} \mathsf{H}(P_{n_k},\epsilon)=\mathsf{H}(P;\epsilon),
\end{equation}
where the last equality follows from the fact that the mapping $\mathsf{H}(\cdot;\epsilon)$ is continuous, as established in Proposition~\ref{Prop:ContinuityH}. Altogether,
\begin{equation}
\mathcal{L}(P,W_0) \leq \epsilon \quad \text{and} \quad \mathcal{U}(P,W_0) = \mathsf{H}(P;\epsilon).
\end{equation}
Therefore, $W_0\in \mathcal{W}_N^*(P;\epsilon)$ which contradicts the fact that $W_0\notin \mathcal{W}_N^*(P;\epsilon)$.

Furthermore, by the strong law of large numbers, $(\hat{P}_n)_{n=1}^\infty$ converges almost surely to $P$ as $n$ goes to infinity. Therefore, if $P_n=\hat{P}_n$, then
\begin{align}
    \Pr\left(\lim_{n\to\infty}\mathsf{dist}(W^*_n,\mathcal{W}_N^*(P;\epsilon)) = 0\right)=1,
\end{align}
as claimed.
\end{proof}

Note that in the previous proof we implicitly assume that, for each $n\in\mathbb{N}$, the (possibly random) privacy mechanism $W^*_n \in \mathcal{W}_N^*(\hat{P}_n;\epsilon)$ is a random variable, i.e., a measurable function. This is a minor subtlety given the rareness of non-measurable sets (functions) \cite{oxtoby2013measure}.

\subsection{Theorem~\ref{thm:setvalued_optimalprivacymechanism}}
\label{Appendix:setvalued_optimalprivacymechanism}

Before we prove Theorem~\ref{thm:setvalued_optimalprivacymechanism}, let us recall some definitions and theorems from set-valued analysis \cite{aubin2009set}.

Given two metric spaces $\mathcal{A}$ and $\mathcal{B}$, a set-valued mapping $F$ from $\mathcal{A}$ to $\mathcal{B}$, denoted by $F:\mathcal{A}\leadsto\mathcal{B}$, is a function from $\mathcal{A}$ to the subsets of $\mathcal{B}$. The domain of a set-valued mapping $F:\mathcal{A}\leadsto\mathcal{B}$ is defined as
\begin{equation}
    \text{Dom}(F) \defined \{a\in\mathcal{A} : F(a) \neq \emptyset\}.
\end{equation}
For $r>0$, $\mathcal{A}_r(a) \defined \{a'\in\mathcal{A} : d_{\mathcal{A}}(a,a') < r\}$ where $d_{\mathcal{A}}$ denotes the metric associated to $\mathcal{A}$.

\begin{defn}[Definition 1.4.1, \cite{aubin2009set}]
A set-valued mapping $F:\mathcal{A}\leadsto\mathcal{B}$ is called upper semicontinuous at $a\in\text{Dom}(F)$ if and only if for every neighborhood $\mathcal{N}$ of $F(a)$ there exists $r>0 $ such that $F(a')\subset\mathcal{N}$ for all $a'\in \mathcal{A}_r(a)$. A set-valued mapping is said to be upper semicontinuous if and only if it is upper semicontinuous at every point of its domain.
\end{defn}

\begin{defn}[Definition 1.4.2, \cite{aubin2009set}]
\label{defn:aubin2009lower_semicontinuous}
A set-valued mapping $F: \mathcal{A} \leadsto \mathcal{B}$ is called lower semicontinuous at $a\in\text{Dom}(F)$ if and only if for every $b\in F(a)$ and every sequence $(a_n)_{n=1}^\infty\subset\text{Dom}(F)$ with $\lim_n a_n = a$, there exists a sequence $(b_n)_{n=1}^\infty$ such that $b_n \in F(a_n)$ for each $n\in\mathbb{N}$ and $\lim_n b_n = b$. A set-valued mapping is said to be lower semicontinuous if and only if it is lower semicontinuous at every point of its domain.
\end{defn}

By definition, a the set-valued mapping $F:\mathcal{A}\leadsto\mathcal{B}$ is continuous at $a\in\mathcal{A}$ if and only if it is both upper and lower semicontinuous at $a$. 

\begin{thm}[Theorem 1.4.13, \cite{aubin2009set}]
\label{thm:aubin2009set_generic_cont}
Let $F$ be a set-valued map from a complete metric space $\mathcal{A}$ to a complete separable metric space $\mathcal{B}$. If $F$ is upper semicontinuous, then it is continuous on a residual set of $\mathcal{A}$.
\end{thm}

The proof of Theorem~\ref{thm:setvalued_optimalprivacymechanism} relies on the following lemma.

\begin{lem}
\label{lem:set-valued_F_continuous}
Assume that conditions (\textbf{C.}\ref{cond:cont}--\ref{cond:comp}) hold true for a given closed set $\mathcal{Q}\subseteq\mathcal{P}$ and a given $N\in\mathbb{N}$. For any $\epsilon\in\mathbb{R}$, the set-valued mapping $\mathcal{W}_N^*(\cdot;\epsilon)$ is upper semicontinuous over $\{Q\in\mathcal{Q}:\epsilon_{\min}(Q) < \epsilon\}$.
\end{lem}

\begin{proof}
For ease of notation, let $F:\mathcal{Q}\leadsto\mathcal{W}_N$ be the set-valued mapping defined by $F(Q) \defined \mathcal{W}_N^*(Q;\epsilon)$. Observe that $\text{Dom}(F) = \{Q\in\mathcal{Q} : \epsilon_{\min}(Q) \leq \epsilon\}$. In order to reach contradiction, assume that the set-valued mapping $F$ is not upper semicontinuous over $\{Q\in\mathcal{Q}:\epsilon_{\min}(Q) < \epsilon\}$. In this case, there exist a joint distribution $Q \in \mathcal{Q}$ with $\epsilon_{\min}(Q) < \epsilon$ and a neighborhood $\mathcal{N}$ of $F(Q)$ such that for every $r>0$ there exists $Q_r \in \mathcal{Q}_r(Q)$ with $F(Q_r)\not\subset\mathcal{N}$. For each $n\in\mathbb{N}$, let $W^*_n$ be such that $W^*_n\in F(Q_{1/n})\cap\mathcal{N}^c$, where $\mathcal{N}^c$ denotes the complement of $\mathcal{N}$ w.r.t.~$\mathcal{W}_N$. Since $\mathcal{N}^c$ is closed and $\mathcal{W}_N$ is compact, we have that $\mathcal{N}^c$ is also compact. Hence, there exists a subsequence $(W^*_{n_k})_{k=1}^\infty$ converging to some $W\in\mathcal{N}^c$. Since $F(Q)\subset\mathcal{N}$, we have that $W\not\in F(Q)$. Recall that $\lim_k Q_{1/n_k} = Q$ by construction, and thus Corollary~\ref{Corollary:OptimaltoOptimal} implies that $W\in \mathcal{W}_N^*(Q;\epsilon)$. Since this contradicts the fact that $W\not\in \mathcal{W}_N^*(Q;\epsilon)$, we conclude that $F$ is necessarily upper semicontinuous over $\{Q\in\mathcal{Q}:\epsilon_{\min}(Q) < \epsilon\}$.
\end{proof}

Now we are in position to prove Theorem~\ref{thm:setvalued_optimalprivacymechanism}.

\begin{proof}
For ease of notation, let
\begin{equation}
\mathcal{Q}^{(\delta)} \defined \{Q\in\mathcal{Q}:\epsilon_{\min}(Q) + \delta \leq \epsilon\}.
\end{equation}
Observe that $\mathcal{Q}^{(\delta)}\subset\{Q\in\mathcal{Q}:\epsilon_{\min}(Q) < \epsilon\}$. Thus, Lemma~\ref{lem:set-valued_F_continuous} implies that the mapping $\mathcal{W}_N^*(\cdot;\epsilon)$ is upper semicontinuous over $\mathcal{Q}^{(\delta)}$. Observe that both $\mathcal{Q}^{(\delta)}$ and $\mathcal{W}_N$ are compact metric spaces. Therefore, Theorem~\ref{thm:aubin2009set_generic_cont} establishes the lower semicontinuity of the mapping $\mathcal{W}_N^*(\cdot;\epsilon)$ on a residual set of $\mathcal{Q}^{(\delta)}$. Theorem~\ref{thm:setvalued_optimalprivacymechanism} follows immediately from this fact and Definition~\ref{defn:aubin2009lower_semicontinuous}.
\end{proof}

\section{Proofs from Section~\ref{Section:UniformRobustness}}
\label{Appendix:ProofsV}
\subsection{Lemma~\ref{lem:uni_pri_mech_well_defined}}
\label{Appendix:Proofuni_pri_mech_well_define}

\begin{proof}
\ 

\textnormal{(i)} Let $W\in\mathcal{W}_N$ be fixed. By condition (\textbf{C.}\ref{cond:lipcont}), the mapping $\mathcal{U}(\cdot,W)$ is (Lipschitz) continuous over $\mathcal{Q}_r(\hat{P})$. Since $\mathcal{Q}_r(\hat{P})$ is compact, the mapping $\mathcal{U}(\cdot,W)$ attains its infimum over $\mathcal{Q}_r(\hat{P})$.

\textnormal{(ii)} Assume that $\mathcal{D}_{\mathcal{Q},N}(\hat{P};\epsilon,r)$ is non-empty. By \textnormal{(i)}, we have that
\begin{equation}
    \mathcal{U}_r(\hat{P},W) = \min_{Q\in\mathcal{Q}_r(\hat{P})} \mathcal{U}(Q,W)
\end{equation}
is well defined. As established in Lemma~\ref{lem:cont_UandL}, the mapping $\mathcal{U}(\cdot,\cdot)$ is continuous over $\mathcal{Q}\times\mathcal{W}_N$ and in particular over $\mathcal{Q}_r(\hat{P})\times\mathcal{W}_N$. Hence, Lemma~\ref{lem:inf_cont_overcomp_cont} implies that $\mathcal{U}_r(\hat{P},\cdot)$ is continuous over $\mathcal{W}_N$.

Note that, by definition,
\begin{equation}
\label{eq:DUniformIntersection}
    \mathcal{D}_{\mathcal{Q},N}(\hat{P};\epsilon,r) = \bigcap_{Q\in\mathcal{Q}_r(\hat{P})} \mathcal{D}_N(Q;\epsilon),
\end{equation}
where $\mathcal{D}_N(Q;\epsilon) = \{W\in\mathcal{W}_N : \mathcal{L}(Q,W) \leq \epsilon\}$. Condition (\textbf{C.}\ref{cond:cont}) implies that, for every $Q\in\mathcal{Q}$, the mapping $\mathcal{L}(Q,\cdot)$ is continuous over $\mathcal{W}_N$. Thus, we have that $\mathcal{D}_N(Q;\epsilon)$ is compact as $\mathcal{W}_N$ is compact. Since the arbitrary intersection of compact sets is also compact, \eqref{eq:DUniformIntersection} readily shows that $\mathcal{D}_{\mathcal{Q},N}(\hat{P};\epsilon,r)$ is compact. Therefore, the (continuous) mapping $\mathcal{U}_r(\hat{P},\cdot)$ attains its supremum over $\mathcal{D}_{\mathcal{Q},N}(\hat{P};\epsilon,r)$, as required.
\end{proof}
\subsection{Theorem~\ref{thm:desn_unfpri_mech}}
\label{Appendix:Proof_lem:desn_unfpri_mech}

Consider the following lemma.
\begin{lem}
\label{lem:append_bound_for_HP}
Assume that conditions (\textbf{C.}\ref{cond:cont}--\ref{cond:comp}) hold true for a given closed set $\mathcal{Q}\subseteq\mathcal{P}$ and a given $N\in\mathbb{N}$. Let $\epsilon\in\Reals$ and $r\geq 0$ be given. If $P \in \mathcal{Q}_r(\hat{P})$ and $\epsilon-C_Lr\geq \epsilon_{\min}(\hat{P})$, then
\begin{equation}
\label{equa:bound_M_1}
\mathsf{H}(P;\epsilon) \leq \mathsf{H}(\hat{P};\epsilon+C_L r) + C_U r.
\end{equation}
\end{lem}
\begin{proof}
We start proving that $\epsilon \geq \epsilon_{\min}(P)$. Remark~\ref{rem:sup_max_Wn} implies that there exists $W'\in\mathcal{W}_N$ such that $\mathcal{L}(\hat{P},W') = \epsilon_{\min}(\hat{P})$. By the Lipschitz continuity given in condition (\textbf{C.}\ref{cond:lipcont}), we have that
\begin{align}
\mathcal{L}(P,W') 
&\leq \mathcal{L}(\hat{P},W')+|\mathcal{L}(\hat{P},W')-\mathcal{L}(P,W')| \\
&\leq \epsilon_{\min}(\hat{P})+C_L \|\hat{P}-P\|_1 \\
&\leq \epsilon_{\min}(\hat{P})+C_L r \leq \epsilon,
\end{align}
where the last inequality follows from the assumption $\epsilon-C_Lr\geq \epsilon_{\min}(\hat{P})$. By the minimality of $\epsilon_{\min}(P)$, we conclude that $\epsilon_{\min}(P)\leq \epsilon$. Then, Remark~\ref{rem:sup_max_Wn} implies that there exists $W\in\mathcal{W}_N$ such that $\mathcal{L}(P,W) \leq \epsilon$ and 
\begin{equation}
\mathsf{H}(P;\epsilon) = \mathcal{U}(P,W).
\end{equation}
By the Lipschitz continuity given in condition (\textbf{C.}\ref{cond:lipcont}), we have that
\begin{equation}
|\mathcal{U}(\hat{P},W)-\mathcal{U}(P,W)| \leq C_U \|\hat{P}-P\|_1 \leq C_U r.
\end{equation}
In particular,
\begin{equation}
\label{equa:bound_M}
\mathsf{H}(P;\epsilon) \leq \mathcal{U}(\hat{P},W) + C_U r.
\end{equation}
Similarly, since
\begin{equation}
|\mathcal{L}(\hat{P},W)-\mathcal{L}(P,W)| \leq C_L\|\hat{P}-P\|_1 \leq C_L r,
\end{equation}
we have $\mathcal{L}(\hat{P},W) \leq \mathcal{L}(P,W)+ C_L r \leq \epsilon + C_L r$. Therefore, from inequality~\eqref{equa:bound_M} and the maximality of the privacy-utility function, we conclude that 
\begin{equation}
\mathsf{H}(P;\epsilon) \leq \mathsf{H}(\hat{P};\epsilon + C_L r) + C_U r,
\end{equation}
as we wanted to prove.
\end{proof}

Now we are in position to prove Theorem~\ref{thm:desn_unfpri_mech}.

\begin{proof}
For any $W\in\mathcal{D}_N(\hat{P};\epsilon-C_Lr)$, we have $\mathcal{L}(\hat{P},W) + C_Lr \leq \epsilon$. By the Lipschitz continuity given in condition (\textbf{C.}\ref{cond:lipcont}), for every $Q\in\mathcal{Q}_r(\hat{P})$,
\begin{equation}
|\mathcal{L}(\hat{P},W)-\mathcal{L}(Q,W)|\leq C_L\|\hat{P}-Q\|_1 \leq C_Lr.
\end{equation}
Hence, $\mathcal{L}(Q,W) \leq \epsilon$ and $W\in\mathcal{D}_N(Q;\epsilon)$ for every $Q\in\mathcal{Q}_r(\hat{P})$. By \eqref{eq:UniformPrivacyMechanismsIntersection}, it follows that $W\in \mathcal{D}_{\mathcal{Q},N}(\hat{P};\epsilon,r)$. 

Let $W^*\in\mathcal{W}_N^*(\hat{P};\epsilon-C_Lr)$ and $W^{\dagger}\in\mathcal{W}^{\dagger}_{\mathcal{Q},N}(\hat{P};\epsilon,r)$. By definition,
\begin{align}
\label{eq:proof_desn_U_eq_H}
\mathcal{U}(\hat{P},W^*)=\mathsf{H}(\hat{P};\epsilon-C_Lr).
\end{align}
By the Lipschitz continuity given in condition (\textbf{C.}\ref{cond:lipcont}), we have that
\begin{equation}
\label{eq:proof_desn_lip_U}
|\mathcal{U}(\hat{P},W^*)-\mathcal{U}(P,W^*)| 
\leq C_U \|\hat{P}-P\|_1 \leq C_U r.
\end{equation}
Combining \eqref{eq:proof_desn_U_eq_H} and \eqref{eq:proof_desn_lip_U} together, we have
\begin{align}
\label{eq::proof_UP_Uhat_bound}
\mathcal{U}(P,W^*)
\geq \mathsf{H}(\hat{P};\epsilon-C_Lr)-C_Ur.
\end{align}
Since $W^{\dagger}\in\mathcal{W}^{\dagger}_{\mathcal{Q},N}(\hat{P};\epsilon,r)\subset\mathcal{D}_{\mathcal{Q},N}(\hat{P};\epsilon,r)$ and $P\in\mathcal{Q}_r(\hat{P})$, \eqref{eq:UniformPrivacyMechanismsIntersection} implies that $W^{\dagger}\in\mathcal{D}_N(P;\epsilon)$. Thus, by the maximality of $\mathsf{H}(P;\epsilon)$, we have
\begin{equation}
\label{eq:proof_Maximality_H_Uniform}
\mathcal{U}(P,W^{\dagger}) \leq \mathsf{H}(P;\epsilon).
\end{equation}
An immediate application of Lemma~\ref{lem:append_bound_for_HP} shows that
\begin{align}
\label{eq::proof_UP_dag_H_bound}
\mathcal{U}(P,W^{\dagger}) \leq \mathsf{H}(P;\epsilon) \leq \mathsf{H}(\hat{P};\epsilon+C_L r) + C_U r.
\end{align}
Therefore, combining \eqref{eq::proof_UP_Uhat_bound} and \eqref{eq::proof_UP_dag_H_bound} together,
\begin{align}
\mathcal{U}(P,W^*) - \mathcal{U}(P,W^{\dagger})
&\geq -\left(\mathsf{H}(\hat{P};\epsilon+C_L r)-\mathsf{H}(\hat{P};\epsilon-C_Lr)+2C_U r\right),
\end{align}
which implied the desired conclusion.
\end{proof}

\subsection{Theorem~\ref{thm:rob_mech_converge_opt_mech}}
\label{Appendix:thm:rob_mech_to_mech_prob}

\begin{proof}
Recall that, by Lemma~\ref{lem:epsilonmin}, the mapping $\epsilon_{\min}(\cdot)$ is continuous over $\mathcal{Q}$. In particular, we have that $\lim_n \epsilon_{\min}(P_n) = \epsilon_{\min}(P)$ whenever $\lim_n P_n = P$. Since $\epsilon_{\min}(P) < \epsilon$, $\lim_n P_n=P$, and $\lim_n r_n=0$, we have that $\epsilon_{\min}(P_n)+C_Lr_n < \epsilon$ for $n$ large enough. Therefore, for $n$ large enough, the first part of Theorem~\ref{thm:desn_unfpri_mech} establishes that the set $\mathcal{D}_{\mathcal{Q},N}(P_n;\epsilon,r_n)$ is non-empty. Furthermore, part~\textnormal{(ii)} of Lemma~\ref{lem:uni_pri_mech_well_defined} implies that $\mathcal{W}^{\dagger}_{\mathcal{Q},N}(P_n;\epsilon,r_n)$ is non-empty as well. Without loss of generality, we assume that $\mathcal{W}^{\dagger}_{\mathcal{Q},N}(P_n;\epsilon,r_n) \neq \emptyset$ for all $n\geq1$.

In order to reach contradiction, we assume that the conclusion does not hold true, i.e., that there exists a sequence $(W_n^\dagger)_{n=1}^\infty$ such that $W^{\dagger}_n\in \mathcal{W}^{\dagger}_{\mathcal{Q},N}(P_n;\epsilon,r_n)$ for all $n\geq1$ and
\begin{equation}
    \limsup_{n\to\infty} \mathsf{dist}(W^{\dagger}_n,\mathcal{W}_N^*(P;\epsilon)) > 0.
\end{equation}
Since both $\mathcal{W}_N$ and $\mathcal{W}_N^*(P;\epsilon)$ are compact (Lemma~\ref{lem:W_star_compact}), there exists a subsequence of $(W^{\dagger}_n)_{n=1}^\infty$ converging to some $W_0\in\mathcal{W}_N$ with $W_0\not\in \mathcal{W}_N^*(P;\epsilon)$. Without loss of generality, we assume that $\lim_n W^{\dagger}_n = W_0$. By the continuity of $\mathcal{L}(\cdot,\cdot)$ established in Lemma~\ref{lem:cont_UandL},
\begin{align}
\label{eq:PrivacyLimitDaggers}
    \mathcal{L}(P,W_0) 
    = \lim_{n\to\infty} \mathcal{L}(P_n,W^{\dagger}_n)
    \leq \epsilon.
\end{align}
The last inequality implies that $W_0\in\mathcal{D}_N(P;\epsilon)$ and hence $\mathcal{U}(P,W_0)\leq \mathsf{H}(P;\epsilon)$. Assume that we have already proved that $\mathcal{U}(P,W_0) \geq \mathsf{H}(P;\epsilon)$. In particular, we would have that $\mathcal{U}(P,W_0) = \mathsf{H}(P;\epsilon)$ and thus $W_0\in \mathcal{W}_N^*(P;\epsilon)$. Since this conclusion contradicts the fact that $W_0\not\in \mathcal{W}_N^*(P;\epsilon)$, the result would follow. Now we focus on proving that $\mathcal{U}(P,W_0)\geq \mathsf{H}(P;\epsilon)$.

Since $\epsilon > \epsilon_{\min}(P)$ and $\lim_n r_n = 0$, we have that $\epsilon-2C_Lr_n > \epsilon_{\min}(P)$ for $n$ large enough. Without loss of generality, we assume that $\epsilon-2C_Lr_n > \epsilon_{\min}(P)$ for all $n\geq1$. For a given $n\geq1$, let $W^*_0\in\mathcal{W}_N$ be such that $\mathcal{L}(P,W^*_0)\leq \epsilon-2C_Lr_n$ and $\mathcal{U}(P,W^*_0)=\mathsf{H}(P;\epsilon-2C_Lr_n)$. By condition (\textbf{C.}\ref{cond:lipcont}), for any $Q\in\mathcal{Q}_{r_n}(P_n)$,
\begin{equation}
\label{eq:ConvergenceUPMLipschitzP}
    |\mathcal{L}(P,W^*_0)-\mathcal{L}(Q,W^*_0)| \leq C_L\|P-Q\|_1 \leq 2C_Lr_n,
\end{equation}
where the last inequality follows from the triangle inequality and the fact that $P,Q\in\mathcal{Q}_{r_n}(P_n)$. Thus,
\begin{equation}
    \mathcal{L}(Q,W^*_0)\leq \mathcal{L}(P,W^*_0)+2C_Lr_n\leq \epsilon,
\end{equation}
for any $Q\in\mathcal{Q}_{r_n}(P_n)$. Hence, $W^*_0\in\mathcal{D}_{\mathcal{Q},N}(P_n;\epsilon,r_n)$ and, by the definition of $\mathcal{W}_{\mathcal{Q},N}^\dagger(P_n;\epsilon,r_n)$,
\begin{align}
\label{eq:proof_thm6_min_Uleq_U1}
    \min_{Q\in\mathcal{Q}_{r_n}(P_n)} \mathcal{U}(Q,W^*_0)
    \leq \min_{Q\in\mathcal{Q}_{r_n}(P_n)} \mathcal{U}(Q,W_n^{\dagger})
    \leq \mathcal{U}(P,W_n^{\dagger}).
\end{align}
As in \eqref{eq:ConvergenceUPMLipschitzP}, we have that
\begin{equation}
    |\mathcal{U}(P,W^*_0)-\mathcal{U}(Q,W^*_0)| \leq 2C_Ur_n.
\end{equation}
Therefore, for any $Q\in\mathcal{Q}_{r_n}(P_n)$,
\begin{align}
    \mathcal{U}(Q,W^*_0)
    &\geq \mathcal{U}(P,W^*_0)-2C_Ur_n\\
    &=\mathsf{H}(P;\epsilon-2C_Lr_n)-2C_Ur_n.
\end{align}
In particular, this implies that
\begin{align}
\label{eq:proof_thm6_min_Ugeq_H2}
    \min_{Q\in\mathcal{Q}_{r_n}(P_n)} \mathcal{U}(Q,W^*_0)
    \geq \mathsf{H}(P;\epsilon-2C_Lr_n)-2C_Ur_n.
\end{align}
Combining \eqref{eq:proof_thm6_min_Uleq_U1} and \eqref{eq:proof_thm6_min_Ugeq_H2} together, we have
\begin{align}
    \mathcal{U}(P,W_n^{\dagger}) &\geq 
    \min_{Q\in\mathcal{Q}_{r_n}(P_n)} \mathcal{U}(Q,W^*_0)\\
    &\geq \mathsf{H}(P;\epsilon-2C_Lr_n)-2C_Ur_n.
\end{align}
Therefore, by the continuity of $\mathcal{U}(P,\cdot)$ and $\mathsf{H}(P;\cdot)$ assumed in condition (\textbf{C.}\ref{cond:cont}),
\begin{align}
    \mathcal{U}(P,W_0) &=\lim_{n\to\infty}\mathcal{U}(P,W^{\dagger}_n)\\
    &\geq \lim_{n\to\infty} \mathsf{H}(P;\epsilon-2C_Lr_n)-2C_Ur_n\\
    &= \mathsf{H}(P;\epsilon).
\end{align}
This concludes the proof of the first part of the theorem.

Now we prove the second part of the theorem. Recall that in this case, for all $n\geq1$, $P_n=\hat{P}_n$ and $r_n \geq (2p\log(n)/n)^{1/2}$ for some $p>1$. Thus,
\begin{align}
    \sum_{n=1}^{\infty} \Pr(\|\hat{P}_n-P\|_1\geq r_n) 
    &\leq \sum_{n=1}^{\infty} \Pr\left(\|\hat{P}_n-P\|_1\geq \sqrt{\frac{2p\log(n)}{n}}\right) \\
    \label{eq:SecondInqBorelCantelli} &\leq \exp(|\mathcal{S}|\cdot|\mathcal{X}|) \sum_{n=1}^{\infty}  \frac{1}{n^p},
\end{align}
where \eqref{eq:SecondInqBorelCantelli} follows directly from \eqref{equa:L1bound_weissman}. Since $p>1$ by assumption, we have that $\sum_n 1/n^p$ is finite. Thus, by a routine application of the Borel-Cantelli Lemma, see, e.g., \cite[Sec.~2.18]{gut2013probability},
\begin{align}
    \Pr\left(\left\{\omega:\|\hat{P}_n(\omega)-P\|_1 < r_n\text{ for all }n \geq N = N(\omega)\right\}\right)=1.
\end{align}
Observe that, by the last equality, the hypotheses of the first part of this theorem are satisfied almost surely. A direct application of the former part of this theorem leads to
\begin{align}
    \Pr\left(\lim_{n\to\infty}\mathsf{dist}(W^{\dagger}_n,\mathcal{W}_N^*(P;\epsilon)) = 0\right)=1,
\end{align}
as we wanted to prove.
\end{proof}

\bibliographystyle{IEEEtran}
\bibliography{references}
\end{document}